\documentclass[12pt]{article}

\usepackage[margin=1in]{geometry}
\usepackage{amsmath,amssymb,amsthm,mathtools}
\usepackage{mathrsfs}
\usepackage{bbm}
\usepackage{enumitem}
\usepackage{xcolor}
\usepackage{comment}
\usepackage{textcomp, gensymb}
\usepackage{gensymb}
\usepackage{enumitem}
\usepackage{graphicx}
\usepackage{scrextend}
\usepackage{blindtext}
\usepackage{caption}
\usepackage{pifont}
\usepackage{url}
\usepackage{authblk}
\usepackage{subcaption}
\usepackage{circuitikz}
\usepackage{listings}
\usepackage{color}
 \usepackage{microtype}
\usepackage{graphicx}
\usepackage{multirow}
\usepackage{hyperref}
\usepackage[nameinlink,capitalize]{cleveref}
\usepackage{dsfont}
\usepackage{natbib}
\usepackage[skip=4pt plus 1pt]{parskip}
\usepackage{setspace}
\usepackage{booktabs}
\usepackage{float}
\usepackage{aliascnt}
\usepackage{color}
\allowdisplaybreaks
\DeclareMathOperator{\Markov}{Markov}
\hypersetup{
  colorlinks=true,
  linkcolor=blue!60!black,
  citecolor=blue!60!black,
  urlcolor=blue!60!black
}

\newtheoremstyle{spaced}
  {12pt}
  {12pt}
  {\itshape}
  {}
  {\bfseries}
  {.}
  {.5em}
  {}

\theoremstyle{spaced}
\newtheorem{theorem}{Theorem}[section]

\newaliascnt{proposition}{theorem}
\newtheorem{proposition}[proposition]{Proposition}
\aliascntresetthe{proposition}

\newaliascnt{lemma}{theorem}
\newtheorem{lemma}[lemma]{Lemma}
\aliascntresetthe{lemma}

\newaliascnt{corollary}{theorem}
\newtheorem{corollary}[corollary]{Corollary}
\aliascntresetthe{corollary}

\newaliascnt{assumption}{theorem}

\aliascntresetthe{assumption}

\newaliascnt{definition}{theorem}
\newtheorem{definition}[definition]{Definition}
\aliascntresetthe{definition}

\theoremstyle{remark}

\newaliascnt{remark}{theorem}
\newtheorem{remark}[remark]{Remark}
\aliascntresetthe{remark}

\usepackage[nameinlink,capitalize]{cleveref}

\crefname{assumption}{Assumption}{Assumptions}
\Crefname{assumption}{Assumption}{Assumptions}

\crefname{equation}{equation}{equations}
\Crefname{equation}{Equation}{Equations}

\crefname{lemma}{Lemma}{Lemmas}
\Crefname{lemma}{Lemma}{Lemmas}

\crefname{theorem}{Theorem}{Theorems}
\Crefname{theorem}{Theorem}{Theorems}

\crefname{proposition}{Proposition}{Propositions}
\Crefname{proposition}{Proposition}{Propositions}

\crefname{corollary}{Corollary}{Corollaries}
\Crefname{corollary}{Corollary}{Corollaries}

\crefname{definition}{Definition}{Definitions}
\Crefname{definition}{Definition}{Definitions}

\crefname{remark}{Remark}{Remarks}
\Crefname{remark}{Remark}{Remarks}

\newcommand{\E}{\mathbb E}
\newcommand{\N}{\mathbb N}
\newcommand{\R}{\mathbb R}
\newcommand{\calP}{\mathcal P}
\newcommand{\calX}{\mathcal X}
\newcommand{\calF}{\mathcal F}
\newcommand{\bbP}{\mathbb P}
\newcommand{\ind}{\mathbbm 1}
\newcommand{\KL}{\mathrm{D}_{\text{KL}}}
\newcommand{\kl}{\mathrm{kl}}

\newcommand{\eps}{\varepsilon}
\newcommand{\Pinf}[1]{\mathbb P_{#1}^{\infty}}
\newcommand{\Einf}[1]{\mathbb E_{#1}^{\infty}}
\newcommand{\Pchg}[3]{\mathbb P_{#1,#2,#3}}
\newcommand{\defeq}{\mathrel{\mathop:}=}
\newcommand{\CADD}{\mathrm{CADD}}

\title{Non-partitioned e-detectors for nonparametric sequential change detection}
\author[1]{Aytijhya Saha}
\author[2]{Aaditya Ramdas}

\affil[1]{Massachusetts Institute of Technology.  \texttt{aytijhya@mit.edu}}
\affil[2]{Carnegie Mellon University. \texttt{aramdas@cmu.edu}}
\date{\today}

\begin{document}
\maketitle

\begin{abstract}
We study the problem of sequential change detection over a general class of probability distributions ($\mathcal P$), where both the pre-change and post-change distributions are unknown and belong to $\mathcal P$. We do not assume a pre-specified partition of $\mathcal P$ into pre- and post-change families. We propose a general class of sequential change detectors obtained by aggregating point-null e-processes over possible changepoints and taking an infimum over candidate no-change distributions. The weights in the aggregation scheme determine whether they attain average run length (ARL) control and probability-of-false-alarm (PFA) control. Under suitable assumptions, we prove that our methods achieve first-order asymptotically optimal detection delay. Concrete examples include sub-Gaussian and bounded mean changes, Gaussian mean changes with unknown variance, as well as changes in
Markov transition matrices. 
\end{abstract}

\section{Introduction}
We consider the following problem of sequential change detection (SCD): for a Polish space $\mathcal{X}$, let $(\calX,\mathcal B)$ be its Borel space. Let $\calP$ be a composite class of probability laws on $\calX$. We are interested in detecting change from one unknown (pre-change) distribution $P\in\calP$ to another (post-change) distribution $Q\in\calP$. That is, under the alternative, for some unknown time $T\in\mathbb N$ and $P,Q\in\calP$,  where $P \neq Q$,
\[
  X_1,\ldots,X_T\stackrel{iid}{\sim} P,
  \qquad
  X_{T+1},X_{T+2},\ldots\stackrel{iid}{\sim} Q.
\]
The no-change null is the composite null
\[
  H_0:\quad X_1,X_2,\ldots\stackrel{iid}{\sim} R
  \quad\text{for some unknown }R\in\calP.
\]

The \underline{non-partitioned SCD problem} stated above does not require any pre-specified partitioning of $\calP$ into
pre- and post-change distribution classes. 

In contrast, the classical \underline{partitioned SCD problem} assumes that \[ P\in\calP_0, \qquad Q\in\calP_1, \qquad \calP_0\cap\calP_1=\varnothing, \] where $\calP_0\subset \calP$ and $\calP_1\subset \calP$ are specified in advance. This partition implicitly encodes prior knowledge about the nature or direction of the change and is crucial for the design of many classical procedures \citep{shiryaev1963optimum,roberts1966comparison,siegmund1995using}, including the recently emerged nonparametric change detection framework, e-detectors \citep{shin2022detectors}. 

Despite its practical importance, the non-partitioned SCD problem is relatively underexplored in the literature. To illustrate the distinction, consider the sub-Gaussian mean change problem, but do not know whether the mean increases or decreases. The partitioned approaches are not applicable for these problems.
Our goal is to construct sequential composite e-detectors that are applicable to broad classes of distributions $\calP$, while achieving theoretically optimal detection delays.

For simplicity, we focus on i.i.d.\ observations in the main text. However, the framework itself does not require independence: in \Cref{sec:markov-example} in the appendix, we show the construction for dependent observations and develop explicit ARL- and PFA-valid e-detectors for changes in the transition matrix of a two-state Markov chain.

The main idea of this paper is simple: we reduce the change-detection problem to a family of point-null e-processes. For every candidate no-change distribution $R\in\calP$, suppose we have an e-process that tests
\[ H_{0,R}:\qquad X_1,X_2,\ldots \stackrel{iid}{\sim} R,\]
against its complement
\[ H_{0,R^c}:\qquad X_1,X_2,\ldots \stackrel{iid}{\sim} S, \text{ for some } S\in \calP\setminus \{R\}.\]
See \Cref{subsec:background} for a formal definition of an e-process.
We then aggregate these point-null e-processes over all possible starting times and take an infimum over $R\in\calP$, obtaining a Shiryaev--Roberts-type composite e-detector. 

Classical information-theoretic lower bounds for quickest change detection, including \cite{lai1998information}, concern Pollak’s or Lorden’s worst-case delay under a positive asymptotic log-likelihood information rate. Composite extensions likewise require the post-change law to be positively separated in KL divergence from the no-change class. These results do not directly apply to our fully non-partitioned setting, because the post-change law
$Q$ itself is an admissible no-change law, so that $\inf_{R\in\calP}\KL(Q\|R)=0$.
Moreover, classical worst-case CADD degenerates because an immediate change to $Q$ is indistinguishable from the no-change law $Q^\infty$. To our knowledge, information-theoretic lower bounds for fixed-changepoint delay in this fully overlapping setting have not previously been developed.


\subsection{Performance metrics}
\label{subsec:metrics}
A sequential change detector is a stopping time $\tau$ with respect to the natural filtration
\[
\mathcal F_t=\sigma(X_1,\ldots,X_t).
\]
We start by introducing a few notations.
 For $R\in\calP$, write
\[
  \Pinf{R}=R^{\infty}
\]
for the no-change law under which $X_1,X_2,\ldots$ are i.i.d. $R$.  For $P,Q\in\calP$ and an integer $T\ge 1$, write
\[
  \mathbb P_{P,Q}^{T}\defeq
  P^T\otimes Q^{\infty}
\]
for the law under which
\[
  X_1,\ldots,X_T\sim P,
  \qquad
  X_{T+1},X_{T+2},\ldots\sim Q,
\]
independently across time.  

The performance of a detector is characterized by two competing objectives: controlling false alarms under the no-change regime and minimizing the detection delay after a change occurs.

\paragraph{Average Run Length (ARL).}
The average run length to false alarm (ARL2FA or ARL in short) is defined as
\[
\mathrm{ARL}(\tau)
:=
\inf_{R\in\calP}
\mathbb E_R^\infty[\tau].
\]
A detector is said to satisfy ARL level $A$ if
$\mathrm{ARL}(\tau)\ge A.$
Thus, larger values of $A$ correspond to fewer false alarms.

In contrast, another strand of the literature emphasizes controlling the false alarm rate instead of ARL, often at the cost of a longer detection delay.

\paragraph{Probability of False Alarm (PFA).} This is defined as
\[
\mathrm{PFA}(\tau)
:=
\sup_{R\in\calP}
\mathbb P_R^\infty(\tau<\infty).
\]
For a target level $\alpha\in(0,1)$, a detector is said to be PFA-valid at level $\alpha$ if
$\mathrm{PFA}(\tau)\le \alpha.$

We propose methods under both ARL and PFA metrics. To assess performance under the changepoint alternative, it is standard to consider the following measure of detection delay.

\paragraph{Conditional Average Detection Delay (CADD).}
Suppose the true data-generating distribution is
\[
X_1,\ldots,X_T\stackrel{iid}{\sim}P,
\qquad
X_{T+1},X_{T+2},\ldots
\stackrel{iid}{\sim}Q,
\]
for some $P,Q\in\calP$ with $P\neq Q$ and changepoint $T\in\mathbb N$.
Then, the conditional average detection delay (CADD) is defined as
\[
\mathrm{CADD}_{P,Q}^{T}(\tau)
:=
\mathbb E_{P,Q}^{T}
\left[
\tau-T
\,\middle|\,
\tau>T
\right].
\]
This quantity measures the expected delay after the changepoint, conditioned on not having stopped prematurely.
Note that it is different from Pollak’s conditional average delay to detection, which considers the supremum of the above quantity over all $T\in\N$.
\paragraph{High Probability Detection Delay Bound.}
For detectors controlling PFA, CADD is not the ``right" metric because if
$\sup_{R\in\mathcal{P}}\bbP_R^\infty(\tau<\infty)\le \alpha<1,$
then $\bbP_Q^\infty(\tau=\infty)\geq 1-\alpha>0,$ since $Q\in\mathcal{P}$.
Let us assume that $P$ and $Q$ are mutually absolutely continuous. Then $\bbP^T_{P,Q}$ and $\bbP^\infty_{Q}$ are also mutually absolutely continuous. Thus, $\bbP^T_{P,Q}(\tau=\infty)>0$, for every fixed finite $T$.  Hence, the CADD, $
\mathbb E_{P,Q}^{T}
\left[
\tau-T
\,\middle|\,
\tau>T
\right],$
is infinite.  This is unavoidable for any globally PFA-valid procedure.
Thus, under PFA control, the natural asymptotic delay statement is a high-probability delay statement, as in \Cref{thm:subg-upper-pfa}.

The goal in SCD is to construct detectors with small detection delay while satisfying prescribed false-alarm constraints.

\subsection{Background and related works}
\label{subsec:background}
\paragraph{E-processes.}
E-processes are a basic tool for safe anytime-valid inference
\citep{ramdas2022game,ramdas2024hypothesis}.  A nonnegative adapted process
$(M_t)_{t\ge0}$ is an e-process for a class $\mathcal Q$ if, for every stopping
time $\tau$,
\[
  \sup_{R\in\mathcal Q}\E_R[M_\tau]\le1.
\]
For composite testing problems, growth-rate criteria range from optimizing
against a fixed alternative (GRO), through minimax growth (GROW), to
pointwise relative growth optimality (REGROW) \citep{grunwald2024safe}.
\citet{ram2026power} prove that every weakly compact null class of i.i.d. laws
on a Polish space admits a power-one REGROW e-process against its complement.
In particular, singleton nulls supply the point-null primitives used here.

\paragraph{E-detectors.}
\citet{shin2022detectors} introduced e-detectors as the change-detection
counterpart of e-processes.  A nonnegative adapted process $(E_t)$ is an
e-detector under a no-change law $R$ if
\[
  \E_R[E_\tau]\le \E_R[\tau]
\]
for every stopping time $\tau$.  Thresholding at $A$ then gives both
$\E_R\tau_A\ge A$ and the local bound
$\bbP_R(\tau_A\le m)\le m/A$.  Restarting an e-process at each time and
summing the restarted processes produces an e-Shiryaev--Roberts detector;
max/reset constructions give e-CUSUM analogues.  Mixtures, baseline
increments, and pruning can trade statistical power against computational
cost while preserving finite-sample ARL guarantees
\citep{shin2022detectors}.  Recent developments include sharp bounded-mean
change-detection theory \citep{ram2026boundedmeans}, post-detection
localization \citep{saha2025post}, and applications to universal quantum
measurements \citep{zecchin2026quantum}.  Our contribution differs from the
standard e-detector setup in one structural respect: the pre- and post-change
families are not disjoint, so we aggregate point-null e-processes and then
minimize over the candidate common law.

\paragraph{Separated SCD.}
The classical literature begins with CUSUM and Shiryaev--Roberts procedures
\citep{page1954continuous,shiryaev1963optimum,roberts1966comparison} and the
Lorden--Pollak minimax formulations \citep{lorden1971procedures,pollak1985optimal}.
Composite separated problems are commonly handled by generalized likelihood
ratios, mixtures, or robust least-favorable models
\citep{siegmund1995using,lai1998information,tartakovsky2014sequential}.
Modern nonparametric and computationally efficient variants include
contrastive discrepancies \citep{puchkin2023contrastive}, robust procedures
for heavy-tailed means \citep{sankararaman2023heavy}, functional pruning for
unknown mean shifts \citep{romano2023focus}, and computational-geometry
methods for multivariate Gaussian changes \citep{pishchagina2026geometry}.
These works typically retain some separation, training period, parametric
structure, or prescribed discrepancy between the pre- and post-change
regimes.

\paragraph{Non-partitioned and fully unknown SCD.}
Several strands come particularly close to the present formulation.
For sub-Gaussian piecewise-constant means, \citet{maillard2019sequential} develop a doubly time-uniform scan/GLR detector when both segment means are unknown and give nonasymptotic delay guarantees.
\citet{alami2020restarted} analyze a restarted Bayesian online detector when
both segment distributions are unknown and obtain nonasymptotic delay and
false-alarm guarantees under their modeling assumptions.  In finite-alphabet
settings, \citet{malik2021universal} estimate the pre-change distribution
empirically and use universal coding for the post-change law, proving
asymptotic optimality; \citet{gulaguli2025markov} extend this idea to finite-
order Markov data.  The change-point-model literature implemented in the
\texttt{cpm} package includes sequential Gaussian and fully nonparametric
procedures with unknown segment parameters or distributional form
\citep{ross2015cpm}.  Data-driven kernel methods have also been studied when
both Markov or hidden-Markov regimes are unknown \citep{zhang2022datadriven}.
For mean functionals, \citet{liang2022non} give a nonparametric quickest
mean-change method, while \citet{shekhar2023sequential,shekhar2023reducing}
reduce sequential change detection to confidence sequences and sequential
estimation.  Conformal test martingales and conformal CUSUM procedures test
the massive i.i.d./exchangeability null without specifying either segment
law \citep{vovk2003testing,vovk2021retrain,vovk2025cusum}.  These approaches
provide important distribution-free or model-adaptive guarantees, but they
do not yield the same general reduction from point-null KL-optimal
e-processes to instance-optimal non-partitioned ARL and PFA detectors over an
arbitrary class $\calP$.

\subsection{Our contribution}
Our main contributions are summarized as follows: 
 \begin{itemize} 
 \item We introduce e-detectors for the  non-partitioned sequential change detection problem 
and prove finite-sample false-alarm guarantees, including ARL control and PFA control, over arbitrary distribution classes $\calP$. 
\item Under suitable assumptions on $\calP$, we derive upper bounds on detection delay. We also derive lower bounds, matching those upper bounds and thus proving that our method achieves asymptotically optimal detection delay, under suitable assumptions.
\item We develop explicit detectors for sub-Gaussian means, bounded means,
 Gaussian means with unknown variance and change in a two-state Markov transition matrix.  
 \end{itemize}
The rest of the paper is organized as follows.
\Cref{sec:method} gives general ARL- and PFA-valid detectors.
\Cref{sec:sub-Gaussian,sec:bounded-mean,sec:gaussian-unknown-variance} give
three concrete examples.  \Cref{sec:upper-bdd,sec:lowerb} state the general
upper and lower bounds. \Cref{sec:experiments-updated} reports
simulation experimental results.\footnote{The code for reproducing all the experimental results in this paper is publicly available at \href{https://github.com/Aytijhya/Nonseparated-e-detector}{https://github.com/Aytijhya/Nonseparated-e-detector}.} \Cref{sec:markov-example} provides another concrete example with dependent data. Proofs of all the theoretical results are provided in the appendix.

\section{Method}
\label{sec:method}

The following is the key ingredient behind our construction.

\begin{definition}[$s$-delay e-processes]\label{ass:eprocess}
For every $R\in\calP$ and every start time $s\in\N$, the process $\{M_{s:t}^R\}_{t\in\N}$
with $M_{s:j}^R=1,$ for $j=1,\cdots,s-1$, is an $s$-delay e-process under $\Pinf{R}$ if it is nonnegative and adapted to the filtration $\{\calF_{t}\}_{t\in\N}$, and for every stopping time $\tau$, 
\[
  \Einf{R}\!\left[M_{s:\tau}^R\mid \calF_{s-1}\right]\le 1.
\]
\end{definition}

 \cite{ram2026power} prove a general existence theorem for sequential tests and e-processes for i.i.d. laws on Polish spaces: weak compactness of the null class is a sufficient condition for power-one tests against the complement, and their REGROW construction yields asymptotically relatively growth-rate optimal e-processes. 
 Since a singleton $\{R\}$ is weakly compact, their result supplies point-null primitives $M^R$ at the level of existence.  For change detection, however, the statistic contains an infimum over $R\in\calP$, so pointwise growth for each fixed $R$ is not enough.
 The key issue is uniform growth over local sets of distributions.  We do not assume that the entire model class $\calP$ is weakly compact.  Instead, \Cref{sec:local-witnesses} develops local REGROW witnesses, which are the only compactness/regularity properties needed by the delay proofs.

\subsection{SR-style e-detector for ARL control}

For a threshold $A>1$, define
\begin{equation}\label{eq:ordinary-arl-detector}
  D_t^{\mathrm{ARL}}
  =
  \inf_{R\in\calP}\sum_{s=1}^t M_{s:t}^R,
  \qquad
  \tau_A^{\mathrm{ARL}}
  =
  \inf\{t\ge1:D_t^{\mathrm{ARL}}\ge A\}.
\end{equation}

Here and henceforth, we always implicitly assume that $D_t^{\mathrm{ARL}}$ and similarly defined quantities are measurable, so that $\tau_A^{\mathrm{ARL}}$ is a well-defined stopping time. This is not immediate by definition, since we are taking an infimum over a possibly nonparametric class $\calP$. In most examples of practical interest, we expect it to be measurable, and we provide several such examples later in the paper.

\begin{theorem}[ARL control]\label{thm:arl-control}
$D_t^{\mathrm{ARL}}$ is an e-detector, and thus for every $R\in\calP$,
$\Einf{R}\tau_A^{\mathrm{ARL}}\ge A,$
and, moreover, for every $m\in\N$,
\begin{equation}
    \label{eq:lfa}
     \Pinf{R}(\tau_A^{\mathrm{ARL}}\le m)\le \frac{m}{A}.
\end{equation}
\end{theorem}

\emph{Proof.} See \Cref{app:proof-validity}.

\subsection{SR-style detector for PFA control}

For false-alarm probability control, choose deterministic weights $ \pi_s>0$ such that
  $\sum_{s=1}^{\infty}\pi_s\le1.$
Define
\begin{equation}\label{eq:ordinary-pfa-detector}
  D_t^{\mathrm{PFA}}
  =
  \inf_{R\in\calP}\sum_{s=1}^t \pi_sM_{s:t}^R,
  \qquad
  \tau_\alpha^{\mathrm{PFA}}
  =
  \inf\left\{t\ge1:D_t^{\mathrm{PFA}}\ge \frac1\alpha\right\}.
\end{equation}
A canonical choice is
\begin{equation}\label{eq:pi-weights}
  \pi_s=\frac{c_\pi}{s\{\log(e s)\}^2},
\end{equation}
where $c_\pi>0$ normalizes the sum to be at most one.  Then
\begin{equation}\label{eq:pi-log}
  -\log \pi_s=\log s+2\log\log(e s)+O(1).
\end{equation}

\begin{theorem}[Global PFA control]\label{thm:pfa-control}
$D_t^{\mathrm{PFA}}$ is an e-process and thus
for every $R\in\calP$, $ \Pinf{R}(\tau_\alpha^{\mathrm{PFA}}<\infty)\le \alpha.$
\end{theorem}

\emph{Proof.} See \Cref{app:proof-validity}.

The next three sections examine some concrete methodological examples.

\section{Example 1: Sub-Gaussian mean-change}
\label{sec:sub-Gaussian}

We first provide a concrete nonparametric example with $\calP$ being the set of all $\sigma^2$-sub-Gaussian distributions.
Let $$\mathcal{P}_\theta = \{ P \in M_1(\mathbb{R}) : \mathbb{E}_P[X]=\theta,\mathbb{E}_P[e^{\lambda(X-\mathbb{E}_P[X])}] \leq e^{\frac{\lambda^2}{2\sigma^2}} , \forall \lambda \in \mathbb{R} \}.$$ 
Note that for each $P\in\calP$, there exists $\theta\in\mathbb R$ such that $P\in\calP_\theta$.
Fix any $\rho>0$ and define
\begin{equation}
\label{eq:mixture-evalue}
  M_{s:t}^{\theta}
  =
  \frac{1}{\sqrt{1+\rho^2\sigma^2 n_{s:t}}}
  \exp\left\{
    \frac{\rho^2 n_{s:t}^2(\bar X_{s:t}-\theta)^2}
         {2(1+\rho^2\sigma^2 n_{s:t})}
  \right\},
\end{equation}
where $n_{s:t}=t-s+1,
  \bar X_{s:t}=\frac{1}{n_{s:t}}\sum_{i=s}^t X_i.$
\begin{proposition}
\label{prop:gaussian-mixture}
 For any $P\in\calP_\theta$, each process $(M_{s:t}^{\theta})_{t}$ is an $s$-delay e-process as in Definition \ref{ass:eprocess}.  In particular, it is a test supermartingale:
\[
  \E_P^\infty[M_{s:t}^{\theta}\mid \calF_{t-1}]
  \le M_{s:t-1}^{\theta},\qquad s<t.
\]
\end{proposition}
\emph{Proof.} See \Cref{app:proof-sub-Gaussian}.

With the above e-process, our detector statistic for ARL control in \eqref{eq:ordinary-arl-detector} reduces to
\begin{equation}
\label{eq:D-def}
D_t^{\mathrm{ARL}}=\inf_{\theta\in\R}\sum_{s=1}^t M_{s:t}^{\theta}.
\end{equation}
Suppose that the true pre-change $P\in\mathcal{P}_\mu$ and the post-change $Q\in\mathcal{P}_\nu$, for some $\mu\neq\nu$.
The following theorem establishes the first-order upper bound in an asymptotic late-change regime in terms of 
\begin{equation}
\label{eq:subg-information}
I:=\inf_{P_0\in \mathcal{P}_\mu,P_1\in \mathcal{P}_\nu}\KL(P_1\|P_0)=\frac{(\mu-\nu)^2}{2\sigma^2}.
\end{equation}

\begin{theorem}[First-order upper bound in the late-change regime]
\label{thm:upper-bdd-sub-Gaussian}
Consider a sequence $\alpha\downarrow 0$, with $A=1/\alpha$ and $L=\log A$.  Suppose the changepoint $T_\alpha$ satisfies
\begin{equation}
\label{eq:late-change-assumption}
  \frac{T_\alpha}{L}\to\infty,
  \qquad
  \log T_\alpha=o(L).
\end{equation}
Then, under $\bbP:=\bbP^{T_\alpha}_{P,Q}$,
\begin{equation}
   \tau_A^{\mathrm{ARL}}-T_\alpha
    \le
    (1+o_\bbP(1))\frac{L}{I}.
\end{equation}
Moreover,
$\E^{T_\alpha}_{P,Q}[(\tau_A^{\mathrm{ARL}}-T_\alpha)^+]
  \le
  (1+o(1))\frac{L}{I},$
and since $\bbP^{T_\alpha}_{P,Q}(\tau_A^{\mathrm{ARL}}\le T_\alpha)=o(1)$,
\[
  \CADD_{T_\alpha}(\tau_A^{\mathrm{ARL}})
  \le
  (1+o(1))\frac{L}{I}.
\]
\end{theorem}
\emph{Proof.} See \Cref{app:proof-sub-Gaussian}.

Similarly, our detector statistic for PFA control in \eqref{eq:ordinary-pfa-detector} reduces to
\begin{equation}
\label{eq:D-def-PFA}
D_t^{\mathrm{PFA}}=\inf_{\theta\in\R}\sum_{s=1}^t \pi_sM_{s:t}^{\theta}.
\end{equation}
Suppose that the true pre-change $P\in\mathcal{P}_\mu$ and the post-change $Q\in\mathcal{P}_\nu$, for some $\mu\neq\nu$.
The following theorem is the positive result for the false-alarm-valid detector.  It is stated as a high-probability detection-delay bound (recall from \Cref{subsec:metrics} that CADD is not a right metric for the PFA setting). Define,
\begin{equation}
\label{eq:hpi-def}
  h_\pi(T)
  :=
  \max\{-\log\pi_{T+1},-\log\pi_{T-\lfloor T/2\rfloor+1}\}
  +
  \frac12\log(1+2\rho^2\sigma^2 T).
\end{equation}
For the weight sequence $\{\pi_s\}_{s\in\N}$ in \eqref{eq:pi-weights},
\begin{equation}
\label{eq:hpi-OlogT}
  h_\pi(T)=O(\log T).
\end{equation}
As before, suppose that the true pre-change $P\in\mathcal{P}_\mu$ and the post-change $Q\in\mathcal{P}_\nu$, for some $\mu\neq\nu$ and let $I=\inf_{P_0\in \mathcal{P}_\mu,P_1\in \mathcal{P}_\nu}\KL(P_1\|P_0)=\frac{(\mu-\nu)^2}{2\sigma^2}.$
\begin{theorem}[Delay upper bound]
\label{thm:subg-upper-pfa}
Consider a sequence $\alpha\downarrow 0$, with $L=\log(1/\alpha)$, and the changepoints $T_\alpha\to\infty$.  Suppose
\begin{equation}
\label{eq:late-assumption-pfa}
  \frac{ T_\alpha}{L+h_\pi(T_\alpha)}\to\infty.
\end{equation}
Then, under $\bbP:=\bbP^{T_\alpha}_{P,Q}$,
\begin{equation}
\label{eq:upper-op}
  \tau_\alpha^{\mathrm{PFA}}-T_\alpha
  \le
  (1+o_{\bbP}(1))\frac{L+h_\pi(T_\alpha)}{I}.
\end{equation}
For \eqref{eq:pi-weights}, this is
\begin{equation}
\label{eq:upper-OlogT}
  \tau_\alpha^{\mathrm{PFA}}-T_\alpha
  \le
  (1+o_{\bbP}(1))
  \frac{L+O(\log T_\alpha)}{I}.
\end{equation}
\end{theorem}
\emph{Proof.} See \Cref{app:proof-sub-Gaussian}.
Comparing this result with the CADD bound for the ARL controlling detector, we see that we need to pay an additional $O(\log T_\alpha)$ price for the 
stricter false alarm guarantee.

\begin{corollary}[When the first-order term is unchanged]
\label{cor:first-order-same}
Use the spending sequence \eqref{eq:pi-weights}.  If $\frac{I T_\alpha}{L}\to\infty,
  \log T_\alpha=o(L),$
then
\[
  \tau_\alpha^\pi-T_\alpha
  \le
  (1+o_{\bbP}(1))\frac{L}{I}.
\]
\end{corollary}

\emph{Proof.} See \Cref{app:proof-sub-Gaussian}.

Thus, the strengthened false-alarm guarantee costs only a second-order additive term whenever $\log T_\alpha=o(L)$.

It is worth mentioning that the above specific construction for the sub-Gaussian example can only give bounds in terms the infimum KL information, $I=\inf_{P_0\in \mathcal{P}_\mu,P_1\in \mathcal{P}_\nu}\KL(P_1\|P_0),$ when the true pre-change $P\in\mathcal{P}_\mu$ and the post-change $Q\in\mathcal{P}_\nu$. In general, this quantity can be substantially smaller than the instance-specific information number, $I^*=\KL(Q\|P)$, which governs the optimal asymptotic delay in the classical separated SCD problems.

A natural question is whether point-null e-processes can instead adapt to the true pair $(P,Q)$ and attain the instance-specific rate $I^*=\KL(Q\|P)$.  The general theory in \Cref{sec:upper-bdd} answers this affirmatively under local uniform-growth conditions.

\section{Example 2: Bounded mean change detection}
\label{sec:bounded-mean}
Let $\calP$ be the set of all distributions on $[0,1]$:
\[
 \mathcal P=\mathcal M_1([0,1]),
\]
and $\mathcal P_\theta$ be the set of all distributions on $[0,1]$ having mean $\theta$, i.e., $\mathcal P_\theta=\{P\in\mathcal P:\E_PX=\theta\}.$
For predictable bets $\lambda_{s,i}(\theta)\in\left[-\frac1{1-\theta},\frac1\theta\right]$,
\begin{equation}
\label{eq:bounded-linear-wealth-main}
 W_{s:t}^{\theta}=\prod_{i=s}^t
 \{1+\lambda_{s,i}(\theta)(X_i-\theta)\}
\end{equation}
is a nonnegative test martingale whenever
$\E[X_i\mid\mathcal F_{i-1}]=\theta$.  Thus mixtures and predictable
no-regret strategies furnish restartable e-processes for the bounded
conditional-mean null.

A particularly natural choice is the two-asset universal portfolio of
\citet{cover1991universal}.  For each $\theta\in(0,1)$, parameterize a constant portfolio by
\[
 \lambda_\theta(p)=\frac{p-\theta}{\theta(1-\theta)},\qquad p\in[0,1],
\]
and mix against the arcsine law
\[
 \Pi_J(dp)=\frac{dp}{\pi\sqrt{p(1-p)}}.
\]
The resulting restartable wealth is
\begin{equation}
\label{eq:bounded-UP-mixture}
 U_{s:t}^{\theta}
 =\int_0^1\prod_{i=s}^t
 \{1+\lambda_\theta(p)(X_i-\theta)\}\,\Pi_J(dp).
\end{equation}
For a prescribed horizon, this continuous mixture can be evaluated exactly
by Gauss--Chebyshev quadrature; it is not replaced below by a finite fixed
portfolio grid.

For a post-change law $Q$, define its betting information against the mean
null $\theta$ by
\begin{equation}
\label{eq:bounded-betting-information}
 I_{\rm bet}(Q,\theta)
 =\sup_{\lambda\in\Lambda_\theta}
   \E_Q\log\{1+\lambda(X-\theta)\}.
\end{equation}
The completed reverse information projection identity is
\begin{equation}
\label{eq:bounded-betting-ripr-main}
 I_{\rm bet}(Q,\theta)
 =\inf_{R\in\mathcal P_\theta}\KL(Q\|R).
\end{equation}

\begin{proposition}[Validity, growth, and variance adaptivity]
\label{prop:bounded-UP-growth-main}
For each $\theta\in(0,1)$, $(U_{s:t}^{\theta})_{t\ge s-1}$ is an $s$-delay
test martingale under every adapted $[0,1]$-valued process satisfying
$\E[X_t\mid\mathcal F_{t-1}]=\theta$.  If
$X_i\stackrel{\rm iid}{\sim}Q$ with mean $b\ne\theta$, then
\[
 \liminf_{n\to\infty}\frac1n\log U_{1:n}^{\theta}
 \ge I_{\rm bet}(Q,\theta)
 \qquad Q^\infty\text{-a.s.}
\]
\end{proposition}

  See
\Cref{app:prop-bddmean} for a proof.

For deterministic weights $w_s\in(0,1]$, define
\begin{equation}
\label{eq:bounded-UP-detector-main}
 D_t^{\rm UP,w}=\inf_{\theta\in(0,1)}
       \sum_{s=1}^t w_sU_{s:t}^{\theta},
 \qquad
 \tau_A^{\rm UP,w}=\inf\{t:D_t^{\rm UP,w}\ge A\}.
\end{equation}
The ARL detector uses $w_s\equiv1$; the PFA detector uses $w_s=\pi_s$ and
$A=1/\alpha$.
It inherits ARL or PFA validity from
\Cref{thm:arl-control,thm:pfa-control}.
\begin{theorem}[Bounded-mean delay]
\label{thm:bounded-main-delay}
Suppose the pre-change mean is $a\in(0,1)$ and the post-change law is $Q$ with
mean $b\ne a$.  Let $I_{\rm bet}:=I_{\rm bet}(Q,a)\in(0,\infty)$.
If $T_A/\log A\to\infty$ and $\log T_A=o(\log A)$, then
\[
 \tau_A^{\rm UP,ARL}-T_A
 \le(1+o_{\bbP}(1))\frac{\log A}{I_{\rm bet}},
\]
and
\[
 \CADD_{P,Q}^{T_A}(\tau_A^{\rm UP,ARL})
 \le(1+o(1))\frac{\log A}{I_{\rm bet}}.
\]
If
$\frac{T_\alpha}
 {\log(1/\alpha)-\log\pi_{T_\alpha+1}}\to\infty,$
then the PFA detector satisfies
\[
 \tau_\alpha^{\rm UP,PFA}-T_\alpha
 \le(1+o_{\bbP}(1))
 \frac{\log(1/\alpha)-\log\pi_{T_\alpha+1}}{I_{\rm bet}}.
\]
\end{theorem}

A complete proof is in
\Cref{app:bounded-UP-delay}.  

For computation, a prescribed horizon $H$ permits exact evaluation of the
continuous arcsine mixture through time $H$ by Gauss--Chebyshev quadrature
with $K\ge(H+1)/2$ nodes: the time-$n$ integrand is a polynomial of degree
$n$ in $p$.  A long-running detector can instead use a no-regret portfolio,
an increasing positive quadrature, or a mixture-of-lower-bounds
implementation.  
A fixed no-regret portfolio takes $O(1)$ work per active start and observation and therefore $O(K)$ work per active start, unless further polynomial-recursion structure is exploited.

\section{Example 3: Gaussian change detection with unknown variance}
\label{sec:gaussian-unknown-variance}

We now take $\mathcal P$ to be the set of all Gaussian distributions:
\[
  \mathcal P
  :=\{N(\mu,\sigma^2):\mu\in\mathbb R,\ \sigma^2>0\}.  
\]
For each $\theta\in\mathbb R $,
define
\[\mathcal N_\theta
  :=\{N(\theta,\sigma^2):\sigma^2>0\}.\]
We consider two change detection problems in the following subsections. First, we consider the problem of detecting changes in the mean parameter only, when the variance is unknown and unrestricted.  Second, we
construct a detector for changes in either the mean or variance parameter and  recover the stronger
instance-specific information rate $\KL(Q\|P)$.

\subsection{Mean change detection using universal-inference t-test e-process}

For every start $s$, let $\widetilde\mu_{s:i-1}$ and
$\widetilde\sigma^2_{s:i-1}$ be predictable estimators, with the latter
strictly positive, based only on $X_s,\ldots,X_{i-1}$.  A concrete online
choice is
\begin{equation}
\label{eq:gaussian-regularized-predictors}
  \widetilde\mu_{s:i-1}
  :=
  \begin{cases}
    m_0,&k=0,\\
    \bar X_{s:i-1},&k\ge1,
  \end{cases}
  \qquad
  \widetilde\sigma^2_{s:i-1}
  :=
  \frac{v_0+\sum_{j=s}^{i-1}
        (X_j-\bar X_{s:i-1})^2}{k+\nu_0},
  \qquad k=i-s,
\end{equation}
where $m_0\in\mathbb R$ and $v_0,\nu_0>0$ are fixed, and the centered sum is
zero when $k=0$.  

For $n=n_{s:t}$, define
\begin{equation}
\label{eq:UI-t-eprocess}
  R_{s:t}^{\theta}
  :=
  \left\{
    \frac1n\sum_{i=s}^t(X_i-\theta)^2
  \right\}^{n/2}
  e^{n/2}
  \prod_{i=s}^t
  \frac1{\widetilde\sigma_{s:i-1}}
  \exp\left\{
    -\frac{(X_i-\widetilde\mu_{s:i-1})^2}
           {2\widetilde\sigma^2_{s:i-1}}
  \right\},
\end{equation}
with $R_{s:s-1}^{\theta}=1$.  Notice that for a nonzero null mean $\theta$,
only the residual sum of squares in the first factor is shifted; the
predictive numerator remains fitted to the unshifted observations.  This is
the power-preserving translation described by \citet{wang2024anytime}.

\begin{proposition}[Universal-inference t e-process]
\label{prop:UI-t-eprocess}
For every $\theta\in\mathbb R$, the process
$(R_{s:t}^{\theta})_{t\ge s-1}$ is an $s$-delay e-process under every
$N(\theta,\sigma^2)^\infty$, uniformly over $\sigma^2>0$.
If $X_i\stackrel{\mathrm{iid}}\sim N(m,v)$ and the predictors satisfy the
consistency and inverse-moment conditions of \citet[Proposition~3.3]{wang2024anytime}
---in particular, the regularized estimators
\eqref{eq:gaussian-regularized-predictors} do---then
\begin{equation}
\label{eq:UI-t-growth}
  \frac1n\log R_{1:n}^{\theta}
  \longrightarrow
  J_{m,v}(\theta)
  :=\frac12\log\left(1+\frac{(m-\theta)^2}{v}\right)
  \qquad\text{almost surely}.
\end{equation}
Moreover,
\begin{equation}
\label{eq:t-reverse-I-projection}
  J_{m,v}(\theta)
  =\inf_{r>0}
    \KL\bigl(N(m,v)\,\|\,N(\theta,r)\bigr),
\end{equation}
so this is the largest possible first-order growth rate for an e-process
valid under the whole composite mean null $\mathcal N_\theta$.
\end{proposition}
\emph{Proof.} See \Cref{app:proof-gaussian}.


For weights $w_s\in(0,1]$, define the studentized detector
\begin{equation}
\label{eq:t-detector}
  D_t^{t,w}
  :=\inf_{\theta\in\mathbb R}
      \sum_{s=1}^t w_sR_{s:t}^{\theta},
  \qquad
  \tau_A^{t,w}
  :=\inf\{t\ge1:D_t^{t,w}\ge A\}.
\end{equation}
It inherits ARL or PFA validity from
\Cref{thm:arl-control,thm:pfa-control}.

The optimization in \eqref{eq:t-detector} is one-dimensional and convex.
Indeed, apart from a positive factor independent of $\theta$,
\[
  R_{s:t}^{\theta}
  \propto
  \left\{\sum_{i=s}^t(X_i-\theta)^2\right\}^{n_{s:t}/2}
  =\|X_{s:t}-\theta\mathbf 1\|_2^{n_{s:t}},
\]
which is convex for $n_{s:t}\ge1$.  Whenever the residual sum of squares is
positive,
\begin{equation}
\label{eq:UI-t-derivative}
  \frac{\partial}{\partial\theta}\log R_{s:t}^{\theta}
  =
  \frac{n_{s:t}^2(\theta-\bar X_{s:t})}
       {\sum_{i=s}^t(X_i-\theta)^2}.
\end{equation}
The objective minimizer therefore lies between the smallest and largest
interval means.  At each new observation, all counts, sums, sums of squares,
regularized predictors, and cumulative predictive log scores for the $t$
active starts can be updated in $O(t)$ operations.  Each objective or
subgradient evaluation is then $O(t)$, so bisection or safeguarded Newton
costs $O(t\log(1/\epsilon_{\mathrm{num}}))$ at time $t$ and uses $O(t)$
memory.  Thus the exact all-start detector has quadratic total cost through a
horizon $N$, just as the other exact SR-style implementations in this paper.


To study the asymptotics of the detection delay, suppose now that
\[
  P=N(a,u),
  \qquad
  Q=N(b,v),
  \qquad a\ne b.
\]
Define the t-information number
\begin{equation}
\label{eq:t-information}
  I_t(P,Q)
  :=J_{b,v}(a)
  =\frac12\log\left(1+\frac{(b-a)^2}{v}\right).
\end{equation}
For a common but unknown variance $u=v=\sigma^2$, this is
$\tfrac12\log(1+(b-a)^2/\sigma^2)$.

\begin{theorem}[Delay of the studentized detector]
\label{thm:t-delay-unified}
Let $A_\alpha\to\infty$, $L_\alpha:=\log A_\alpha$, and let
$w_{s,\alpha}\in(0,1]$.  Define
\begin{equation}
\label{eq:t-effective-boundary}
  B_\alpha^t
  :=L_\alpha-\log w_{T_\alpha+1,\alpha}.
\end{equation}
Assume
\begin{equation}
\label{eq:t-late-change-condition}
  B_\alpha^t\to\infty,
  \qquad
  \frac{T_\alpha}{B_\alpha^t}\to\infty,
  \qquad
  L_\alpha-\log w_{1,\alpha}=o(T_\alpha).
\end{equation}
For every fixed $\varepsilon>0$, let
\[
  d_\alpha
  :=\left\lceil
       (1+\varepsilon)
       \frac{B_\alpha^t}{I_t(P,Q)}
     \right\rceil.
\]
Then
\begin{equation}
\label{eq:t-delay-hp}
  \bbP_{P,Q}^{T_\alpha}
  \bigl(\tau_{A_\alpha}^{t,w}\le T_\alpha+d_\alpha\bigr)
  \longrightarrow1.
\end{equation}
Consequently, for the ARL detector,
\[
  \tau_A^{t,\mathrm{ARL}}-T_A
  \le
  (1+o_{\bbP}(1))\frac{\log A}{I_t(P,Q)}
\]
whenever $T_A/\log A\to\infty$.  For the PFA detector with
\eqref{eq:pi-weights},
\begin{equation}
\label{eq:t-PFA-delay}
  \tau_\alpha^{t,\mathrm{PFA}}-T_\alpha
  \le
  (1+o_{\bbP}(1))
  \frac{\log(1/\alpha)+\log T_\alpha
        +2\log\log(eT_\alpha)+O(1)}{I_t(P,Q)}.
\end{equation}
The same statements hold conditionally on no false alarm before $T_\alpha$
whenever the corresponding ARL or PFA pre-change false-alarm probability
vanishes.
\end{theorem}
\emph{Proof.} See \Cref{app:proof-gaussian}.

The denominator $I_t(P,Q)$ matches the reverse information projection in
\eqref{eq:t-reverse-I-projection}.  Thus it is minimax optimal if the
pre-change null is specified only by its mean and its variance is allowed to
range freely.  It need not equal the instance-specific information
$\KL(Q\|P)$ once a long prefix can learn the actual pre-change variance; we
return to this distinction below.

\subsection{An instance-optimal full Gaussian change detector}

If the entire Gaussian law is treated as the candidate no-change
distribution, universal inference yields a point-null martingale with full
KL growth.  Let
\[
  q_{s:i-1}(x)
  :=\phi_{\widetilde\mu_{s:i-1},
           \widetilde\sigma^2_{s:i-1}}(x)
\]
be any predictable Gaussian density, for example the online predictors in
\eqref{eq:gaussian-regularized-predictors}.  For $\theta\in\mathbb R$ and
$r>0$, define
\begin{equation}
\label{eq:full-Gaussian-martingale}
  L_{s:t}^{\theta,r}
  :=\prod_{i=s}^t
  \frac{q_{s:i-1}(X_i)}{\phi_{\theta,r}(X_i)},
  \qquad L_{s:s-1}^{\theta,r}:=1.
\end{equation}
Under $N(\theta,r)^\infty$, this is a test martingale in the canonical
filtration.  Under $N(m,v)^\infty$ and consistent predictors,
\begin{equation}
\label{eq:full-Gaussian-growth}
  \frac1n\log L_{1:n}^{\theta,r}
  \longrightarrow
  \KL\bigl(N(m,v)\|N(\theta,r)\bigr)
  \qquad\text{almost surely}.
\end{equation}
Define
\begin{equation}
\label{eq:full-Gaussian-detector}
  D_t^{\mathrm G,w}
  :=\inf_{\theta\in\mathbb R,\ r>0}
      \sum_{s=1}^t w_sL_{s:t}^{\theta,r},
  \qquad
  \tau_A^{\mathrm G,w}
  :=\inf\{t\ge1:D_t^{\mathrm G,w}\ge A\}.
\end{equation}

\begin{theorem}[Instance-optimal Gaussian delay]
\label{thm:full-Gaussian-delay}
Let $P=N(a,u)$ and $Q=N(b,v)$ be distinct, and put
\begin{equation}
\label{eq:full-Gaussian-information}
  I_{\mathrm G}(P,Q)
  :=\KL(Q\|P)
  =\frac12\left\{
       \log\frac uv+
       \frac{v+(b-a)^2}{u}-1
     \right\}.
\end{equation}
Let $A_\alpha\to\infty$, $L_\alpha=\log A_\alpha$, and define
$B_\alpha^{\mathrm G}:=L_\alpha-
\log w_{T_\alpha+1,\alpha}$.  If
$ B_\alpha^{\mathrm G}\to\infty,
  \frac{T_\alpha}{B_\alpha^{\mathrm G}}\to\infty,
  L_\alpha-\log w_{1,\alpha}=o(T_\alpha),$
then, for every fixed $\varepsilon>0$,
\begin{equation}
\label{eq:full-Gaussian-delay-hp}
  \bbP_{P,Q}^{T_\alpha}
  \left(
    \tau_{A_\alpha}^{\mathrm G,w}-T_\alpha
    \le
    \left\lceil
      (1+\varepsilon)
      \frac{B_\alpha^{\mathrm G}}{I_{\mathrm G}(P,Q)}
    \right\rceil
  \right)
  \longrightarrow1.
\end{equation}
In particular, when $u=v=\sigma^2$,
$ I_{\mathrm G}(P,Q)=\frac{(b-a)^2}{2\sigma^2}.$
Thus the ARL delay is first-order $\log A/I_{\mathrm G}(P,Q)$.  The PFA
delay is first-order
\[
 \frac{\log(1/\alpha)-\log\pi_{T+1}}{I_{\mathrm G}(P,Q)}.
\]
\end{theorem}

The optimization in \eqref{eq:full-Gaussian-detector} is a smooth
convex problem after reparametrizing by the Gaussian natural parameters
\[
  \eta_1=\theta/r,
  \qquad
  \eta_2=-1/(2r)<0.
\]
For each start $s$,
\[
  L_{s:t}^{\eta}
  =C_{s:t}
   \exp\{n_{s:t}A(\eta)-\eta_1S_{s:t}-\eta_2Q_{s:t}\},
  \qquad
  Q_{s:t}:=\sum_{i=s}^tX_i^2,
\]
where $C_{s:t}$ is independent of $\eta$ and
\[
  A(\eta)
  =-\frac{\eta_1^2}{4\eta_2}
    +\frac12\log\frac{\pi}{-\eta_2}
\]
is the convex Gaussian log-partition function.  The sum over starts is
therefore convex on $\{\eta_2<0\}$.  Prefix sufficient statistics and
predictive log scores give $O(t)$ objective, gradient, and Hessian
evaluations; a warm-started damped Newton method is consequently practical.

When the variance is common,
\begin{equation}
\label{eq:t-vs-full-info}
  I_t(P,Q)
  =\frac12\log\left(1+\frac{(b-a)^2}{\sigma^2}\right)
  \le
  I_{\mathrm G}(P,Q)
  =\frac{(b-a)^2}{2\sigma^2},
\end{equation}
with strict inequality for a nonzero fixed standardized shift.  The
studentized detector is growth-rate optimal for a mean-null with unrestricted
variance, and it is locally efficient because
$\tfrac12\log(1+x)=x/2+o(x)$ as $x\downarrow0$.  The full-Gaussian detector
uses the prefix to learn the actual variance and attains the
instance-specific denominator $\KL(Q\|P)$.  

\section{Detection delay upper bound for general $\calP$}\label{sec:upper-bdd}
\subsection{Adjusters and nondecreasing block e-processes}

The restarted e-process $t\mapsto M_{s:t}^R$ need not be nondecreasing.  Hence, at time $T+d$, the block $M_{1:T+d}^R$ need not retain the pure pre-change evidence $M_{1:T}^R$.  The fix is to apply the ``adjuster'' idea \citep{choe2026combining} to the running maximum.

\begin{definition}[Adjuster]
\label{def:adjuster}
An increasing, right-continuous function $a:[1,\infty]\to[0,\infty]$ is called an adjuster if
\[
  \int_1^{\infty} \frac{a(x)}{x^2}\,dx\le 1.
\]
We call $a$ growth preserving if, additionally,
\[
  \log a(e^y)=y-o(y)\qquad (y\to\infty).
\]
\end{definition}
The integral condition preserves e-validity after taking a running
maximum, while the growth-preserving condition ensures that this
operation changes log-growth only by a lower-order term.
A convenient example is
\begin{equation}\label{eq:canonical-adjuster}
  a(x)=\frac{2x}{(2+\log x)^2},\qquad x\ge 1.
\end{equation}
Indeed, with $u=2+\log x$,
\[
  \int_1^\infty \frac{2x}{(2+\log x)^2}\frac{dx}{x^2}
  =\int_1^\infty \frac{2}{x(2+\log x)^2}\,dx
  =\int_2^\infty \frac{2}{u^2}\,du=1,
\]
and
\[
  \log a(e^y)=y-2\log(2+y)+\log 2=y-o(y).
\]

The next lemma is the running-maximum special case of the adjuster theory of  \cite{choe2026combining}.  

\begin{lemma}[Adjusted running maximum is an e-process]\label{lem:adjusted-max}
Let $(E_t)_{t\ge0}$ be an e-process under a null law $\mathbb P$, with $E_0=1$, and let
\[
  E_t^*=\max_{0\le u\le t}E_u.
\]
If $a$ is an adjuster, then $(a(E_t^*))_{t\ge0}$ is an e-process under $\mathbb P$.
\end{lemma}

\begin{definition}[Adjusted $s$-delay e-process]\label{def:adjusted-block}
For every $R\in\calP$, $s\le t$, and $s$-delay e-process $M_{s:u}^R$, define the process $ \{\overline M_{s:t}^R\}_{t\in\N}$ such that $ \overline M_{s:t}^R=1,$ for $t=1,\cdots,s-1$ and
\[
  \overline M_{s:t}^R=a(M_{s:t}^{R,*}),
\]
where $M_{s:t}^{R,*}=\max_{s-1\le u\le t}M_{s:u}^R$. Then we call $ \{\overline M_{s:t}^R\}_{t\in\N}$ an adjusted $s$-delay e-process.
\end{definition}
Note that $t\mapsto \overline M_{s:t}^R$ is nondecreasing for each fixed $s\leq t$ and $R$.
\begin{corollary}[Validity and evidence retention]\label{cor:valid-retention}
For every $R\in\calP$ and every start $s$, $(\overline M_{s:t}^R)_{t}$ is an $s$-delay e-process under $\Pinf{R}$.  Moreover, for $s\le u\le t$,
\[
  \overline M_{s:t}^R\ge a(M_{s:u}^R).
\]
In particular, if the true change is at $T$, then at time $T+d$,
\begin{equation}
\label{eq:adj}
     \overline M_{1:T+d}^R\ge a(M_{1:T}^R),
  \qquad
  \overline M_{T+1:T+d}^R\ge a(M_{T+1:T+d}^R).
\end{equation}
\end{corollary}

\emph{Proof.} See \Cref{app:proof-local-witnesses}.

The first inequality in \eqref{eq:adj} retains pre-change evidence against candidates
far from the true pre-change law, while the second retains evidence
from the pure post-change block against nearby candidates.

\subsection{Verifying uniform growth without weak compactness of $\calP$}\label{sec:local-witnesses}

\begin{lemma}[Local KL neighborhoods]\label{lem:neighborhoods}
Let $P,Q\in\calP$ and suppose $I^*=\KL(Q\|P)\in(0,\infty)$.  For every $\eta\in(0,I^*)$, the set
\[
  G_\eta=\{R\in\calP:\KL(Q\|R)>I^*-\eta\}
\]
is a relatively weakly open neighborhood of $P$.  Moreover, if $U\subseteq\calP$ is any relatively weakly open neighborhood of $P$, then
\[
  c_U\defeq \inf_{R\in\calP\setminus U}\KL(P\|R)>0.
\]
Consequently, if $U_\eta$ is any relatively weakly open set satisfying $P\in U_\eta\subseteq G_\eta$, then
\begin{align}
  \inf_{R\notin U_\eta}\KL(P\|R)&\ge c_{U_\eta}>0,\label{eq:separate-P}\\
  \inf_{R\in U_\eta}\KL(Q\|R)&\ge I^*-\eta.\label{eq:local-Q}
\end{align}
\end{lemma}

\emph{Proof.} See \Cref{app:proof-local-witnesses}.

Thus, candidates near $P$ can be rejected after the change at rate
nearly $I^*$, while candidates outside a neighborhood of $P$ are
uniformly separated from $P$ and can be rejected using the long
pre-change prefix.

We start with the following definition.
Let $K\subseteq\calP$, $ \mathcal A\subseteq K^c$, and define
\[
    \Phi_K(S)
    :=
    \inf_{R\in K}\KL(S\|R),
    \qquad S\in\calP.
\]
\begin{definition}[Simultaneous REGROW regularity]
\label{def:simultaneous-regrow}
We say that $K$ is \emph{simultaneously REGROW-regular over
$\mathcal A$} if there exists a single nondecreasing e-process $\{E_n^{K,\mathcal A}\}_{n\ge 0}$
which is valid under every $R\in K$, and satisfies, simultaneously for every
$S\in\mathcal A$,
\[
    \liminf_{n\to\infty}
    \frac{1}{n}\log E_n^{K,\mathcal A}
    \ge
    \Phi_K(S),
    \qquad
    S^\infty\text{-a.s.}
\]
\end{definition}

The essential requirement is that the same e-process works for all
alternatives $S\in\mathcal A$; it may not be chosen after the true
alternative is known.

Next, we show that weak compactness of $K$ is a sufficient condition for this property by Theorem 2 of \cite{ram2026power}, along with \Cref{lem:growth-preserved} to make that e-process non-decreasing.  Weak compactness is not claimed to be necessary; in their power-one theorem, the underlying local requirement is positivity and weak lower semicontinuity of $\Phi_K$ at the alternatives of interest.

\begin{lemma}[Weakly compact classes are simultaneously REGROW-regular]
\label{lem:weakly-compact-regrow}
Let $\mathcal X$ be a Polish space, and let
$ K\subseteq\mathcal M_1(\mathcal X)$
be a nonempty weakly compact class of probability measures.
Then $K$ is {simultaneously REGROW-regular over any
$\mathcal A \subseteq K^c$}. 
\end{lemma}

\emph{Proof.} See \Cref{app:proof-local-witnesses}.

\begin{definition}[Countable local REGROW witness basis]\label{def:witness-basis}
We say that $\calP$ has a countable local REGROW witness basis if there exists a countable family $\mathscr B=\{B_j:j\ge 1\}$ of relatively weakly open subsets of $\calP$ such that:
\begin{enumerate}[label=\textup{(\roman*)}]
\item
For every $P\in\calP$ and every relatively weakly open
neighborhood $G$ of $P$, there exists $j\ge1$ such that
\[
    P\in B_j
    \subseteq
    \overline B_j
    \subseteq
    G,
\]
where $\overline B_j$ denotes relative weak closure in $\calP$.

\item
For every $j\ge1$, the class $\overline B_j$ is weakly compact. 
\item
For every $j\ge1$, the exterior class
$B_j^c:=\calP\setminus B_j$
is simultaneously REGROW-regular over $B_j$. 
\end{enumerate}
Moreover, $\mathscr B$ is called the countable local REGROW witness basis of $\calP$.
\end{definition}

Indeed, a weakly compact class is a compact metric space under the
relative weak topology. It therefore has a countable basis with
compact closures, and the complement of each open basis element is
also weakly compact (and hence, simultaneously REGROW-regular). Thus global weak compactness is a convenient
sufficient condition for \Cref{def:witness-basis}, as stated formally below.

\begin{proposition}[Weakly compact classes admit REGROW witness]
\label{lem:compact-witness-basis}
Let $\mathcal{X}$ be a Polish space, $ \mathcal{B}$ be its Borel $\sigma$-field, and let $\mathcal{P}$ be a weakly compact class of probability measures on $\mathcal{X}$. Then $\mathcal P$ admits a countable
universal local REGROW witness basis in the sense of
\Cref{def:witness-basis}.
\end{proposition}
\emph{Proof.} See \Cref{app:proof-local-witnesses}.

\begin{proposition}[Gaussian location family and REGROW witness]
\label{prop:gaussian-local-witness}
Let $\mathcal{P} = \{N(\theta, 1) : \theta \in \mathbb{R}\}$ be the unit-variance Gaussian location family. Then $\mathcal{P}$ is not weakly compact, but it admits a countable local REGROW witness basis in the sense of Definition \ref{def:witness-basis}.
\end{proposition}

\emph{Proof.} See \Cref{app:proof-local-witnesses}.

The above result shows that \Cref{def:witness-basis} is strictly more flexible
than global weak compactness.
The following theorem shows how local REGROW witnesses help in achieving a uniform growth rate.

\begin{theorem}[Universal local witnesses imply uniform growth]
\label{thm:witness-implies-growth}
Assume that $\calP$ has a countable local REGROW
witness basis $\mathscr B=\{B_j:j\ge1\}.$
Then for each $R\in\calP$ and $s\in\N$, there exists an $s$-delay e-process $\{\overline M_{s:t}^{R}\}_t$ with the following properties. For every fixed $R$ and $s$, the map
$t\mapsto\overline M_{s:t}^{R}$ is nondecreasing for $t\ge s$.
Moreover, for any $P,Q\in\calP$ such that
$I^*:=\KL(Q\|P)\in(0,\infty),$
and any fixed $\eta_{\mathrm{in}}\in(0,I^*)$,  there exists
$j_*=j_*(P,Q,\eta_{\mathrm{in}})\in\N$ such that
we have
$ P\in B_{j_*},
    \Phi_{\overline B_{j_*}}(Q)
    \ge
    I^*-\frac{\eta_{\mathrm{in}}}{2},$
and
$c_{\eta_{\mathrm{in}}}
    :=
    \Phi_{B_{j_*}^c}(P)
    >0,$
and for every
$\eta_{\mathrm{out}}\in(0,c_{\eta_{\mathrm{in}}})$,
\begin{equation}
\label{eq:pre-growth}
    \mathbb P_P^\infty
    \left(
        \inf_{R\notin B_{j_*}}
        \log\overline M_{1:n}^{R}
        \ge
        n\bigl(c_{\eta_{\mathrm{in}}}
        -\eta_{\mathrm{out}}\bigr)
    \right)
    \longrightarrow1,
\end{equation}
and
\begin{equation}
\label{eq:post-growth}
    \mathbb P_Q^\infty
    \left(
        \inf_{R\in B_{j_*}}
        \log\overline M_{1:n}^{R}
        \ge
        n\bigl(I^*-\eta_{\mathrm{in}}\bigr)
    \right)
    \longrightarrow1.
\end{equation}
\end{theorem}
\emph{Proof.} See \Cref{app:proof-local-witnesses}.

\begin{remark}[What is really assumed]
The detector itself is defined from point-null e-processes.  The witness augmentation is a way to prove uniform delay bounds for the infimum over the point-nulls.  The genuinely needed condition is not compactness of all of $\calP$, but rather the existence of local witnesses for the two classes that matter near a change: a compact neighborhood of $P$, which handles post-change evidence uniformly near $P$, and the exterior of that neighborhood, which handles pre-change anchoring uniformly away from $P$.
\end{remark}

\subsection{A generic upper bound}

We state the result once for a generic weight sequence $w_s$.  The ARL detector corresponds to $w_s\equiv1$ and threshold $A$.  The PFA detector corresponds to $w_s=\pi_s$ with $\sum_{s=1}^\infty\pi_s=1$ and threshold $A_\alpha=1/\alpha$.
Let $A_\alpha\to\infty$ and set $L_\alpha=\log A_\alpha$.  Let $T_\alpha$ be the changepoint, and let $w_s=w_{s,\alpha}\in(0,1]$ be deterministic weights.  Define
\begin{equation}\label{eq:Balpha}
  B_\alpha=L_\alpha-\log w_{T_\alpha+1}.
\end{equation}
For the ARL detector, $w_s\equiv1$, so $B_\alpha=L_\alpha$.  For the PFA detector with weights \eqref{eq:pi-weights}, $A_\alpha=1/\alpha$ and
\[
  B_\alpha=\log(1/\alpha)+\log T_\alpha+2\log\log(eT_\alpha)+O(1).
\]

For weights $w_s$, define 
\begin{equation}\label{eq:generic-adjusted-detector}
  D_t^w=\inf_{R\in\calP}\sum_{s=1}^t w_s\overline M_{s:t}^R,
  \qquad
  \tau_\alpha^w=\inf\{t\ge1:D_t^w\ge A_\alpha\}.
\end{equation}

\begin{theorem}[Late-change delay bound]
\label{thm:delay-upper-generic}
Assume that $\calP$ has a countable local REGROW witness basis.  Consider a sequence
$\alpha\downarrow 0$, with $A_\alpha=1/\alpha, L_\alpha=\log(A_\alpha)$.  Fix $\eps>0$ and set
\[
  d_\alpha=\left\lceil (1+\eps)\frac{B_\alpha}{I^*}\right\rceil,
  \qquad
  B_\alpha=L_\alpha-\log w_{T_\alpha+1}.
\]
Suppose there exists $\eta_{\mathrm{in}}\in(0,I^*)$ and, for the corresponding
$c_{\eta_{\mathrm{in}}}>0$ furnished by
\Cref{thm:witness-implies-growth}, there exists an
$\eta_{\mathrm{out}}\in(0,c_{\eta_{\mathrm{in}}})$ such that
\begin{align}
  T_\alpha
  \bigl(c_{\eta_{\mathrm{in}}}-\eta_{\mathrm{out}}\bigr)
  +\log w_1
  &\ge L_\alpha
  \quad\text{eventually},
  \label{eq:anchor-condition}\\
  d_\alpha(I^*-\eta_{\mathrm{in}})
  +\log w_{T_\alpha+1}
  &\ge L_\alpha
  \quad\text{eventually}.
  \label{eq:post-condition}
\end{align}
Then, for each $s$ and $R\in\calP$, there exists an $s$-delay e-process
$\{\overline M_{s:t}^R\}_{t\geq s-1}$, such that the corresponding stopping time
$\tau_\alpha^w$ in \eqref{eq:generic-adjusted-detector} satisfies
\[
  \mathbb P_{P,Q}^{T_\alpha}
  \bigl(\tau_\alpha^w\le T_\alpha+d_\alpha\bigr)\to1.
\]
\end{theorem}

\emph{Proof.} See \Cref{app:proof-generic-upper}.

\begin{corollary}[ARL late-change delay]\label{cor:arl-delay}
 Consider a sequence $\alpha\downarrow 0$, with $A=1/\alpha$ and $L=\log A$.  Suppose the changepoint $T_A$ satisfies $ \frac{T_A}{L}\to\infty.$
 If $\calP$ has a countable local REGROW witness basis, then there exists an ARL detector $\tau_A^{\mathrm{ARL}}$ of the form \eqref{eq:ordinary-arl-detector}, such that for every $\eps>0$,
\[
  \mathbb P_{P,Q}^{T_A}\left(
  \tau_A^{\mathrm{ARL}}-T_A
  \le
  (1+\eps)\frac{L}{I^*}+1
  \right)\to1.
\]
If, additionally, $T_A/A\to0$, then
$\mathbb P_{P,Q}^{T_A}(\tau_A^{\mathrm{ARL}}>T_A)\to1,$ and hence,
for every $\eps>0$,
\[
  \mathbb P_{P,Q}^{T_A}\left(
  \tau_A^{\mathrm{ARL}}-T_A
  \le
  (1+\eps)\frac{L}{I^*}+1
  \mid \tau_A^{\mathrm{ARL}}>T_A\right)\to1.
\]
\end{corollary}

\emph{Proof.} See \Cref{app:proof-generic-upper}.

\begin{corollary}[PFA late-change delay]\label{cor:pfa-delay}
Consider a sequence $\alpha\downarrow 0$, with $A=1/\alpha$ and $L=\log A$.  Suppose the changepoint $T_A$ satisfies $ \frac{T_A}{L}\to\infty.$
 If $\calP$ has a countable local REGROW witness basis, then there exists a PFA detector $\tau_\alpha^{\mathrm{PFA}}$ of the form \eqref{eq:ordinary-pfa-detector}, such that for every $\eps>0$,
\[
  \mathbb P_{P,Q}^{T_A}\left(
  \tau_\alpha^{\mathrm{PFA}}-T_A
  \le
  (1+\eps)\frac{\log(1/\alpha)-\log\pi_{T_A+1}}{I^*}+1
  \right)\to1.
\]
For the canonical weights \eqref{eq:pi-weights}, this upper bound is
\[
  (1+\eps)\frac{
    \log(1/\alpha)+\log T_A+2\log\log(eT_A)+O(1)
  }{I^*}.
\]
Moreover,
\[
  \Pchg{P}{Q}{T_A}(\tau_\alpha^{\mathrm{PFA}}> T_A)
  \ge1- \alpha.
\]
So, false alarms before the change vanish,  i.e., $\Pchg{P}{Q}{T_A}(\tau_\alpha^{\mathrm{PFA}}> T_A)\to 1$, as $\alpha\downarrow0$.
\end{corollary}

\emph{Proof.} See \Cref{app:proof-generic-upper}.

\begin{remark}[Why the post-change start alone is not enough]
The term starting at $s=T+1$ cannot reject $R=Q$, because the block $X_{T+1:T+d}$ is from i.i.d.\ $Q$.  The prefix term is needed to reject $R=Q$ and nearby alternatives.  The adjusted detector keeps the prefix evidence alive through monotonicity, avoiding the need to analyze the raw cross-block $M_{1:T+d}^R$.
\end{remark}

Note that the above result proves the existence of such e-detectors achieving a detection delay bound that depends on $(P,Q)$, while the construction of those e-detectors does not rely on the knowledge of $(P,Q)$.

\section{Detection-delay lower bounds}
\label{sec:lowerb}
The matching lower bounds follow by change of measure.  We state the three
forms needed to interpret the upper bounds; proofs are in \Cref{app:proof-lower-bounds}.  Let $I^*=\KL(Q\|P)$ and
$L_\alpha=\log(1/\alpha)$.

\begin{theorem}[Pointwise PFA lower bound]
\label{thm:pointwise-lower}
Fix \(P,Q\in\mathcal P\) such that
$I^*:=D_{\mathrm{KL}}(Q\|P)\in(0,\infty),$
and fix \(\varepsilon\in(0,1)\). Let \(\alpha\downarrow0\), set
$L_\alpha:=\log(1/\alpha),$
and let \(\tau_\alpha\) be any PFA-\(\alpha\) stopping time. For any
deterministic changepoint \(T_\alpha\) and any integer \(d_\alpha\ge1\)
satisfying
$d_\alpha
    \le
    (1-\varepsilon)\frac{L_\alpha}{I^*},$
we have
\[
    \mathbb P_{P,Q}^{T_\alpha}
    \left(
        T_\alpha<\tau_\alpha\le T_\alpha+d_\alpha
    \right)
    \longrightarrow0.
\]
\end{theorem}

Let \(d_\alpha,T_{\max,\alpha}\in\mathbb N\), and define
\[
    N_\alpha
    :=
    \left\lfloor
        \frac{T_{\max,\alpha}}{d_\alpha+1}
    \right\rfloor,
    \qquad
    T_{j,\alpha}
    :=
    (j-1)(d_\alpha+1)+1,
    \quad
    j=1,\ldots,N_\alpha.
\]
Then the windows
\[
    (T_{j,\alpha},T_{j,\alpha}+d_\alpha],
    \qquad j=1,\ldots,N_\alpha,
\]
are pairwise disjoint and contained in
\(\{1,\ldots,T_{\max,\alpha}\}\).

Assume
\[
    d_\alpha\longrightarrow\infty,
    \qquad
    N_\alpha\longrightarrow\infty,
\]
and that, for some fixed \(\beta\in[0,1)\) and some
\(\rho_\alpha\to0\),
\begin{equation}
\label{eq:horizon-uniform-power}
    \inf_{1\le j\le N_\alpha}
    \mathbb P_{P,Q}^{T_{j,\alpha}}
    \left(
        T_{j,\alpha}<\tau_\alpha
        \le T_{j,\alpha}+d_\alpha
    \right)
    \ge
    1-\beta-\rho_\alpha.
\end{equation}

\begin{theorem}[Horizon-uniform PFA lower bound]
\label{thm:horizon-lower} Let \(\tau_\alpha\) be any PFA-\(\alpha\) stopping time.
Under the preceding conditions, for every fixed \(\eta>0\),
\begin{equation}
\label{eq:horizon-lower}
    d_\alpha
    \ge
    \frac{
        L_\alpha+\log N_\alpha+\log(1-\beta)+o(1)
    }{
        I^*+\eta
    }.
\end{equation}
In particular,
\[
    d_\alpha
    \ge
    \frac{
        L_\alpha+
        \log\left\lfloor
            T_{\max,\alpha}/(d_\alpha+1)
        \right\rfloor
        +O(1)
    }{
        I^*+\eta
    }.
\]
Moreover, the start-time price is asymptotically
\(\log(T_{\max,\alpha}/d_\alpha)\), and it equals
\(\log T_{\max,\alpha}\) to first order whenever
\[
    \log d_\alpha=o(\log T_{\max,\alpha}).
\]
\end{theorem}

\begin{theorem}[Early-change impossibility]
\label{thm:early-change-impossibility}
Let \(\tau_\alpha\) be any globally PFA-\(\alpha\) detector, and fix
\(P,Q\in\mathcal P\) with \(P\ll Q\). Then, for every \(T\in\mathbb N\),
\[
    \mathbb P_{P,Q}^{T}(\tau_\alpha<\infty)
    \le
    \min\left\{
        1,\,
        \frac{
            T D_{\mathrm{KL}}(P\|Q)+\log 2
        }{
            \log(1/\alpha)
        }
    \right\}.
\]
Consequently, for any sequence \(T_\alpha\),
$T_\alpha 
    =
    o\!\left(\log\frac1\alpha\right)$
implies
$\mathbb P_{P,Q}^{T_\alpha}(\tau_\alpha<\infty)\longrightarrow0,$ as $\alpha\downarrow0$.
\end{theorem}

An ARL analogue under the local finite-horizon guarantee
$\sup_R\Pinf R(\tau_A\le m)\le m/A$ is stated and proved in \Cref{app:proof-arl-early}.
Together, these results show that the upper bound established in the previous section is pointwise optimal and
that the PFA spending penalty has the correct horizon-uniform logarithmic
order.

\section{Experiments}
\label{sec:experiments-updated}

\subsection{Sub-Gaussian comparison with repeated CSs}
We first retain the original comparison with the repeated-CS detector of
\citet{shekhar2023reducing}.  Data are $N(0,1)$ up to $T=500$ and
$N(\delta,1)$ thereafter, with $A=1000$, $\rho^2=1$. For this model, the relevant information number is
$ I=D_{\mathrm{KL}}\bigl(N(\delta,1)\|N(0,1)\bigr)
    =
    \frac{\delta^2}{2},$
so the first-order delay benchmark is
$d_0:=\frac{\log A}{I}
    =
    \frac{2\log A}{\delta^2}.$
The empirical CADD is estimated from \(4000\) independent replications,
conditioning on the detector not having stopped before \(T\).

\begin{table}[H]
\centering
\small
\caption{Empirical CADD in the unit-variance Gaussian mean-shift experiment.}
\label{tab:subg-repeated-cs}
\begin{tabular}{cccc}
\toprule
$\delta$ & $\log(A)/I$ & Our e-detector & repeated CS \\
\midrule
$1.00$ & $13.82$ & $17.297$  &  $28.489$ \\
$1.25$ & $8.84$  & $11.297$ & $16.335$ \\
$1.50$ & $6.14$  & $8.129$  & $10.905$ \\
$2.00$ & $3.45$  & $5.177$ & $6.309$ \\
\bottomrule
\end{tabular}
\end{table}

We observe from \Cref{tab:subg-repeated-cs} that our ARL e-detector using \eqref{eq:D-def} is consistently and substantially faster. Its empirical delay is also substantially closer to
the information benchmark \(\log(A)/I\).

From the next experiment onwards, we only implement our e-detectors.

\subsection{Bounded observations: continuous universal portfolios}
Take $a=0.3$, $b=0.5$, and compare Bernoulli$(0.5)$,
$\operatorname{Beta}(10,10)$, and the two-point law on
$\{0.45,0.55\}$.  Their $I_{\rm bet}$ values are $0.087177$, $0.485201$, and
$0.505800$.  We compute the continuous Jeffreys UP exactly by
Gauss--Chebyshev quadrature, rather than by a fixed grid of bets. The results are reported in \Cref{tab:bounded-compact}. Panel A confirms both the optimal point-null rate and low-variance adaptivity.
Panel B uses $d_0=\log(A)/\kl(b,a)$ and
$T/d_0=(\log A)^2/5$.  At finite $T$, the relevant prediction is the smallest
$d$ such that
\[
 \inf_\theta\max\{d\kl(b,\theta),(T+d)\kl((Ta+db)/(T+d),\theta)\}
 \ge\log A.
\]
The ratio to this finite-prefix prediction decreases toward one.  Thus, the
moderate-threshold discrepancy is explained by the joint asymptotic and the
logarithmic mixture regret, not by a failure of the information constant.

\begin{table}[H]
\centering
\small
\caption{Bounded experiments.  Panel A gives exact point-null UP medians with
median$/\{\log(A)/I_{\rm bet}\}$ in parentheses ($4000$ replications).  Panel B
uses the conservative two-start composite statistic for the least-favorable
Bernoulli change and compares the median with the finite-prefix prediction
$d_{\rm FT}$ defined by the information-balance equation in the text.}
\label{tab:bounded-compact}
\begin{tabular}{rccc}
\toprule
\multicolumn{4}{c}{Panel A: exact point-null universal portfolio}\\
\midrule
$A$ & Bernoulli & Beta$(10,10)$ & two-point\\
\midrule
$10^{5}$  & $153\;(1.159)$ & $28\;(1.180)$ & $27\;(1.186)$\\
$10^{12}$ & $349\;(1.101)$ & $62\;(1.089)$ & $60\;(1.098)$\\
$10^{20}$ & $563\;(1.066)$ & $101\;(1.064)$& $97\;(1.065)$\\
\midrule
\multicolumn{4}{c}{Panel B: joint boundary/prefix asymptotic}\\
\midrule
$A$ & $T/d_0$ & $(d_0,d_{\rm FT})$ & median / $d_{\rm FT}$\\
\midrule
$10^{5}$  & $26.5$  & $(132.1,244.6)$ & $314.0/244.6=1.284$\\
$10^{8}$  & $67.9$  & $(211.3,290.2)$ & $330.5/290.2=1.139$\\
$10^{12}$ & $152.7$ & $(317.0,385.0)$ & $422.0/385.0=1.096$\\
$10^{20}$ & $424.2$ & $(528.3,589.7)$ & $635.5/589.7=1.078$\\
\bottomrule
\end{tabular}
\end{table}

\section{Conclusion}

We introduced a general framework for non-partitioned sequential change detection, in which both the pre-change and post-change distributions are unknown and belong to the same model class. Our construction combines point-null e-processes that start at each time point and then minimizes over candidate no-change laws, yielding detectors with finite-sample ARL or PFA guarantees. We developed explicit examples for sub-Gaussian means, bounded means, etc. Under suitable conditions, the resulting procedures attain first-order detection delays governed by the instance-specific information rate \(D_{\mathrm{KL}}(Q\|P)\). The lower bounds reveal a fundamental distinction between late changes, for which classical KL-type delay rates are attainable, and early changes, for which insufficient pre-change information can make reliable detection impossible.

\bibliographystyle{apalike}
\bibliography{ref}

\appendix

\section{Dependent Data Example: Change in a two-state Markov transition matrix}
\label{sec:markov-example}

The validity arguments underlying \Cref{thm:arl-control,thm:pfa-control}
do not require the observations to be independent.  They only require,
for every candidate no-change law and every restart time, a conditionally
valid delayed e-process.  We illustrate this point with a two-state Markov
chain.

Let \(X_0,X_1,\ldots\) take values in \(\{0,1\}\), and let
\[
    \mathcal F_t:=\sigma(X_0,\ldots,X_t).
\]
For \(p,q\in(0,1)\), write
\[
    \Markov(p,q)
\]
for the time-homogeneous two-state Markov chain satisfying
\[
    \mathbb P(X_t=1\mid\mathcal F_{t-1})
    =
    \begin{cases}
        p, & X_{t-1}=0,\\
        q, & X_{t-1}=1.
    \end{cases}
\]
Thus \(p\) is the probability of transitioning from state \(0\) to
state \(1\), while \(q\) is the probability of transitioning from
state \(1\) to state \(1\).

Fix \(\kappa\in(0,1/2)\), and consider the compact parameter space
\[
    \Theta_\kappa
    :=
    [\kappa,1-\kappa]^2.
\]
The no-change class is
\[
    \mathcal P_{\mathrm{MC}}
    :=
    \{
        \Markov(p,q):(p,q)\in\Theta_\kappa
    \}.
\]
All validity statements below hold conditionally on \(X_0\), and
therefore hold for any initial distribution of \(X_0\).

For a candidate parameter \(r=(p,q)\in\Theta_\kappa\), define
\[
    \rho_r(x)
    :=
    \begin{cases}
        p, & x=0,\\
        q, & x=1.
    \end{cases}
\]
For \(1\le s\le t\), let
\[
    N_{xy}^{s:t}
    :=
    \sum_{i=s}^t
    \mathbbm 1\{X_{i-1}=x,X_i=y\},
    \qquad x,y\in\{0,1\},
\]
and let
\[
    N_x^{s:t}:=N_{x0}^{s:t}+N_{x1}^{s:t}.
\]
The transition likelihood under the candidate parameter \(r=(p,q)\)
is
\[
    L_{s:t}^{p,q}
    :=
    p^{N_{01}^{s:t}}
    (1-p)^{N_{00}^{s:t}}
    q^{N_{11}^{s:t}}
    (1-q)^{N_{10}^{s:t}}.
\]

We interpret a changepoint \(T\) as a change in the transition kernel:
transitions \(1,\ldots,T\) are generated by the pre-change kernel,
while transitions \(T+1,T+2,\ldots\) are generated by the post-change
kernel.

\subsection{A restartable Markov e-process}

Let
\[
    \Pi_J(du)
    :=
    \frac{du}{\pi\sqrt{u(1-u)}},
    \qquad u\in(0,1),
\]
be the Jeffreys, or arcsine, distribution.  Independently mixing the
two transition probabilities against \(\Pi_J\), define
\begin{align}
    M_{s:t}^{p,q}
    :={}&
    \frac{
        B\left(
            N_{01}^{s:t}+\frac12,
            N_{00}^{s:t}+\frac12
        \right)
    }{
        B\left(\frac12,\frac12\right)
        p^{N_{01}^{s:t}}
        (1-p)^{N_{00}^{s:t}}
    }
    \nonumber\\
    &\times
    \frac{
        B\left(
            N_{11}^{s:t}+\frac12,
            N_{10}^{s:t}+\frac12
        \right)
    }{
        B\left(\frac12,\frac12\right)
        q^{N_{11}^{s:t}}
        (1-q)^{N_{10}^{s:t}}
    },
    \label{eq:markov-jeffreys-eprocess}
\end{align}
with
\[
    M_{s:s-1}^{p,q}:=1.
\]
Here \(B(\cdot,\cdot)\) denotes the beta function.

Equivalently, define the predictable Jeffreys estimators
\[
    \widehat\rho_{s,i-1}(x)
    :=
    \frac{
        N_{x1}^{s:i-1}+\frac12
    }{
        N_x^{s:i-1}+1
    },
    \qquad x\in\{0,1\},
\]
where all counts are zero when \(i=s\).  Then
\begin{equation}
\label{eq:markov-eprocess-predictive}
    M_{s:t}^{p,q}
    =
    \prod_{i=s}^t
    \frac{
        \widehat\rho_{s,i-1}(X_{i-1})^{X_i}
        \{1-\widehat\rho_{s,i-1}(X_{i-1})\}^{1-X_i}
    }{
        \rho_{(p,q)}(X_{i-1})^{X_i}
        \{1-\rho_{(p,q)}(X_{i-1})\}^{1-X_i}
    }.
\end{equation}

For the expectation bound below, it is convenient to add a fixed
baseline and define
\begin{equation}
\label{eq:markov-baseline-eprocess}
    \widetilde M_{s:t}^{p,q}
    :=
    \frac{1+M_{s:t}^{p,q}}{2}.
\end{equation}
This modification preserves e-validity and all positive first-order
growth rates, while ensuring
\[
    \widetilde M_{s:t}^{p,q}\ge\frac12
\]
on every sample path.

\begin{proposition}[Validity and entropy-rate growth]
\label{prop:markov-validity-growth}
For every \((p,q)\in\Theta_\kappa\) and every start time \(s\),
$\{\widetilde M_{s:t}^{p,q}\}_{t\ge s-1}$
is an \(s\)-delay e-process under \(\Markov(p,q)\).  In fact,
\(\{M_{s:t}^{p,q}\}_{t\ge s-1}\) is a test martingale:
\[
    \mathbb E_{p,q}
    \left[
        M_{s:t}^{p,q}
        \mid\mathcal F_{t-1}
    \right]
    =
    M_{s:t-1}^{p,q},
    \qquad t\ge s.
\]
For \(r=(p,q)\in\Theta_\kappa\), let
$  \varpi_r(0)
    =
    \frac{1-q}{1-q+p},
    \varpi_r(1)
    =
    \frac{p}{1-q+p}$
denote its stationary distribution.  For
$r_1=(p_1,q_1),
    r_0=(p_0,q_0),$
define the Markov relative-entropy rate
\begin{equation}
\label{eq:markov-information-rate}
    \mathcal I(r_1\|r_0)
    :=
    \varpi_{r_1}(0)
    \kl(p_1,p_0)
    +
    \varpi_{r_1}(1)
    \kl(q_1,q_0),
\end{equation}
where
$ \kl(x,y)
    :=
    x\log\frac{x}{y}
    +(1-x)\log\frac{1-x}{1-y}.$
If the observations are generated by \(\Markov(r_1)\), then, for
every \(r_0\in\Theta_\kappa\),
\[
    \frac1n
    \log\widetilde M_{1:n}^{r_0}
    \longrightarrow
    \mathcal I(r_1\|r_0)
    \qquad\text{almost surely}.
\]
Moreover, this convergence is uniform over \(r_0\in\Theta_\kappa\):
\begin{equation}
\label{eq:markov-uniform-growth}
    \sup_{r_0\in\Theta_\kappa}
    \left|
        \frac1n\log\widetilde M_{1:n}^{r_0}
        -
        \mathcal I(r_1\|r_0)
    \right|
    \longrightarrow0
    \qquad\text{almost surely}.
\end{equation}
\end{proposition}

\begin{proof}
Fix \(r=(p,q)\in\Theta_\kappa\).  From
\eqref{eq:markov-eprocess-predictive},
\[
\begin{aligned}
&\mathbb E_{p,q}
\left[
    \frac{
        \widehat\rho_{s,t-1}(X_{t-1})^{X_t}
        \{1-\widehat\rho_{s,t-1}(X_{t-1})\}^{1-X_t}
    }{
        \rho_r(X_{t-1})^{X_t}
        \{1-\rho_r(X_{t-1})\}^{1-X_t}
    }
    \,\middle|\,
    \mathcal F_{t-1}
\right]
\\
&\qquad=
\sum_{y\in\{0,1\}}
\widehat\rho_{s,t-1}(X_{t-1})^y
\{1-\widehat\rho_{s,t-1}(X_{t-1})\}^{1-y}
=1.
\end{aligned}
\]
Thus \(M_{s:t}^{p,q}\) is a nonnegative test martingale.  Since
\(\widetilde M=(1+M)/2\), the baseline-stabilized process is also an
\(s\)-delay e-process.

We next prove the growth assertion.  Put
\[
    \widehat p_n
    :=
    \frac{N_{01}^{1:n}}{N_0^{1:n}},
    \qquad
    \widehat q_n
    :=
    \frac{N_{11}^{1:n}}{N_1^{1:n}},
\]
whenever the corresponding denominator is nonzero.  Terms associated
with an unvisited state are interpreted as zero.  Standard bounds for
the Jeffreys beta mixture give a constant \(C<\infty\), independent
of the data and of \(r_0=(p_0,q_0)\in\Theta_\kappa\), such that
\begin{align}
&N_0^{1:n}\kl(\widehat p_n,p_0)
+
N_1^{1:n}\kl(\widehat q_n,q_0)
-C\log(n+1)
\nonumber\\
&\qquad\le
\log M_{1:n}^{p_0,q_0}
\le
N_0^{1:n}\kl(\widehat p_n,p_0)
+
N_1^{1:n}\kl(\widehat q_n,q_0).
\label{eq:markov-mixture-regret}
\end{align}
Under \(\Markov(r_1)\), the Markov ergodic theorem gives
\[
    \frac{N_0^{1:n}}n\to\varpi_{r_1}(0),
    \qquad
    \frac{N_1^{1:n}}n\to\varpi_{r_1}(1),
\]
and
\[
    \widehat p_n\to p_1,
    \qquad
    \widehat q_n\to q_1
\]
almost surely.  Since \(\Theta_\kappa\) is compact and binary relative
entropy is uniformly continuous on
\([0,1]\times[\kappa,1-\kappa]\), division of
\eqref{eq:markov-mixture-regret} by \(n\) gives
\[
    \sup_{r_0\in\Theta_\kappa}
    \left|
        \frac1n\log M_{1:n}^{r_0}
        -
        \mathcal I(r_1\|r_0)
    \right|
    \longrightarrow0
\]
almost surely.  Replacing \(M\) by \((1+M)/2\) does not change the
normalized logarithm, which proves
\eqref{eq:markov-uniform-growth}.
\end{proof}

\subsection{ARL- and PFA-valid Markov e-detectors}

Using the point-null processes in
\eqref{eq:markov-baseline-eprocess}, define
\begin{equation}
\label{eq:markov-arl-detector}
    D_t^{\mathrm{MC},\mathrm{ARL}}
    :=
    \inf_{(p,q)\in\Theta_\kappa}
    \sum_{s=1}^t
    \widetilde M_{s:t}^{p,q},
    \qquad
    \tau_A^{\mathrm{MC},\mathrm{ARL}}
    :=
    \inf\left\{
        t\ge1:
        D_t^{\mathrm{MC},\mathrm{ARL}}\ge A
    \right\}.
\end{equation}
For deterministic weights \(\pi_s>0\) satisfying
\[
    \sum_{s=1}^\infty\pi_s\le1,
\]
define
\begin{equation}
\label{eq:markov-pfa-detector}
    D_t^{\mathrm{MC},\mathrm{PFA}}
    :=
    \inf_{(p,q)\in\Theta_\kappa}
    \sum_{s=1}^t
    \pi_s\widetilde M_{s:t}^{p,q},
    \qquad
    \tau_\alpha^{\mathrm{MC},\mathrm{PFA}}
    :=
    \inf\left\{
        t\ge1:
        D_t^{\mathrm{MC},\mathrm{PFA}}
        \ge\frac1\alpha
    \right\}.
\end{equation}

\begin{corollary}[Finite-sample false-alarm validity]
\label{cor:markov-false-alarm}
For every \((p,q)\in\Theta_\kappa\),
\[
    \mathbb E_{p,q}^\infty
    \left[
        \tau_A^{\mathrm{MC},\mathrm{ARL}}
    \right]
    \ge A,
\]
and, for every \(m\in\mathbb N\),
\[
    \mathbb P_{p,q}^\infty
    \left(
        \tau_A^{\mathrm{MC},\mathrm{ARL}}\le m
    \right)
    \le\frac mA.
\]
Moreover,
\[
    \sup_{(p,q)\in\Theta_\kappa}
    \mathbb P_{p,q}^\infty
    \left(
        \tau_\alpha^{\mathrm{MC},\mathrm{PFA}}<\infty
    \right)
    \le\alpha.
\]
\end{corollary}

\begin{proof}
The proof is identical to those of
\Cref{thm:arl-control,thm:pfa-control}.  Those arguments use only the
restartable conditional e-process property and do not use independence.
\end{proof}

\subsection{Detection-delay guarantees}

Suppose the pre-change transition kernel is
\[
    P=\Markov(p_0,q_0),
\]
and the post-change transition kernel is
\[
    Q=\Markov(p_1,q_1),
\]
where both parameter pairs belong to \(\Theta_\kappa\) and
\((p_0,q_0)\ne(p_1,q_1)\).  Define
\begin{equation}
\label{eq:markov-Istar}
    I_{\mathrm{MC}}^*
    :=
    \mathcal I\bigl((p_1,q_1)\|(p_0,q_0)\bigr)
    >0.
\end{equation}

The following theorem is the Markov analogue of the previous
sub-Gaussian, bounded-mean, and Gaussian delay results.

For a transition-frequency array
\(\gamma=(\gamma_{xy})_{x,y\in\{0,1\}}\), write
\[
    \gamma_x:=\gamma_{x0}+\gamma_{x1}.
\]
Whenever \(\gamma_0,\gamma_1>0\), define
\[
    \Psi(\gamma,r)
    :=
    \sum_{x=0}^1
    \gamma_x
    \kl\left(
        \frac{\gamma_{x1}}{\gamma_x},
        \rho_r(x)
    \right).
\]
The Jeffreys-mixture regret bound \eqref{eq:markov-mixture-regret}
and \(\widetilde M\ge M/2\) imply
\[
    \log\widetilde M_{s:s+n-1}^{r}
    \ge
    n\Psi(\widehat\Gamma_{s:s+n-1},r)
    -C\log(n+1)-\log2.
\]

\begin{theorem}[Delay for a two-state Markov change]
\label{thm:markov-delay}
Let \(P=\Markov(p_0,q_0)\) and \(Q=\Markov(p_1,q_1)\) be distinct
members of \(\mathcal P_{\mathrm{MC}}\).
For the ARL detector, let \(A\to\infty\), put \(L=\log A\), and suppose
the changepoint \(T_A\) satisfies
$\frac{T_A}{L}\longrightarrow\infty,
    \log T_A=o(L).$
Then
\[
    \tau_A^{\mathrm{MC},\mathrm{ARL}}-T_A
    \le
    (1+o_{\mathbb P}(1))
    \frac{L}{I_{\mathrm{MC}}^*}.
\]
Moreover,
\[
    \operatorname{CADD}_{P,Q}^{T_A}
    \left(
        \tau_A^{\mathrm{MC},\mathrm{ARL}}
    \right)
    \le
    (1+o(1))
    \frac{L}{I_{\mathrm{MC}}^*}.
\]
For the PFA detector, let
$L_\alpha:=\log(1/\alpha),
    B_\alpha
    :=
    L_\alpha-\log\pi_{T_\alpha+1}.$
If $   \frac{T_\alpha}{B_\alpha}\longrightarrow\infty,$
then
\[
    \tau_\alpha^{\mathrm{MC},\mathrm{PFA}}-T_\alpha
    \le
    (1+o_{\mathbb P}(1))
    \frac{B_\alpha}{I_{\mathrm{MC}}^*}.
\]
For the canonical weights \eqref{eq:pi-weights},
this becomes
\[
    \tau_\alpha^{\mathrm{MC},\mathrm{PFA}}-T_\alpha
    \le
    (1+o_{\mathbb P}(1))
    \frac{
        \log(1/\alpha)
        +\log T_\alpha
        +2\log\log(eT_\alpha)
        +O(1)
    }{
        I_{\mathrm{MC}}^*
    }.
\]
\end{theorem}

\begin{proof}[Proof of \Cref{thm:markov-delay}]
Fix \(\varepsilon>0\). For generic weights \(w_s\), put
\[
    B:=L-\log w_{T+1},
    \qquad
    d
    :=
    \left\lceil
        (1+\varepsilon)
        \frac{B}{I_{\mathrm{MC}}^*}
    \right\rceil.
\]
Under the stated assumptions,
\[
    d\to\infty,
    \qquad
    \frac dT\to0.
\]

Write
\[
    r_0=(p_0,q_0),
    \qquad
    r_1=(p_1,q_1),
    \qquad
    I_{\mathrm{MC}}^*
    =
    \mathcal I(r_1\|r_0).
\]
Choose \(\eta>0\) sufficiently small that
\begin{equation}
\label{eq:markov-eta-choice}
    (1+\varepsilon)
    \frac{I_{\mathrm{MC}}^*-3\eta}
         {I_{\mathrm{MC}}^*}
    >1.
\end{equation}
By continuity, there is a compact relative neighborhood
\(U\subset\Theta_\kappa\) of \(r_0\) such that
\[
    \inf_{r\in U}\mathcal I(r_1\|r)
    \ge I_{\mathrm{MC}}^*-\eta.
\]
Also,
\[
    c_U
    :=
    \inf_{r\notin U}\mathcal I(r_0\|r)
    >0.
\]

Let \(\widehat\Gamma_{\mathrm{pre}}\) and
\(\widehat\Gamma_{\mathrm{post},d}\) be the empirical transition
frequencies in the first \(T\) and next \(d\) transitions,
respectively. Since \(\Psi\) is uniformly continuous on the relevant
compact sets, there is \(\delta>0\) such that
\begin{align}
    \|\widehat\Gamma_{\mathrm{post},d}-\Gamma_{r_1}\|_\infty
    \le\delta
    &\implies
    \inf_{r\in U}
    \Psi(\widehat\Gamma_{\mathrm{post},d},r)
    \ge I_{\mathrm{MC}}^*-2\eta,
    \label{eq:markov-post-good}\\
    \|\gamma-\Gamma_{r_0}\|_\infty\le\delta
    &\implies
    \inf_{r\notin U}\Psi(\gamma,r)
    \ge \frac{c_U}{2}.
    \label{eq:markov-pre-good}
\end{align}

On the event
\[
    G_d
    :=
    \left\{
        \|\widehat\Gamma_{\mathrm{post},d}
          -\Gamma_{r_1}\|_\infty\le\delta
    \right\}
    \cap
    \left\{
        \|\widehat\Gamma_{\mathrm{pre}}
          -\Gamma_{r_0}\|_\infty\le\frac{\delta}{2}
    \right\},
\]
the regret inequality and \eqref{eq:markov-post-good} give, for all
sufficiently large \(L\),
\[
    \inf_{r\in U}
    \log\widetilde M_{T+1:T+d}^{r}
    \ge
    d(I_{\mathrm{MC}}^*-3\eta).
\]
By \eqref{eq:markov-eta-choice},
\[
    \inf_{r\in U}
    \log\left\{
        w_{T+1}\widetilde M_{T+1:T+d}^{r}
    \right\}
    \ge L.
\]

Moreover,
\[
    \widehat\Gamma_{1:T+d}
    =
    \frac{T}{T+d}\widehat\Gamma_{\mathrm{pre}}
    +
    \frac{d}{T+d}\widehat\Gamma_{\mathrm{post},d}.
\]
Since \(d/T\to0\), on \(G_d\),
\[
    \|\widehat\Gamma_{1:T+d}-\Gamma_{r_0}\|_\infty
    \le\delta
\]
eventually. Hence \eqref{eq:markov-pre-good} gives
\[
    \inf_{r\notin U}
    \log\widetilde M_{1:T+d}^{r}
    \ge
    \frac{c_U}{3}(T+d)
\]
for all sufficiently large \(T\), where the logarithmic regret has
again been absorbed. Since
\[
    L-\log w_1=o(T),
\]
we obtain
\[
    \inf_{r\notin U}
    \log\left\{
        w_1\widetilde M_{1:T+d}^{r}
    \right\}
    \ge L.
\]
Thus \(G_d\) implies \(\tau^w\le T+d\).

By \Cref{lem:markov-transition-concentration}, uniformly over the
state at the changepoint,
\begin{align}
    \mathbb P_{P,Q}^{T}(\tau^w>T+d)
    &\le
    C_\kappa e^{-c_\kappa\delta^2d}
    +
    C_\kappa e^{-c_\kappa\delta^2T}
    \nonumber\\
    &\le
    C e^{-c_\varepsilon L},
    \label{eq:markov-short-tail}
\end{align}
because \(d\ge c_\varepsilon' L\) and \(T/L\to\infty\). This proves
the high-probability delay bound.

We now specialize to the ARL detector and prove the second,
macroscopic tail bound. Put
\[
    H:=\lfloor T/2\rfloor.
\]
For \(\lambda\in[0,1/2]\), define
\[
    \Gamma_\lambda
    :=
    \frac{\Gamma_{r_0}+\lambda\Gamma_{r_1}}
         {1+\lambda},
\]
and let \(\bar r_\lambda\) be the transition kernel obtained by
normalizing the two rows of \(\Gamma_\lambda\). The set
\[
    K:=\{\bar r_\lambda:0\le\lambda\le1/2\}
\]
is compact and does not contain \(r_1\). Indeed, each coordinate of
\(\bar r_\lambda\) is a strict weighted average of the corresponding
coordinates of \(r_0\) and \(r_1\) whenever those coordinates differ.

Choose a relatively open set \(V\subset\Theta_\kappa\) such that
\[
    K\subset V,
    \qquad
    r_1\notin\overline V.
\]
Then
\[
    a_V
    :=
    \inf_{r\in\overline V}
    \mathcal I(r_1\|r)
    >0.
\]
Furthermore,
\[
    b_V
    :=
    \inf_{\substack{0\le\lambda\le1/2\\r\notin V}}
    \Psi(\Gamma_\lambda,r)
    >0,
\]
because \(\Psi(\Gamma_\lambda,r)=0\) only when
\(r=\bar r_\lambda\in V\).

If both the pre-change transition frequencies and the \(H\)
post-change transition frequencies are sufficiently close to
\(\Gamma_{r_0}\) and \(\Gamma_{r_1}\), respectively, uniform
continuity and the regret bound imply
\[
    \inf_{r\in V}
    \log\widetilde M_{T+1:T+H}^{r}
    \ge
    \frac{a_VH}{2},
\]
and
\[
    \inf_{r\notin V}
    \log\widetilde M_{1:T+H}^{r}
    \ge
    \frac{b_V(T+H)}{2}
\]
for all sufficiently large \(T\). Since \(T/L\to\infty\), both
right-hand sides exceed \(L\). Thus the ARL detector stops by
\(T+H\) on this event. Another application of
\Cref{lem:markov-transition-concentration} therefore gives
\begin{equation}
\label{eq:markov-long-tail}
    \mathbb P_{P,Q}^{T}
    \left(
        \tau_A^{\mathrm{MC},\mathrm{ARL}}>T+H
    \right)
    \le
    C e^{-cT}.
\end{equation}

Finally, the baseline
\[
    \widetilde M_{s:t}^{r}\ge\frac12
\]
implies
\[
    D_t^{\mathrm{MC},\mathrm{ARL}}\ge\frac t2,
    \qquad
    \tau_A^{\mathrm{MC},\mathrm{ARL}}\le\lceil2A\rceil
\]
pathwise. Writing
\[
    Y
    :=
    \left(
        \tau_A^{\mathrm{MC},\mathrm{ARL}}-T
    \right)^+,
\]
we obtain from \eqref{eq:markov-short-tail} and
\eqref{eq:markov-long-tail}
\begin{align*}
    \mathbb E_{P,Q}^{T}[Y]
    &\le
    d
    +
    T\,\mathbb P_{P,Q}^{T}(Y>d)
    +
    2A\,\mathbb P_{P,Q}^{T}(Y>H)\\
    &\le
    d
    +
    CTe^{-c_\varepsilon L}
    +
    2CAe^{-cT}.
\end{align*}
Since
\[
    \log T=o(L),
    \qquad
    \frac TL\to\infty,
    \qquad
    A=e^L,
\]
the last two terms converge to zero. Hence, for every fixed
\(\varepsilon>0\),
\[
    \mathbb E_{P,Q}^{T}[Y]
    \le
    \left\lceil
        (1+\varepsilon)
        \frac{L}{I_{\mathrm{MC}}^*}
    \right\rceil
    +o(1).
\]
The usual fixed-\(\varepsilon\) \(\limsup\) argument gives
\[
    \mathbb E_{P,Q}^{T}[Y]
    \le
    (1+o(1))
    \frac{L}{I_{\mathrm{MC}}^*}.
\]

Finally, the local false-alarm guarantee gives
\[
    \mathbb P_{P,Q}^{T}
    \left(
        \tau_A^{\mathrm{MC},\mathrm{ARL}}\le T
    \right)
    =
    \mathbb P_{P}^{\infty}
    \left(
        \tau_A^{\mathrm{MC},\mathrm{ARL}}\le T
    \right)
    \le
    \frac TA
    =
    e^{\log T-L}
    \longrightarrow0.
\]
Therefore,
\[
    \operatorname{CADD}_{P,Q}^{T}
    \left(
        \tau_A^{\mathrm{MC},\mathrm{ARL}}
    \right)
    =
    \frac{
        \mathbb E_{P,Q}^{T}[Y]
    }{
        \mathbb P_{P,Q}^{T}
        (\tau_A^{\mathrm{MC},\mathrm{ARL}}>T)
    }
    \le
    (1+o(1))
    \frac{L}{I_{\mathrm{MC}}^*}.
\]
\end{proof}

\paragraph{Computation.}
For every start \(s\), the four transition counts
\[
    N_{00}^{s:t},\quad N_{01}^{s:t},\quad
    N_{10}^{s:t},\quad N_{11}^{s:t}
\]
can be updated in constant time after each new observation.  At a
fixed time \(t\), the objective
\[
    (p,q)
    \longmapsto
    \sum_{s=1}^t
    w_s\widetilde M_{s:t}^{p,q}
\]
is smooth and convex on the compact rectangle
\(\Theta_\kappa\), because each summand is a positive constant times
the exponential of a convex negative transition log-likelihood.
Thus the infimum is a two-dimensional convex optimization problem.
The exact all-start implementation costs \(O(t)\) per objective,
gradient, or Hessian evaluation and \(O(t)\) memory at time \(t\).
As in the preceding examples, pruning or geometric start grids can
reduce the all-start computational cost while preserving validity
when the resulting statistic is a pointwise lower bound on the full
detector.

\section{Proofs of the main results}\label{app:main-proofs}

\subsection{Finite-sample ARL and PFA validity}\label{app:proof-validity}

\begin{proof}[Proof of \Cref{thm:arl-control}]
Fix $R_0\in\calP$ and define the oracle SR statistic
\[
  S_t^{R_0}=\sum_{s=1}^t  M_{s:t}^{R_0}.
\]
Let $\tau$ be any almost surely finite stopping time.  By conditional e-process validity after time $s-1$,
\[
  \Einf{R_0}\!\left[\ind\{\tau\ge s\}M_{s:\tau}^{R_0}\right]
  \le
  \Pinf{R_0}(\tau\ge s).
\]
Summing over $s\ge1$ gives
\[
  \Einf{R_0}S_\tau^{R_0}
  \le
  \sum_{s\ge1}\Pinf{R_0}(\tau\ge s)
  =
  \Einf{R_0}\tau,
\]
proving the e-detector claim.
Apply this with $\tau=\tau_A^{\mathrm{ARL}}\wedge m$.  Since
\[
  D_t^{\mathrm{ARL}}\le S_t^{R_0},
\]
on the event $\{\tau_A^{\mathrm{ARL}}\le m\}$ we have
\[
  S_{\tau_A^{\mathrm{ARL}}\wedge m}^{R_0}
  =S_{\tau_A^{\mathrm{ARL}}}^{R_0}
  \ge D_{\tau_A^{\mathrm{ARL}}}^{\mathrm{ARL}}
  \ge A.
\]
Therefore
\[
  A\Pinf{R_0}(\tau_A^{\mathrm{ARL}}\le m)
  \le
  \Einf{R_0}S_{\tau_A^{\mathrm{ARL}}\wedge m}^{R_0}
  \le
  \Einf{R_0}(\tau_A^{\mathrm{ARL}}\wedge m)
  \le m.
\]
This proves the finite-horizon bound.  Letting $m\to\infty$ yields
\[
  A\Pinf{R_0}(\tau_A^{\mathrm{ARL}}<\infty)
  \le
  \Einf{R_0}\tau_A^{\mathrm{ARL}}.
\]
If $\Pinf{R_0}(\tau_A^{\mathrm{ARL}}<\infty)<1$, then $\Einf{R_0}\tau_A^{\mathrm{ARL}}=\infty$.  Otherwise the above inequality gives $$\Einf{R_0}\tau_A^{\mathrm{ARL}}\ge A.$$
\end{proof}

\begin{proof}[Proof of \Cref{thm:pfa-control}]
Fix $R_0\in\calP$.  For each $s$, define
\[
  Z_{s,t}^{R_0}=\begin{cases}
  1, & t<s,\\
  M_{s:t}^{R_0}, & t\ge s.
  \end{cases}
\]
By restartable validity, $(Z_{s,t}^{R_0})_{t\ge0}$ is an e-process under $\Pinf{R_0}$.  Hence
\[
  E_t^{R_0}
  =1-\sum_{s=1}^\infty\pi_s+
  \sum_{s=1}^\infty\pi_s Z_{s,t}^{R_0}
\]
is an e-process under $\Pinf{R_0}$.  Also,
\[
  D_t^{\mathrm{PFA}}
  \le
  \sum_{s=1}^t \pi_sM_{s:t}^{R_0}
  \le
  E_t^{R_0},
\]
hence is also an e-process as claimed.
Therefore, by Ville's inequality,
\[
  \Pinf{R_0}(\tau_\alpha^{\mathrm{PFA}}<\infty)
  \le
  \Pinf{R_0}\left(\sup_{t\ge0}E_t^{R_0}\ge\frac1\alpha\right)
  \le \alpha.
\]
Taking the supremum over $R_0$ proves the claim.
\end{proof}

\subsection{Generic late-change upper bounds}\label{app:proof-generic-upper}

\begin{proof}[Proof of \Cref{thm:delay-upper-generic}]
Define $t_\alpha=T_\alpha+d_\alpha$. Since we have assumed that
$\calP$ has a countable local REGROW witness basis, by
\Cref{thm:witness-implies-growth}, for the given
$\eta_{\mathrm{in}}\in(0,I^*)$, there exists
$B_{\eta_{\mathrm{in}}}\in\mathscr B$ such that
$P\in B_{\eta_{\mathrm{in}}}$ and,
we have
\[
  \Phi_{\overline B_{j_*}}(Q)
  \ge I^*-\eta_{\mathrm{in}},
  \qquad
  c_{\eta_{\mathrm{in}}}
  \defeq
  \Phi_{B_{j_*}^c}(P)>0.
\]
Moreover, for the given
$\eta_{\mathrm{out}}\in(0,c_{\eta_{\mathrm{in}}})$ and for each
$s$ and $R\in\calP$, there exists an $s$-delay e-process
$\{\overline M_{s:t}^R\}_{t\geq s-1}$ that is nondecreasing in $t$
and satisfies \eqref{eq:pre-growth} and \eqref{eq:post-growth} with
this $B_{j_*}$, $c_{\eta_{\mathrm{in}}}$,
$\eta_{\mathrm{out}}$, and $\eta_{\mathrm{in}}$.

On the pre-change block $1:T_\alpha$, the law is $P^{T_\alpha}$.
By \eqref{eq:pre-growth},
\begin{equation}
\label{eq:event-pre}
  \inf_{R\notin B_{j_*}}
  \log \overline M_{1:T_\alpha}^R
  \ge
  T_\alpha
  \bigl(c_{\eta_{\mathrm{in}}}-\eta_{\mathrm{out}}\bigr)
\end{equation}
holds with probability tending to one.  Since
$\overline M_{1:t}^R$ is nondecreasing in $t$,
\[
  \inf_{R\notin B_{j_*}}
  \log \overline M_{1:t_\alpha}^R
  \ge
  \inf_{R\notin B_{j_*}}
  \log \overline M_{1:T_\alpha}^R.
\]
Thus, on the event \eqref{eq:event-pre}, for every
$R\notin B_{j_*}$,
\[
  \log\left(w_1\overline M_{1:t_\alpha}^R\right)
  \ge
  \log w_1+
  T_\alpha
  \bigl(c_{\eta_{\mathrm{in}}}-\eta_{\mathrm{out}}\bigr)
  \ge L_\alpha
  =\log A_\alpha
\]
by \eqref{eq:anchor-condition}.  Hence, on the event
\eqref{eq:event-pre}, for every
$R\notin B_{j_*}$,
\begin{equation}
\label{eq:outside-large}
  \sum_{s=1}^{t_\alpha}
  w_s\overline M_{s:t_\alpha}^R
  \ge A_\alpha.
\end{equation}

On the post-change block
$T_\alpha+1:T_\alpha+d_\alpha$, the law is $Q^{d_\alpha}$.
By \eqref{eq:post-growth},
\begin{equation}
\label{eq:event-post}
  \inf_{R\in B_{j_*}}
  \log
  \overline M_{T_\alpha+1:T_\alpha+d_\alpha}^R
  \ge
  d_\alpha(I^*-\eta_{\mathrm{in}})
\end{equation}
holds with probability tending to one.  On this event, for every
$R\in B_{j_*}$,
\[
  \log\left(
    w_{T_\alpha+1}
    \overline M_{T_\alpha+1:T_\alpha+d_\alpha}^R
  \right)
  \ge
  \log w_{T_\alpha+1}
  +d_\alpha(I^*-\eta_{\mathrm{in}})
  \ge L_\alpha
\]
by \eqref{eq:post-condition}.  Therefore, on the event
\eqref{eq:event-post}, for every
$R\in B_{j_*}$,
\begin{equation}
\label{eq:inside-large}
  \sum_{s=1}^{t_\alpha}
  w_s\overline M_{s:t_\alpha}^R
  \ge A_\alpha.
\end{equation}

With probability tending to one, both
\eqref{eq:outside-large} and \eqref{eq:inside-large} hold.  On that
event, for every $R\in\calP$,
\[
  \sum_{s=1}^{t_\alpha}
  w_s\overline M_{s:t_\alpha}^R
  \ge A_\alpha.
\]
Taking the infimum over $R$ gives
$D_{t_\alpha}^w\ge A_\alpha$, hence
$\tau_\alpha^w\le t_\alpha$.  This proves the result.
\end{proof}

\begin{proof}[Proof of \Cref{cor:arl-delay}]
Fix $\eps>0$ and apply \Cref{thm:delay-upper-generic} with
$w_s\equiv1$. Then
\[
    B_\alpha=L,
    \qquad
    d_A=\left\lceil(1+\eps)\frac{L}{I^*}\right\rceil .
\]
Choose
\[
    \eta_{\mathrm{in}}
    :=
    \frac{\eps}{2(1+\eps)}I^*
    \in(0,I^*),
\]
let $c_{\eta_{\mathrm{in}}}>0$ be furnished by
\Cref{thm:witness-implies-growth}, and set
\[
    \eta_{\mathrm{out}}
    :=
    \frac{c_{\eta_{\mathrm{in}}}}{2}.
\]
Since $T_A/L\to\infty$,
\[
    T_A
    \bigl(c_{\eta_{\mathrm{in}}}-\eta_{\mathrm{out}}\bigr)
    +\log w_1
    =
    \frac{c_{\eta_{\mathrm{in}}}}{2}T_A
    \ge L
\]
eventually. Moreover,
\begin{align*}
    d_A(I^*-\eta_{\mathrm{in}})
    +\log w_{T_A+1}
    &\ge
    (1+\eps)\frac{L}{I^*}
    (I^*-\eta_{\mathrm{in}})\\
    &=
    L(1+\eps)
    \left(1-\frac{\eps}{2(1+\eps)}\right)\\
    &=
    L\left(1+\frac{\eps}{2}\right)
    \ge L.
\end{align*}
Thus both \eqref{eq:anchor-condition} and
\eqref{eq:post-condition} hold. Hence
\Cref{thm:delay-upper-generic} gives
\[
    \mathbb P_{P,Q}^{T_A}
    \bigl(\tau_A^{\mathrm{ARL}}\le T_A+d_A\bigr)
    \to1.
\]
Since
\[
    d_A
    \le
    (1+\eps)\frac{L}{I^*}+1,
\]
the first claim follows.

For the second claim, the finite-horizon false-alarm bound in
\Cref{thm:arl-control} gives
\[
    \Pinf{P}(\tau_A^{\mathrm{ARL}}\le T_A)
    \le \frac{T_A}{A}\to0.
\]
Before time $T_A$, the change law $\Pchg{P}{Q}{T_A}$ agrees with
$\Pinf{P}$, so
\[
    \Pchg{P}{Q}{T_A}
    (\tau_A^{\mathrm{ARL}}>T_A)\to1.
\]
Conditioning the first conclusion on this event therefore preserves
the probability-one limit, proving the conditional statement.
\end{proof}

\begin{proof}[Proof of \Cref{cor:pfa-delay}]
Apply \Cref{thm:delay-upper-generic} with $w_s=\pi_s$, so
\[
    B_\alpha=L_\alpha-\log\pi_{T_\alpha+1},
    \qquad
    d_\alpha=
    \left\lceil
        (1+\eps)\frac{B_\alpha}{I^*}
    \right\rceil.
\]
Choose
\[
    \eta_{\mathrm{in}}
    =\frac{\eps I^*}{2(1+\eps)},
    \qquad
    \eta_{\mathrm{out}}
    =\frac{c_{\eta_{\mathrm{in}}}}{2}.
\]
Since $T_\alpha/L_\alpha\to\infty$ and $\log\pi_1$ is constant,
\[
    T_\alpha
    (c_{\eta_{\mathrm{in}}}-\eta_{\mathrm{out}})
    +\log\pi_1
    =
    \frac{c_{\eta_{\mathrm{in}}}}{2}T_\alpha+\log\pi_1
    \ge L_\alpha
\]
eventually. Moreover,
\begin{align*}
    d_\alpha(I^*-\eta_{\mathrm{in}})
    +\log\pi_{T_\alpha+1}
    &\ge
    (1+\eps)\frac{B_\alpha}{I^*}
    (I^*-\eta_{\mathrm{in}})
    +\log\pi_{T_\alpha+1}\\
    &=
    \left(1+\frac{\eps}{2}\right)B_\alpha
    +\log\pi_{T_\alpha+1}\\
    &=
    L_\alpha+\frac{\eps}{2}B_\alpha
    \ge L_\alpha.
\end{align*}
Thus both conditions of \Cref{thm:delay-upper-generic} hold, and hence
\[
    \mathbb P_{P,Q}^{T_\alpha}
    \left(
        \tau_\alpha^{\mathrm{PFA}}-T_\alpha
        \le
        (1+\eps)
        \frac{L_\alpha-\log\pi_{T_\alpha+1}}{I^*}
        +1
    \right)
    \to1.
\]
The expression for the canonical weights follows from
\eqref{eq:pi-log}.

Finally, the change law agrees with $\Pinf{P}$ up to time
$T_\alpha$, so by \Cref{thm:pfa-control},
\[
    \mathbb P_{P,Q}^{T_\alpha}
    (\tau_\alpha^{\mathrm{PFA}}\le T_\alpha)
    =
    \Pinf{P}(\tau_\alpha^{\mathrm{PFA}}\le T_\alpha)
    \le \alpha.
\]
Therefore,
\[
    \mathbb P_{P,Q}^{T_\alpha}
    (\tau_\alpha^{\mathrm{PFA}}>T_\alpha)
    \ge 1-\alpha\to1.
\]
\end{proof}

\subsection{Information lower bounds}\label{app:proof-lower-bounds}

\begin{proof}[Proof of \Cref{thm:pointwise-lower}]
Let
\[
    A_\alpha
    :=
    \{T_\alpha<\tau_\alpha\le T_\alpha+d_\alpha\}.
\]
Since \(\tau_\alpha\) is PFA-\(\alpha\) valid,
\[
    \Pinf{P}(A_\alpha)
    \le
    \Pinf{P}(\tau_\alpha<\infty)
    \le \alpha.
\]
On \(\calF_{T_\alpha+d_\alpha}\), the likelihood ratio of
\(\Pchg{P}{Q}{T_\alpha}\) with respect to \(\Pinf{P}\) is
\[
    \Lambda_\alpha
    =
    \prod_{i=T_\alpha+1}^{T_\alpha+d_\alpha}
    \frac{dQ}{dP}(X_i).
\]
Therefore, for every \(c_\alpha\in\mathbb R\),
\begin{align}
    \mathbb P_{P,Q}^{T_\alpha}(A_\alpha)
    &=
    \mathbb E^\infty_P
    \left[
        \Lambda_\alpha\mathbf 1_{A_\alpha}
    \right]
    \nonumber\\
    &\le
    e^{c_\alpha}\Pinf{P}(A_\alpha)
    +
    \mathbb P_{P,Q}^{T_\alpha}
    (\log\Lambda_\alpha>c_\alpha).
    \label{eq:pointwise-lower-com}
\end{align}
Take
\[
    c_\alpha
    =
    \left(1-\frac{\varepsilon}{2}\right)L_\alpha.
\]
Then
\[
    e^{c_\alpha}\Pinf{P}(A_\alpha)
    \le
    e^{c_\alpha}\alpha
    =
    e^{-\varepsilon L_\alpha/2}
    \longrightarrow0.
\]

It remains to show that
\[
    \mathbb P_{P,Q}^{T_\alpha}
    (\log\Lambda_\alpha>c_\alpha)
    \longrightarrow0.
\]
We consider two cases.

First, along any subsequence for which \(d_\alpha\to\infty\), the
weak law of large numbers under \(\Pchg{P}{Q}{T_\alpha}\) gives
\[
    \log\Lambda_\alpha
    =
    d_\alpha I^*
    +
    o_{\mathbb P}(d_\alpha).
\]
Since
\[
    d_\alpha I^*
    \le
    (1-\varepsilon)L_\alpha
\]
and \(d_\alpha=O(L_\alpha)\), it follows that
\[
    \log\Lambda_\alpha
    \le
    (1-\varepsilon)L_\alpha
    +
    o_{\mathbb P}(L_\alpha).
\]
Consequently,
\begin{align*}
&\mathbb P_{P,Q}^{T_\alpha}
\left(
    \log\Lambda_\alpha
    >
    \left(1-\frac{\varepsilon}{2}\right)L_\alpha
\right)\\
&\qquad\le
\mathbb P_{P,Q}^{T_\alpha}
\left(
    o_{\mathbb P}(L_\alpha)
    >
    \frac{\varepsilon}{2}L_\alpha
\right)
\longrightarrow0.
\end{align*}

Second, consider a subsequence along which \(d_\alpha\) is bounded.
Since \(d_\alpha\) is integer-valued, it has a further subsequence on
which \(d_\alpha=d\) for some fixed \(d\ge1\). Along this subsequence,
\[
    \log\Lambda_\alpha
    \stackrel{d}{=}
    \sum_{i=1}^{d}
    \log\frac{dQ}{dP}(Y_i),
    \qquad
    Y_1,\ldots,Y_d\stackrel{\mathrm{iid}}{\sim}Q,
\]
whose distribution does not depend on \(\alpha\). Hence
\(\log\Lambda_\alpha\) is tight. Since
\[
    c_\alpha
    =
    \left(1-\frac{\varepsilon}{2}\right)L_\alpha
    \longrightarrow\infty,
\]
we again have
\[
    \mathbb P_{P,Q}^{T_\alpha}
    (\log\Lambda_\alpha>c_\alpha)
    \longrightarrow0.
\]

Thus every subsequence has a further subsequence along which the
second term in \eqref{eq:pointwise-lower-com} converges to zero.
Therefore, it converges to zero along the full sequence. Combining
the two terms in \eqref{eq:pointwise-lower-com} yields
\[
    \mathbb P_{P,Q}^{T_\alpha}
    \left(
        T_\alpha<\tau_\alpha
        \le T_\alpha+d_\alpha
    \right)
    \longrightarrow0.
\]
\end{proof}

\begin{proof}[Proof of \Cref{thm:horizon-lower}]
For \(j=1,\ldots,N_\alpha\), define
\[
    A_{j,\alpha}
    :=
    \left\{
        T_{j,\alpha}<\tau_\alpha
        \le T_{j,\alpha}+d_\alpha
    \right\}.
\]
Because the corresponding time windows are pairwise disjoint, the
events \(A_{1,\alpha},\ldots,A_{N_\alpha,\alpha}\) are pairwise
disjoint under every probability law.

On \(\mathcal F_{T_{j,\alpha}+d_\alpha}\), the likelihood ratio of
\(\mathbb P_{P,Q}^{T_{j,\alpha}}\) with respect to
\(\mathbb P_P^\infty\) is
\[
    \Lambda_{j,\alpha}
    :=
    \prod_{i=T_{j,\alpha}+1}^{T_{j,\alpha}+d_\alpha}
    \frac{dQ}{dP}(X_i).
\]
Since \(d_\alpha\to\infty\), the weak law of large numbers gives, for
every fixed \(\eta>0\),
\begin{equation}
\label{eq:horizon-lr-lln}
    u_\alpha
    :=
    \mathbb P_{P,Q}^{T_{j,\alpha}}
    \left(
        \log\Lambda_{j,\alpha}
        >
        d_\alpha(I^*+\eta)
    \right)
    \longrightarrow0.
\end{equation}
The quantity \(u_\alpha\) does not depend on \(j\), because under
\(\mathbb P_{P,Q}^{T_{j,\alpha}}\),
\[
    \log\Lambda_{j,\alpha}
    \stackrel{d}{=}
    \sum_{i=1}^{d_\alpha}
    \log\frac{dQ}{dP}(Y_i),
    \qquad
    Y_1,\ldots,Y_{d_\alpha}
    \stackrel{\mathrm{iid}}{\sim}Q.
\]
Thus the convergence in \eqref{eq:horizon-lr-lln} is uniform over
\(j\).

By change of measure,
\begin{align*}
    \mathbb P_P^\infty(A_{j,\alpha})
    &=
    \mathbb E_{P,Q}^{T_{j,\alpha}}
    \left[
        \Lambda_{j,\alpha}^{-1}
        \mathbf 1_{A_{j,\alpha}}
    \right]\\
    &\ge
    e^{-d_\alpha(I^*+\eta)}
    \mathbb P_{P,Q}^{T_{j,\alpha}}
    \left(
        A_{j,\alpha}
        \cap
        \left\{
            \log\Lambda_{j,\alpha}
            \le d_\alpha(I^*+\eta)
        \right\}
    \right)\\
    &\ge
    e^{-d_\alpha(I^*+\eta)}
    \left\{
        \mathbb P_{P,Q}^{T_{j,\alpha}}(A_{j,\alpha})
        -u_\alpha
    \right\}.
\end{align*}
Using \eqref{eq:horizon-uniform-power}, uniformly over
\(j=1,\ldots,N_\alpha\),
\[
    \mathbb P_P^\infty(A_{j,\alpha})
    \ge
    e^{-d_\alpha(I^*+\eta)}
    \{1-\beta-\rho_\alpha-u_\alpha\}.
\]

Since the \(A_{j,\alpha}\)'s are pairwise disjoint and
\(\tau_\alpha\) is PFA-\(\alpha\) valid,
\begin{align*}
    \alpha
    &\ge
    \mathbb P_P^\infty(\tau_\alpha<\infty)\\
    &\ge
    \sum_{j=1}^{N_\alpha}
    \mathbb P_P^\infty(A_{j,\alpha})\\
    &\ge
    N_\alpha
    e^{-d_\alpha(I^*+\eta)}
    \{1-\beta-\rho_\alpha-u_\alpha\}.
\end{align*}
The term in braces is positive eventually because
\(\beta<1\), \(\rho_\alpha\to0\), and \(u_\alpha\to0\).
Taking logarithms therefore gives
\[
    d_\alpha(I^*+\eta)
    \ge
    L_\alpha+\log N_\alpha
    +\log\{1-\beta-\rho_\alpha-u_\alpha\}.
\]
Moreover,
\[
    \log\{1-\beta-\rho_\alpha-u_\alpha\}
    =
    \log(1-\beta)+o(1),
\]
which proves \eqref{eq:horizon-lower}.

It remains to simplify the start-time term. Since
\[
    N_\alpha
    =
    \left\lfloor
        \frac{T_{\max,\alpha}}{d_\alpha+1}
    \right\rfloor
    \longrightarrow\infty,
\]
we have
\[
    \log N_\alpha
    =
    \log\left(
        \frac{T_{\max,\alpha}}{d_\alpha+1}
    \right)
    +o(1).
\]
Because \(d_\alpha\to\infty\),
\[
    \log(d_\alpha+1)=\log d_\alpha+o(1),
\]
and hence
\[
    \log N_\alpha
    =
    \log\left(
        \frac{T_{\max,\alpha}}{d_\alpha}
    \right)
    +o(1).
\]
Finally, if
\[
    \log d_\alpha=o(\log T_{\max,\alpha}),
\]
then
\[
    \log N_\alpha
    =
    \log T_{\max,\alpha}-\log d_\alpha+o(1)
    =
    (1+o(1))\log T_{\max,\alpha}.
\]
\end{proof}

\begin{proof}[Proof of \Cref{thm:early-change-impossibility}]
Let
\[
A:=\{\tau<\infty\},
\qquad
p:=\mathbb P_{P,Q}^{T}(A),
\qquad
q:=\mathbb P_Q^\infty(A).
\]
Since $Q\in\mathcal P$ and $\tau$ has PFA at most $\alpha$,
\[
q\leq\alpha.
\]

The laws $\mathbb P_{P,Q}^{T}=P^T\otimes Q^\infty$ and
$\mathbb P_Q^\infty=Q^\infty$ differ only in their first $T$
coordinates. Hence
\[
D_{\mathrm{KL}}
\bigl(
\mathbb P_{P,Q}^{T}\,\|\,\mathbb P_Q^\infty
\bigr)
=
T D_{\mathrm{KL}}(P\|Q).
\]
Applying the data-processing inequality to the measurable map
$\omega\mapsto\mathbbm 1_A(\omega)$ gives
\[
T D_{\mathrm{KL}}(P\|Q)
\geq
\operatorname{kl}(p,q),
\]
where
\[
\operatorname{kl}(p,q)
=
p\log\frac pq
+
(1-p)\log\frac{1-p}{1-q}
\]
denotes the binary relative entropy.
Moreover,
\begin{align*}
\operatorname{kl}(p,q)
&=
p\log\frac1q
+
p\log p
+
(1-p)\log(1-p)
-
(1-p)\log(1-q)\\
&\geq
p\log\frac1q-\log 2\\
&\geq
p\log\frac1\alpha-\log 2.
\end{align*}
The first inequality uses
 the fact that $-p\log p-(1-p)\log(1-p)\leq\log 2$
and $-\log(1-q)\geq0$, while the second uses $q\leq\alpha$.
Therefore,
\[
p\log\frac1\alpha
\leq
T D_{\mathrm{KL}}(P\|Q)+\log 2,
\]
which proves
\[
\mathbb P_{P,Q}^{T}(\tau<\infty)
\leq
\frac{
T D_{\mathrm{KL}}(P\|Q)+\log 2
}{
\log(1/\alpha)
}.
\]
Taking the minimum with one gives the stated nonasymptotic bound.

Finally, if
\[
T_\alpha 
=
o\bigl(\log(1/\alpha)\bigr),
\]
then
\[
\frac{
T_\alpha D_{\mathrm{KL}}(P\|Q)+\log 2
}{
\log(1/\alpha)
}
\longrightarrow0,
\]
which proves the final result.
\end{proof}

\subsection{Early-change lower bound under local ARL control}\label{app:proof-arl-early}

We now give analogues of \Cref{thm:pointwise-lower}
under the local finite-horizon ARL guarantee
\begin{equation}
\label{eq:local-arl-guarantee}
    \sup_{R\in\mathcal P}
    \mathbb P_R^\infty(\tau_A\le m)
    \le
    \frac{m}{A},
    \qquad m\in\mathbb N.
\end{equation}
This guarantee is stronger than the usual condition
\(\inf_{R\in\mathcal P}\mathbb E_R^\infty[\tau_A]\ge A\), and is
satisfied by the ARL e-detectors constructed in this paper.

\begin{theorem}[Pointwise lower bound under local ARL control]
\label{thm:pointwise-arl-lower}
Fix \(P,Q\in\mathcal P\) such that
\[
    I^*:=D_{\mathrm{KL}}(Q\|P)\in(0,\infty),
\]
and fix \(\varepsilon\in(0,1)\). Let \(A\to\infty\), and let
\(\tau_A\) satisfy \eqref{eq:local-arl-guarantee}. Suppose that
\(T_A,d_A\in\mathbb N\) satisfy \(T_A+d_A<A\) and
\[
    H_A
    :=
    \log\frac{A}{T_A+d_A}
    \longrightarrow\infty.
\]
If
\[
    d_A
    \le
    (1-\varepsilon)\frac{H_A}{I^*},
\]
then
\[
    \mathbb P_{P,Q}^{T_A}
    \left(
        T_A<\tau_A\le T_A+d_A
    \right)
    \longrightarrow0.
\]
\end{theorem}

\begin{proof}
Let
\[
    \mathcal A_A
    :=
    \{T_A<\tau_A\le T_A+d_A\}.
\]
By the local ARL guarantee,
\[
    \mathbb P_P^\infty(\mathcal A_A)
    \le
    \mathbb P_P^\infty(\tau_A\le T_A+d_A)
    \le
    \frac{T_A+d_A}{A}
    =
    e^{-H_A}.
\]
On \(\mathcal F_{T_A+d_A}\), the likelihood ratio of
\(\mathbb P_{P,Q}^{T_A}\) with respect to \(\mathbb P_P^\infty\) is
\[
    \Lambda_A
    :=
    \prod_{i=T_A+1}^{T_A+d_A}
    \frac{dQ}{dP}(X_i).
\]
Hence, for every \(c_A\in\mathbb R\),
\begin{align}
    \mathbb P_{P,Q}^{T_A}(\mathcal A_A)
    &=
    \mathbb E_P^\infty
    \left[
        \Lambda_A\mathbf 1_{\mathcal A_A}
    \right]
    \nonumber\\
    &\le
    e^{c_A}\mathbb P_P^\infty(\mathcal A_A)
    +
    \mathbb P_{P,Q}^{T_A}
    (\log\Lambda_A>c_A).
    \label{eq:pointwise-arl-change-measure}
\end{align}
Take
\[
    c_A
    :=
    \left(1-\frac{\varepsilon}{2}\right)H_A.
\]
Then
\[
    e^{c_A}\mathbb P_P^\infty(\mathcal A_A)
    \le
    e^{-\varepsilon H_A/2}
    \longrightarrow0.
\]

It remains to control the likelihood-ratio term. Along any
subsequence for which \(d_A\to\infty\), the weak law of large numbers
gives
\[
    \log\Lambda_A
    =
    d_AI^*+o_{\mathbb P}(d_A).
\]
Since
\[
    d_AI^*
    \le
    (1-\varepsilon)H_A
    \qquad\text{and}\qquad
    d_A=O(H_A),
\]
we obtain
\[
    \log\Lambda_A
    \le
    (1-\varepsilon)H_A+o_{\mathbb P}(H_A),
\]
and therefore
\[
    \mathbb P_{P,Q}^{T_A}
    \left(
        \log\Lambda_A>
        \left(1-\frac{\varepsilon}{2}\right)H_A
    \right)
    \longrightarrow0.
\]

Along any subsequence for which \(d_A\) is bounded, there is a further
subsequence on which \(d_A=d\) is constant. Along that subsequence,
\[
    \log\Lambda_A
    \stackrel{d}{=}
    \sum_{i=1}^d
    \log\frac{dQ}{dP}(Y_i),
    \qquad
    Y_i\stackrel{\mathrm{iid}}{\sim}Q,
\]
so \(\log\Lambda_A\) is tight. Since \(c_A\to\infty\), the same
probability again converges to zero. Thus the likelihood-ratio term
vanishes along the full sequence. Substitution into
\eqref{eq:pointwise-arl-change-measure} proves the result.
\end{proof}

\begin{corollary}
\label{cor:pointwise-arl-lower}
Under the conditions of \Cref{thm:pointwise-arl-lower}, if additionally
\[
    \log(T_A+d_A)=o(\log A),
\]
then, for every fixed \(\varepsilon\in(0,1)\),
\[
    d_A
    \le
    (1-\varepsilon)\frac{\log A}{I^*}
\]
implies
\[
    \mathbb P_{P,Q}^{T_A}
    \left(
        T_A<\tau_A\le T_A+d_A
    \right)
    \longrightarrow0.
\]
\end{corollary}

\begin{proof}
The assumption gives
\[
    H_A
    =
    \log A-\log(T_A+d_A)
    =
    (1+o(1))\log A.
\]
Thus, after replacing \(\varepsilon\) by any slightly smaller positive
constant, the conclusion follows from
\Cref{thm:pointwise-arl-lower}.
\end{proof}

\begin{theorem}[Early-change impossibility under ARL control with local guarantee]
\label{thm:early-change-local-arl}
Let $\tau_A$ be a stopping time satisfying
\begin{equation}
\label{eq:local-fa-control}
\sup_{R\in\mathcal P}
\mathbb P_R^\infty(\tau_A\le m)
\le \frac{m}{A}
\qquad
\text{for every }m\in\mathbb N.
\end{equation}
Fix $P,Q\in\mathcal P$ and $T,d\in\mathbb N$ such that
$T+d<A$, and suppose $P\ll Q$. Then
\begin{equation}
\label{eq:early-change-unconditional-bound}
\mathbb P_{P,Q}^{T}
\bigl(T<\tau_A\le T+d\bigr)
\le
\min\left\{
1,\,
\frac{
T D_{\mathrm{KL}}(P\|Q)+\log 2
}{
\log\!\left(A/(T+d)\right)
}
\right\}.
\end{equation}
Moreover, since
$\mathbb P_{P,Q}^{T}(\tau_A>T)
\ge 1-\frac{T}{A},$
we have
\begin{equation}
\label{eq:early-change-conditional-bound}
\mathbb P_{P,Q}^{T}
\left(
\tau_A-T\le d
\,\middle|\,
\tau_A>T
\right)
\le
\min\left\{
1,\,
\frac{
T D_{\mathrm{KL}}(P\|Q)+\log 2
}{
(1-T/A)\log\!\left(A/(T+d)\right)
}
\right\}.
\end{equation}
In particular, consider sequences $A\to\infty$, $T_A$, and $d_A$
such that
$T_A+d_A<A,
\frac{T_A}{A}\to0,$
and
\[
T_A 
=
o\left(
\log\frac{A}{T_A+d_A}
\right).
\]
Then
\[
\mathbb P_{P,Q}^{T_A}
\left(
\tau_A-T_A\le d_A
\,\middle|\,
\tau_A>T_A
\right)
\longrightarrow0.
\]
\end{theorem}

\begin{proof}
Let
\[
B:=\{T<\tau_A\le T+d\},
\]
which is measurable with respect to $\mathcal F_{T+d}$. Define
\[
p:=\mathbb P_{P,Q}^{T}(B),
\qquad
q:=\mathbb P_Q^\infty(B).
\]
Since $Q\in\mathcal P$, the local false-alarm guarantee
\eqref{eq:local-fa-control} gives
\[
q
\le
\mathbb P_Q^\infty(\tau_A\le T+d)
\le
\frac{T+d}{A}.
\]

Restricted to $\mathcal F_{T+d}$, the change law is
\[
P^T\otimes Q^d,
\]
whereas the all-$Q$ no-change law is
\[
Q^{T+d}.
\]
Because these laws differ only in their first $T$ coordinates,
\[
D_{\mathrm{KL}}
\left(
P^T\otimes Q^d
\,\middle\|\,
Q^{T+d}
\right)
=
T D_{\mathrm{KL}}(P\|Q).
\]
Applying the data-processing inequality to the map
$\omega\mapsto\mathbbm 1_B(\omega)$ yields
\[
T D_{\mathrm{KL}}(P\|Q)
\ge
\operatorname{kl}(p,q),
\]
where
\[
\operatorname{kl}(p,q)
=
p\log\frac pq
+
(1-p)\log\frac{1-p}{1-q}.
\]
Using
\[
\operatorname{kl}(p,q)
\ge
p\log\frac1q-\log 2,
\]
we obtain
\[
T D_{\mathrm{KL}}(P\|Q)
\ge
p\log\frac1q-\log 2
\ge
p\log\frac{A}{T+d}-\log 2.
\]
Rearranging proves \eqref{eq:early-change-unconditional-bound}.

Before time $T$, the change law agrees with the no-change law
$P^\infty$. Therefore,
\[
\mathbb P_{P,Q}^{T}(\tau_A\le T)
=
\mathbb P_P^\infty(\tau_A\le T)
\le
\frac{T}{A},
\]
and consequently
\[
\mathbb P_{P,Q}^{T}(\tau_A>T)
\ge
1-\frac{T}{A}.
\]
It follows that
\begin{align*}
\mathbb P_{P,Q}^{T}
\left(
\tau_A-T\le d
\,\middle|\,
\tau_A>T
\right)
&=
\frac{
\mathbb P_{P,Q}^{T}(T<\tau_A\le T+d)
}{
\mathbb P_{P,Q}^{T}(\tau_A>T)
} \\
&\le
\frac{
T D_{\mathrm{KL}}(P\|Q)+\log 2
}{
(1-T/A)\log\!\left(A/(T+d)\right)
},
\end{align*}
which proves \eqref{eq:early-change-conditional-bound}.
 The asymptotic conclusion follows
immediately.
\end{proof}

\section{Proofs for the concrete examples}\label{app:example-proofs}

\subsection{Sub-Gaussian mean-change example}\label{app:proof-sub-Gaussian}

\begin{proof}[Proof of \Cref{prop:gaussian-mixture}]
Fix $\theta\in\mathbb R$, a starting time $s\geq 1$, and a distribution
$P\in\mathcal P_\theta$. For each $\lambda\in\mathbb R$, define
\[
L_{s:t}^{\lambda}
:=
\exp\left\{
\lambda\sum_{i=s}^{t}(X_i-\theta)
-\frac{\lambda^2\sigma^2}{2}(t-s+1)
\right\},
\qquad t\geq s,
\]
with $L_{s:s-1}^{\lambda}=1$.

Since $P\in\mathcal P_\theta$, we have
\[
\mathbb E_P\!\left[
e^{\lambda(X_t-\theta)}
\right]
\leq
e^{\lambda^2\sigma^2/2}
\qquad
\text{for every }\lambda\in\mathbb R.
\]
Moreover, under $P^\infty$, $X_t$ is independent of
$\mathcal F_{t-1}$. Therefore, for every $t\geq s$,
\begin{align*}
\mathbb E_P^\infty
\left[
L_{s:t}^{\lambda}
\mid \mathcal F_{t-1}
\right]
&=
L_{s:t-1}^{\lambda}
\mathbb E_P
\left[
\exp\left\{
\lambda(X_t-\theta)
-\frac{\lambda^2\sigma^2}{2}
\right\}
\right] \\
&\leq
L_{s:t-1}^{\lambda}.
\end{align*}
Thus, for each fixed $\lambda\in\mathbb R$,
$\{L_{s:t}^{\lambda}\}_{t\geq s-1}$ is a nonnegative
test supermartingale under $P^\infty$.

Let $\Pi_\rho$ be the centered Gaussian distribution
$\mathcal N(0,\rho^2)$ on $\mathbb R$, and define the Gaussian mixture
\[
\widetilde M_{s:t}^{\theta}
:=
\int_{\mathbb R}
L_{s:t}^{\lambda}\,
\Pi_\rho(d\lambda),
\qquad
t\geq s-1.
\]
By Tonelli’s theorem and the preceding conditional supermartingale
inequality,
\begin{align*}
\mathbb E_P^\infty
\left[
\widetilde M_{s:t}^{\theta}
\mid\mathcal F_{t-1}
\right]
&=
\int_{\mathbb R}
\mathbb E_P^\infty
\left[
L_{s:t}^{\lambda}
\mid\mathcal F_{t-1}
\right]
\Pi_\rho(d\lambda) \\
&\leq
\int_{\mathbb R}
L_{s:t-1}^{\lambda}\,
\Pi_\rho(d\lambda) \\
&=
\widetilde M_{s:t-1}^{\theta}.
\end{align*}
Hence $\{\widetilde M_{s:t}^{\theta}\}_{t\geq s-1}$ is also a
nonnegative test supermartingale, with
$\widetilde M_{s:s-1}^{\theta}=1$.

It remains to evaluate the Gaussian integral. Write
\[
n_{s:t}:=t-s+1,
\qquad
\overline X_{s:t}
:=
\frac{1}{n_{s:t}}\sum_{i=s}^{t}X_i.
\]
Using the density of $\mathcal N(0,\rho^2)$, we obtain
\begin{align*}
\widetilde M_{s:t}^{\theta}
&=
\frac{1}{\sqrt{2\pi\rho^2}}
\int_{\mathbb R}
\exp\left\{
\lambda n_{s:t}
(\overline X_{s:t}-\theta)
-\frac{\lambda^2\sigma^2n_{s:t}}{2}
-\frac{\lambda^2}{2\rho^2}
\right\}
\,d\lambda \\
&=
\frac{1}{\sqrt{2\pi\rho^2}}
\int_{\mathbb R}
\exp\left\{
-\frac{1+\rho^2\sigma^2n_{s:t}}{2\rho^2}
\lambda^2
+
\lambda n_{s:t}
(\overline X_{s:t}-\theta)
\right\}
\,d\lambda.
\end{align*}
Completing the square and evaluating the Gaussian integral yields
\[
\widetilde M_{s:t}^{\theta}
=
\frac{1}{
\sqrt{1+\rho^2\sigma^2n_{s:t}}
}
\exp\left\{
\frac{
\rho^2n_{s:t}^2
(\overline X_{s:t}-\theta)^2
}{
2\left(1+\rho^2\sigma^2n_{s:t}\right)
}
\right\}.
\]
Thus $\widetilde M_{s:t}^{\theta}=M_{s:t}^{\theta}$, where
$M_{s:t}^{\theta}$ is the process defined in \eqref{eq:mixture-evalue}.

Finally, a nonnegative test supermartingale initialized at one is an
e-process.  Hence the process
$\{M_{s:t}^{\theta}:t\ge s-1\}$ is an $s$-delay e-process under every
$P\in\mathcal P_\theta$.
\end{proof}

\begin{proof}[Proof of \Cref{thm:upper-bdd-sub-Gaussian}]
\textbf{Part 1. High probability bound on delay:}

Fix $\varepsilon>0$ and set
\[
  d_\alpha=\left\lceil(1+\varepsilon)\frac{L}{I}\right\rceil,
  \qquad
  T=T_\alpha,
  \qquad
  d=d_\alpha.
  \]
Use the two intervals from Lemma~\ref{lem:sufficient-stopping} with $m=T$:
  \[
    J_0=(T+1):(T+d),
    \qquad
    J_1=1:(T+d).
    \]
Let
\[
  \bar X_{\rm pre}=\bar X_{1:T},
  \qquad
  \bar X_{\rm post}=\bar X_{T+1:T+d}.
  \]
Then
\[
  \bar X_0=\bar X_{\rm post},
  \qquad
  \bar X_1=\frac{T\bar X_{\rm pre}+d\bar X_{\rm post}}{T+d},
  \]
and hence
\begin{equation}
\label{eq:mean-separation-late}
\bar X_0-\bar X_1
=
  \frac{T}{T+d}(\bar X_{\rm post}-\bar X_{\rm pre}).
\end{equation}
Because $T/d\to\infty$,
\[
  \frac{T}{T+d}=1+o(1).
  \]
By Sub-Gaussian concentration,
\[
  \bar X_{\rm pre}=\mu+O_{\bbP}(T^{-1/2}),
  \qquad
  \bar X_{\rm post}=\nu+O_{\bbP}(d^{-1/2}),
  \]
so \eqref{eq:mean-separation-late} gives
\begin{equation}
\label{eq:sep-is-delta}
\bar X_0-\bar X_1
=
  (\nu-\mu)\{1+o_{\bbP}(1)\}.
\end{equation}
Also, by \eqref{eq:kappa-asymp},
\[
  \kappa_d^{-1/2}
  =
    \sqrt{\frac{2\sigma^2}{d}}\{1+o(1)\},
  \]
and, since $T/d\to\infty$,
\[
  \kappa_{T+d}^{-1/2}
  =
    \sqrt{\frac{2\sigma^2}{T+d}}\{1+o(1)\}
  =o(d^{-1/2}).
  \]
Therefore the quantity $B_{T,d}$ from \eqref{eq:Bmd-def} satisfies
\begin{align*}
B_{T,d}
&=
  \frac{(\bar X_0-\bar X_1)^2}
{(\kappa_d^{-1/2}+\kappa_{T+d}^{-1/2})^2} \\
&=
  \frac{(\nu-\mu)^2\{1+o_{\bbP}(1)\}}
{2\sigma^2/d}\nonumber\\
&=
  I d\{1+o_{\bbP}(1)\}.
\end{align*}
Since $d=(1+\varepsilon)L/I+O(1)$,
\[
  B_{T,d}=(1+\varepsilon)L\{1+o_{\bbP}(1)\}.
  \]
The mixture penalty is negligible because $\log(T+d)=o(L)$:
  \[
    \frac12\log(1+\rho^2\sigma^2(T+d))=o(L).
    \]
Thus
\[
  B_{T,d}
  -\frac12\log(1+\rho^2\sigma^2(T+d))
  \ge L
  \]
with probability tending to one.  By Lemma~\ref{lem:sufficient-stopping},
\[
  \bbP^T_{P,Q}(\tau_A^{\mathrm{ARL}}\le T+d)\to 1.
  \]
This proves
\[
  \tau_A^{\mathrm{ARL}}-T\le \left\lceil(1+\varepsilon)\frac{L}{I}\right\rceil\le(1+2\varepsilon)\frac{L}{I}
  \]
with probability tending to one.  Since $\varepsilon>0$ is arbitrary, the claimed in-probability upper bound follows.

\textbf{Part 2. CADD Bound:}
Let $Y = (\tau_A - T)^+$. To compute the expectation, we partition the tail into three regimes using $b_\alpha = \lfloor T/2 \rfloor$:
\begin{equation} \label{eq:expectation_split}
    \mathbb{E}^T_{P,Q}[Y] \le d + b_\alpha \mathbb{P}^T_{P,Q}(Y > d) + \tau_{max} \mathbb{P}^T_{P,Q}(Y > b_\alpha).
\end{equation}

\textit{Step A ($\tau_{max}$: A deterministic upper bound on the stopping time).}
Let \(a=\rho^2\sigma^2\). Since the exponential factor in
\(M^\theta_{s:t}\) is at least one, for every data sequence and every
\(\theta\in\mathbb R\),
\[
M^\theta_{s:t}
\ge
\frac{1}{\sqrt{1+a(t-s+1)}}.
\]
Consequently,
\[
D_t
=
\inf_{\theta\in\mathbb R}\sum_{s=1}^tM^\theta_{s:t}
\ge
\sum_{n=1}^t\frac{1}{\sqrt{1+an}}
\ge
\frac{t}{\sqrt{1+at}}
\ge
\sqrt{\frac{t}{1+a}},
\]
where the last inequality follows from \(t\ge1\). Therefore, with
$t_A=\left\lceil(1+\rho^2\sigma^2)A^2\right\rceil,$
we have \(D_{t_A}\ge A\) for every sample path, and hence
$\tau_A\le t_A=
\left\lceil(1+\rho^2\sigma^2)A^2\right\rceil.$
Thus, since \(A=e^L\), there is a constant
\(C'_{\rho,\sigma}>0\) such that
\[
\tau_A\le \tau_{\max}:=C'_{\rho,\sigma}\exp{(2L)}
\]
deterministically.

\textit{Step B (Bounding $\mathbb{P}^T_{P,Q}(Y > d)$):} 
By the definition of the stopping time, if the detector fails to stop by time $T+d$ (i.e., $\tau_A > T+d$), the statistic must remain below the threshold: $\log D_{T+d} < L$. 
Meanwhile, the first part of \Cref{lem:sufficient-stopping} establishes the deterministic lower bound $\log D_{T+d} \ge B_{T,d} - \frac{1}{2}\log(1+\rho^2\sigma^2(T+d))$. Chaining these inequalities yields:
\[
    B_{T,d} - \frac{1}{2}\log(1+\rho^2\sigma^2(T+d)) \le \log D_{T+d} < L.
\]
Because the penalty term is $o(L)$, this implies $B_{T,d} < L + o(L)$.

To see what this implies for the data, we expand $B_{T,d}$. Let $Z = \bar{X}_{post} - \bar{X}_{pre}$. From the definition of the detector's lower bound, we have:
\[
    B_{T,d} = \frac{Z^2}{2\sigma^2(1/d + 1/T)}\{1+o(1)\} = \frac{Z^2 d}{2\sigma^2}\{1+o(1)\},
\]
where the second equality follows because $d/T \to 0$. Next, we substitute our chosen horizon $d = \lceil(1+\epsilon)L/I\rceil \ge (1+\epsilon)L/I$, and recall the definition of the KL proxy $I = \frac{(\nu-\mu)^2}{2\sigma^2}$:
\[
    B_{T,d} \ge \frac{Z^2}{2\sigma^2} \frac{(1+\epsilon)L}{(\nu-\mu)^2 / (2\sigma^2)} \{1+o(1)\} = \frac{Z^2}{(\nu-\mu)^2}(1+\epsilon)L \{1+o(1)\}.
\]
The failure-to-stop condition $B_{T,d} < L + o(L)$ therefore requires:
\[
    \frac{Z^2}{(\nu-\mu)^2}(1+\epsilon)L \{1+o(1)\} < L + o(L).
\]
Dividing both sides by $L (1+\epsilon)$ yields:
\[
    \frac{Z^2}{(\nu-\mu)^2} \{1+o(1)\} < \frac{1 + o(1)}{1+\epsilon}.
\]
Because $\epsilon > 0$ is a strictly positive constant, the right-hand side is strictly less than 1 for all sufficiently large $L$. Therefore, for $\delta =\frac{\epsilon}{2(1+\epsilon)}$ (depending purely on $\epsilon$) and for all large $L$:
\[
    |Z| \le |\nu-\mu|(1-\delta).
\]
The true expected value of $Z$ is $\nu - \mu$. The above inequality demonstrates that to fail to stop by time $T+d$, the empirical difference $Z$ must deviate significantly below its expected magnitude. 

Under $\mathbb{P}^T_{P,Q}$, $Z$ is the difference of independent sub-Gaussian sample averages, with variance proxy $\sigma_*^2 = \sigma^2/d + \sigma^2/T = \frac{\sigma^2}{d}\{1+o(1)\}$. By Hoeffding's inequality for sub-Gaussian variables:
\[
    \mathbb{P}^T_{P,Q}\big(|Z| \le |\nu-\mu|(1-\delta)\big) \le \exp\left( - \frac{\delta^2 (\nu-\mu)^2}{2 \sigma_*^2} \right) = \exp\Big( - \delta^2 I d \{1+o(1)\} \Big).
\]
Since $d = \lceil(1+\epsilon)L/I\rceil$, the exponent is strictly of order $L$. Thus, for $c_\epsilon:=\frac{\epsilon^2}{4(1+\epsilon)}$, we have $\mathbb{P}^T_{P,Q}(Y > d) \le \exp(-c_\epsilon L + o(L))$.

\textit{Step C (Bounding $\mathbb{P}^T_{P,Q}(Y > b_\alpha)$):} 
We apply the exact same chaining of inequalities at the macroscopic time scale $b_\alpha = \lfloor T/2 \rfloor$. By the definition of the stopping time, failing to stop by $T + b_\alpha$ implies $\log D_{T+b_\alpha} < L$. 

Applying the deterministic lower bound from \Cref{lem:sufficient-stopping} at this horizon, we have:
\[
    B_{T,b_\alpha} - \frac{1}{2}\log(1+\rho^2\sigma^2(T+b_\alpha)) \le \log D_{T+b_\alpha} < L.
\]
Because $b_\alpha = \lfloor T/2 \rfloor$, the penalty term is $\frac{1}{2}\log(1+\rho^2\sigma^2(1.5T)) = O(\log T)$. By the theorem's assumption, $\log T = o(L)$, so the penalty term is $o(L)$. Rearranging gives:
\[
    B_{T,b_\alpha} < L + o(L).
\]

To see what this requires of the data, let $Z = \bar{X}_{post} - \bar{X}_{pre}$ be the empirical difference over these large blocks (where $\bar{X}_{post}$ is the average over $b_\alpha$ samples). Expanding $B_{T,b_\alpha}$, we get:
\[
    B_{T,b_\alpha} = \frac{Z^2}{2\sigma^2(1/b_\alpha + 1/T)}\{1+o(1)\}.
\]
Since $b_\alpha = \lfloor T/2 \rfloor$, we can bound the inverse sample sizes for all $T \ge 6$:
\[
    \frac{1}{b_\alpha} + \frac{1}{T} = \frac{1}{\lfloor T/2 \rfloor} + \frac{1}{T} \le \frac{4}{T}.
\]
Substituting this bound into the expansion for $B_{T,b_\alpha}$ yields:
\[
    B_{T,b_\alpha} \ge \frac{Z^2}{2\sigma^2 (4/T)} \{1+o(1)\} = \frac{Z^2 T}{8\sigma^2} \{1+o(1)\}.
\]
Therefore, the failure-to-stop condition $B_{T,b_\alpha} < L + o(L)$ mathematically forces:
\[
    \frac{Z^2 T}{8\sigma^2} \{1+o(1)\} < L + o(L).
\]
Rearranging this to isolate $Z^2$ gives:
\[
    Z^2 \le \frac{8\sigma^2 L}{T} \{1+o(1)\}.
\]
By the late-change assumption $T/L \to \infty$, we know that $L/T \to 0$. As a result, the right-hand side converges to zero. This implies that for all sufficiently large $A$ (and consequently large $T$), $Z^2 \le \frac{1}{4}(\nu-\mu)^2$, i.e., $|Z| \le \frac{1}{2}|\nu-\mu|$.

However, the true expected value of $Z$ under $\mathbb{P}^T_{P,Q}$ is $\nu - \mu$. For $|Z| \le \frac{1}{2}|\nu-\mu|$ to occur, the variable $Z$ must deviate from its expectation by at least $\frac{1}{2}|\nu-\mu|$. 

Under $\mathbb{P}^T_{P,Q}$, $Z$ is the difference of independent sub-Gaussian sample averages with variance proxy:
\[
    \sigma_*^2 = \frac{\sigma^2}{b_\alpha} + \frac{\sigma^2}{T} \le \frac{4\sigma^2}{T}.
\]
By Hoeffding's inequality for sub-Gaussian variables, the probability of this massive deviation is:
\[
    \mathbb{P}^T_{P,Q}\left( |Z - (\nu-\mu)| \ge \frac{1}{2}|\nu-\mu| \right) \le \exp\left( - \frac{(|\nu-\mu|/2)^2}{2\sigma_*^2} \right).
\]
Substituting $\sigma_*^2 \le 4\sigma^2/T$ into the exponent gives:
\[
    \exp\left( - \frac{(\nu-\mu)^2 / 4}{8\sigma^2 / T} \right) = \exp\left( - \frac{(\nu-\mu)^2 T}{32\sigma^2} \right).
\]
Recalling that $I = \frac{(\nu-\mu)^2}{2\sigma^2}$, the exponent is exactly $-I T / 16$. Thus, letting $c = 1/16$, we conclude:
\[
    \mathbb{P}^T_{P,Q}(Y > b_\alpha) \le \exp\left( - c I T \right).
\]
Since $IT \gg L$, this probability decays exponentially faster than $\exp(-L)$ and contributes negligibly to the expected delay.

\textit{Step D (Combining the bounds):} 
 Substituting these bounds into
\eqref{eq:expectation_split}, we obtain
\begin{align}
\mathbb E^T_{P,Q}[Y]
\le{}&
\left\lceil
(1+\epsilon)\frac{L}{I}
\right\rceil
+
T\exp\{-c_\epsilon L+o(L)\}
+
C'_{\rho,\sigma}
\exp\left(2L-\frac{IT}{16}\right).
\label{eq:expectation-combined}
\end{align}

We next show that the last two terms in
\eqref{eq:expectation-combined} vanish. Write the $o(L)$ term in the
second term as $r_{\alpha,\epsilon}$, where, for every fixed
$\epsilon>0$,
\[
\frac{r_{\alpha,\epsilon}}{L}\longrightarrow 0
\qquad\text{as }\alpha\downarrow0.
\]
Since $c_\epsilon>0$ is fixed, for all sufficiently small $\alpha$,
$r_{\alpha,\epsilon}
\le
\frac{c_\epsilon}{2}L.$
Therefore,
\begin{align*}
T\exp\{-c_\epsilon L+r_{\alpha,\epsilon}\}
&\le
T\exp\left(-\frac{c_\epsilon}{2}L\right)\\
&=
\exp\left(
\log T-\frac{c_\epsilon}{2}L
\right)\\
&=
\exp\left(
\log T-
\frac{\epsilon^2}{8(1+\epsilon)}L
\right).
\end{align*}
Because $\log T=o(L)$,
$\log T-
\frac{\epsilon^2}{8(1+\epsilon)}L
=
-\frac{\epsilon^2}{8(1+\epsilon)}L+o(L)
\longrightarrow-\infty.$
Hence
\[
T\exp\{-c_\epsilon L+o(L)\}\longrightarrow0.
\]

For the third term, since $I>0$ is fixed and $T/L\to\infty$,
$2L-\frac{IT}{16}
=
L\left(
2-\frac{I}{16}\frac{T}{L}
\right)
\longrightarrow-\infty.$
Consequently,
\[
C'_{\rho,\sigma}
\exp\left(2L-\frac{IT}{16}\right)
\longrightarrow0.
\]
It follows from \eqref{eq:expectation-combined} that, for every fixed
$\epsilon>0$,
\[
\mathbb E^T_{P,Q}[Y]
\le
\left\lceil
(1+\epsilon)\frac{L}{I}
\right\rceil
+o(1).
\]
Using $\lceil x\rceil\le x+1$, we obtain
\[
\mathbb E^T_{P,Q}[Y]
\le
(1+\epsilon)\frac{L}{I}+1+o(1).
\]
Dividing both sides by $L/I$ gives
\[
\frac{\mathbb E^T_{P,Q}[Y]}{L/I}
\le
1+\epsilon+\frac{I}{L}
+\frac{o(1)}{L/I}.
\]
Since $L\to\infty$,
\[
\limsup_{\alpha\downarrow0}
\frac{\mathbb E^T_{P,Q}[Y]}{L/I}
\le
1+\epsilon.
\]
The preceding inequality holds for every fixed $\epsilon>0$.
Therefore,
\[
\limsup_{\alpha\downarrow0}
\frac{\mathbb E^T_{P,Q}[Y]}{L/I}
\le
\inf_{\epsilon>0}(1+\epsilon)
=
1.
\]
Equivalently,
\[
\mathbb E^T_{P,Q}
\left[(\tau_A-T)^+\right]
\le
(1+o(1))\frac{L}{I}.
\]

Finally, by the local false-alarm guarantee \eqref{eq:lfa},
\[
\mathbb P^T_{P,Q}(\tau_A\le T)
=
\mathbb P^\infty_P(\tau_A\le T)
\le
\frac{T}{A}
=
\exp(\log T-L)
\longrightarrow0,
\]
where the last convergence follows from $\log T=o(L)$. Hence
\[
\mathbb P^T_{P,Q}(\tau_A>T)=1-o(1).
\]
Since $Y=(\tau_A-T)^+$, we have the exact identity
\[
\mathbb E^T_{P,Q}[Y]
=
\mathbb P^T_{P,Q}(\tau_A>T)
\,
\mathbb E^T_{P,Q}
[\tau_A-T\mid \tau_A>T].
\]
It follows that
\begin{align*}
\operatorname{CADD}_{P,Q}^{T}(\tau_A)
=
\mathbb E^T_{P,Q}
[\tau_A-T\mid \tau_A>T]
&=
\frac{
\mathbb E^T_{P,Q}[(\tau_A-T)^+]
}{
\mathbb P^T_{P,Q}(\tau_A>T)
}\\
&\le
\frac{
(1+o(1))L/I
}{
1-o(1)
}\\
&=
(1+o(1))\frac{L}{I}.
\end{align*}
\end{proof}

\begin{proof}[Proof of \Cref{thm:subg-upper-pfa}]
Fix $\varepsilon>0$.  Write $T=T_\alpha$ and
\[
  H_\alpha=L+h_\pi(T_\alpha),
  \qquad
  d=d_\alpha
  =
  \left\lceil(1+\varepsilon)\frac{H_\alpha}{I}\right\rceil,
  \qquad
  m=\lfloor T/2\rfloor.
\]
Assumption \eqref{eq:late-assumption-pfa} gives $d/T\to 0$, so for all sufficiently small $\alpha$, $m+d\le 2T$.  Therefore the deterministic penalty in Lemma~\ref{lem:weighted-sufficient-stopping} is bounded by $h_\pi(T)$.

Let
\[
  \bar X_{\rm pre}=\bar X_{T-m+1:T},
  \qquad
  \bar X_{\rm post}=\bar X_{T+1:T+d}.
\]
Then
\[
  \bar X_0=\bar X_{\rm post},
  \qquad
  \bar X_1=\frac{m\bar X_{\rm pre}+d\bar X_{\rm post}}{m+d},
\]
so
\begin{equation}
\label{eq:mean-sep-weighted}
  \bar X_0-\bar X_1
  =
  \frac{m}{m+d}(\bar X_{\rm post}-\bar X_{\rm pre}).
\end{equation}
Since $m/d\to\infty$,
\[
  \frac{m}{m+d}=1+o(1).
\]
Sub-Gaussian concentration gives
\[
  \bar X_{\rm pre}=\mu+O_{\bbP}(m^{-1/2}),
  \qquad
  \bar X_{\rm post}=\nu+O_{\bbP}(d^{-1/2}),
\]
and hence
\begin{equation}
\label{eq:sep-weighted-delta}
  \bar X_0-\bar X_1
  =
  (\nu-\mu)\{1+o_{\bbP}(1)\}.
\end{equation}
Also,
\[
  \kappa_d^{-1/2}
  =
  \sqrt{\frac{2\sigma^2}{d}}\{1+o(1)\},
  \qquad
  \kappa_{m+d}^{-1/2}
  =
  \sqrt{\frac{2\sigma^2}{m+d}}\{1+o(1)\}
  =o(d^{-1/2}).
\]
Thus
\[
  B_{m,d}
  =
  I d\{1+o_{\bbP}(1)\}
  =
  (1+\varepsilon)H_\alpha\{1+o_{\bbP}(1)\}.
\]
By Lemma~\ref{lem:weighted-sufficient-stopping},
\[
  \log D_{T+d}^{\mathrm{PFA}}
  \ge
  B_{m,d}-h_\pi(T)
  =
  (1+\varepsilon)H_\alpha\{1+o_{\bbP}(1)\}-h_\pi(T).
\]
Since $H_\alpha=L+h_\pi(T)$, the right side is at least $L$ with probability tending to one.  Hence
\[
  \bbP^T_{P,Q}(\tau_\alpha^{\mathrm{PFA}}\le T+d)\to 1.
\]
This proves that for every fixed $\varepsilon>0$,
\begin{equation}
\label{eq:upper-probability}
  \bbP^{T_\alpha}_{P,Q}
  \left(
    \tau_\alpha^{\mathrm{PFA}}-T_\alpha
    \le
    (1+\varepsilon)\frac{L+h_\pi(T_\alpha)}{I}
  \right)
  \to 1.
\end{equation}
\eqref{eq:upper-op} and \eqref{eq:upper-OlogT} follow immediately.
\end{proof}

\subsection{Bounded mean change example}
\label{app:bounded-UP}
Throughout this appendix, $\theta\in(0,1)$ and
$X_i\in[0,1]$.  Define
\begin{equation}
\label{eq:UP-one-step-factor}
 f_{\theta,p}(x)
 :=1+\frac{p-\theta}{\theta(1-\theta)}(x-\theta)
 =x\frac p\theta+(1-x)\frac{1-p}{1-\theta},
 \qquad p\in[0,1],
\end{equation}
and
\[
 W_{s:t}^{\theta,p}:=\prod_{i=s}^t f_{\theta,p}(X_i),
 \qquad
 U_{s:t}^{\theta}:=\int_0^1W_{s:t}^{\theta,p}\,\Pi_J(dp).
\]
For $u,v\in(0,1)$, write
\begin{equation}
\label{eq:binary-kl}
    \kl(u,v)
  :=u\log\frac uv+(1-u)\log\frac{1-u}{1-v}.
\end{equation}

\subsubsection{Validity, growth, and variance adaptivity proofs}
\label{app:prop-bddmean}
\begin{lemma}[Legal-bet validity]
\label{lem:UP-validity}
If $\E[X_i\mid\mathcal F_{i-1}]=\theta$, then for every fixed $p\in[0,1]$,
$(W_{s:t}^{\theta,p})_{t\ge s-1}$ is a nonnegative test martingale, and
$(U_{s:t}^{\theta})_{t\ge s-1}$ is a nonnegative test martingale.
\end{lemma}
\begin{proof}
The two endpoint values in \eqref{eq:UP-one-step-factor} are
$(1-p)/(1-\theta)$ and $p/\theta$, so the factor is nonnegative on
$[0,1]$.  Moreover,
\[
 \E[f_{\theta,p}(X_i)\mid\mathcal F_{i-1}]
 =1+\frac{p-\theta}{\theta(1-\theta)}
   \{\E[X_i\mid\mathcal F_{i-1}]-\theta\}=1.
\]
Multiplication gives the first claim.  Tonelli's theorem gives the second.
\end{proof}

\begin{lemma}[Pathwise regret against the best constant portfolio]
\label{lem:UP-pathwise-regret}
For every data sequence $x_1,\ldots,x_n\in[0,1]$,
\begin{equation}
\label{eq:UP-pathwise-regret}
 U_{1:n}^{\theta}
 \ge \frac{2}{\pi(n+1)}
      \sup_{p\in[0,1]}W_{1:n}^{\theta,p}.
\end{equation}
Consequently the log-regret is at most
$\log(n+1)+\log(\pi/2)$.
\end{lemma}
\begin{proof}
The arcsine density is at least $2/\pi$ on $[0,1]$, hence
\[
 U_{1:n}^{\theta}\ge\frac2\pi\int_0^1W_{1:n}^{\theta,p}\,dp.
\]
Each factor in \eqref{eq:UP-one-step-factor} is a nonnegative linear
combination of $p$ and $1-p$.  Expanding their product gives a Bernstein
polynomial
\[
 W_{1:n}^{\theta,p}
 =\sum_{k=0}^n a_k{n\choose k}p^k(1-p)^{n-k},
 \qquad a_k\ge0.
\]
Since the Bernstein basis sums to one,
$\sup_pW_{1:n}^{\theta,p}\le\max_k a_k$.  Since each basis element integrates
to $1/(n+1)$,
\[
 \int_0^1W_{1:n}^{\theta,p}\,dp
 =\frac1{n+1}\sum_{k=0}^na_k
 \ge\frac1{n+1}\max_ka_k\geq \frac1{n+1}\sup_pW_{1:n}^{\theta,p}.
\]
Combining the displays proves \eqref{eq:UP-pathwise-regret}.
\end{proof}

The next direct lower bound will be useful for prefix anchoring and for
controlling the expectation of the detection delay.  Put
$S_n=\sum_{i=1}^nX_i$ and $\bar X_n=S_n/n$.

\begin{lemma}[Direct KL lower bound for the universal portfolio]
\label{lem:UP-direct-KL}
There is a numerical constant $c_J>0$ such that, for all $n\ge1$, all data
$x_1,\ldots,x_n\in[0,1]$, and every $\theta\in(0,1)$,
\begin{equation}
\label{eq:UP-direct-KL}
 U_{1:n}^{\theta}
 \ge \frac{c_J}{\sqrt{n+1}}
       \exp\{n\kl(\bar X_n,\theta)\}.
\end{equation}
In particular, $U_{1:n}^{\theta}\ge c_J/\sqrt{n+1}$ pathwise.
\end{lemma}
\begin{proof}
Weighted AM--GM applied to \eqref{eq:UP-one-step-factor} gives
\[
 f_{\theta,p}(x)
 \ge \left(\frac p\theta\right)^x
      \left(\frac{1-p}{1-\theta}\right)^{1-x}.
\]
Multiplying and integrating against the arcsine law yields
\begin{equation}
\label{eq:UP-arcsine-integral}
 U_{1:n}^{\theta}
 \ge
 \frac{\Gamma(S_n+1/2)\Gamma(n-S_n+1/2)}
      {\pi\Gamma(n+1)\theta^{S_n}(1-\theta)^{n-S_n}}.
\end{equation}
\Cref{lem:uniform-robbins-stirling} imply, for every $z\in[0,n]$,
\[
 \frac{\Gamma(z+1/2)\Gamma(n-z+1/2)}{\pi\Gamma(n+1)}
 \ge \frac{c_J}{\sqrt{n+1}}
       \left(\frac zn\right)^z
       \left(1-\frac zn\right)^{n-z},
\]
with the conventions $0^0=1$.  One obtains a common constant by applying
the usual upper and lower Stirling bounds when $z,n-z\ge1$ and checking the
two compact boundary ranges $z\le1$ and $n-z\le1$ directly using continuity
and $\Gamma(1/2)=\sqrt\pi$.  Substitute $z=S_n$ into
\eqref{eq:UP-arcsine-integral} and take logarithms to obtain
\eqref{eq:UP-direct-KL}.
\end{proof}

\label{app:bounded-UP-information}
\begin{lemma}[Completed reverse information projection]
\label{lem:UP-ripr}
For every law $Q$ on $[0,1]$ and $\theta\in(0,1)$,
\[
 I_{\rm bet}(Q,\theta)
 =\inf_{R:\E_RX=\theta}\KL(Q\|R).
\]
\end{lemma}
\begin{proof}
For $\lambda\in\Lambda_\theta$, set
$f_\lambda(x)=1+\lambda(x-\theta)$.  If $\E_RX=\theta$, then
$\E_Rf_\lambda=1$.  The variational inequality for relative entropy gives
\[
 \E_Q\log f_\lambda
 \le\KL(Q\|R)+\log\E_Rf_\lambda
 =\KL(Q\|R).
\]
Taking the supremum over $\lambda$ and then the infimum over $R$ proves one
direction.

Let $\lambda^*$ maximize the concave function
$\lambda\mapsto\E_Q\log f_\lambda$.  If $\lambda^*$ is interior, first-order
optimality gives
\[
 \E_Q\frac{X-\theta}{f_{\lambda^*}(X)}=0.
\]
Since $f_{\lambda^*}=1+\lambda^*(X-\theta)$, this also implies
$\E_Q[1/f_{\lambda^*}(X)]=1$.  Define
$dR^*/dQ=1/f_{\lambda^*}$.  Then $R^*$ is a probability law,
$\E_{R^*}X=\theta$, and
\[
 \KL(Q\|R^*)=\E_Q\log f_{\lambda^*}=I_{\rm bet}(Q,\theta).
\]

If $\lambda^*=1/\theta$, then
$f_{\lambda^*}(X)=X/\theta$.  The left derivative condition gives
$\E_Q[\theta/X]\le1$.  Put mass with density $\theta/X$ relative to $Q$ and
place the remaining mass at zero.  The resulting probability law has mean
$\theta$ and again attains
$\E_Q\log(X/\theta)$.  The case
$\lambda^*=-1/(1-\theta)$ is symmetric: use density
$(1-\theta)/(1-X)$ and place the remaining mass at one.  If an endpoint
objective equals $-\infty$, it cannot be the maximizer because
$\E_QX\ne\theta$ admits a sufficiently small bet of the correct sign with
positive expected log wealth.  This completes the proof.
\end{proof}

\begin{lemma}[Worst-case and local betting information]
\label{lem:UP-minimax-local}
Let \(Q\) be a distribution on \([0,1]\), and write
\[
    b:=\E_Q[X],
    \qquad
    v_Q:=\text{Var}_Q(X).
\]
Then, for every \(\theta\in(0,1)\),
\[
    I_{\rm bet}(Q,\theta)\ge \kl(b,\theta),
    \qquad
    \inf_{Q:\,\E_Q[X]=b} I_{\rm bet}(Q,\theta)
    =
    \kl(b,\theta).
\]
\end{lemma}

\begin{proof}
Recall that
\[
    I_{\rm bet}(Q,\theta)
    =
    \sup_{\lambda\in\Lambda_\theta}
    \E_Q\log\{1+\lambda(X-\theta)\},
    \qquad
    \Lambda_\theta
    =
    \left[-\frac{1}{1-\theta},\frac{1}{\theta}\right].
\]

Choosing \(p=b\) in \eqref{eq:UP-one-step-factor} and applying weighted
AM--GM gives
\[
    \log f_{\theta,b}(x)
    \ge
    x\log\frac{b}{\theta}
    +(1-x)\log\frac{1-b}{1-\theta}.
\]
Taking \(Q\)-expectations yields
\[
    I_{\rm bet}(Q,\theta)\ge \kl(b,\theta).
\]
For \(Q=\operatorname{Bern}(b)\), the objective equals
\[
    b\log\{1+\lambda(1-\theta)\}
    +(1-b)\log(1-\lambda\theta),
\]
whose maximizer is
\[
    \lambda=\frac{b-\theta}{\theta(1-\theta)}.
\]
Its maximum value is \(\kl(b,\theta)\). The cases \(b\in\{0,1\}\) follow
directly. Hence
\[
    \inf_{Q:\,\E_QX=b}I_{\rm bet}(Q,\theta)=\kl(b,\theta).
\]

\end{proof}

\begin{proof}[Proof of
\Cref{prop:bounded-UP-growth-main}]
E-process validity is \Cref{lem:UP-validity}.  For any fixed $p\in(0,1)$, the function
$x\mapsto\log f_{\theta,p}(x)$ is bounded and continuous on $[0,1]$, so the
strong law gives
\[
 \frac1n\log W_{1:n}^{\theta,p}
 \longrightarrow \E_Q\log f_{\theta,p}(X)
 \qquad Q^\infty\text{-a.s.}
\]
Together with \Cref{lem:UP-pathwise-regret}, this implies
\[
 \liminf_{n\to\infty}\frac1n\log U_{1:n}^{\theta}
 \ge \sup_{p\in(0,1)}\E_Q\log f_{\theta,p}(X)
 =I_{\rm bet}(Q,\theta).
\]
The last equality follows because $p\mapsto\lambda_\theta(p)$ maps
$[0,1]$ onto $\Lambda_\theta$ and interior $p$ values approximate either
endpoint.

The last two lemmas prove all information claims in
\Cref{prop:bounded-UP-growth-main}.

\end{proof}

\subsubsection{Detection-delay proofs}
\label{app:bounded-UP-delay}


\begin{lemma}[Uniform reduction to one post-change random walk]
\label{lem:UP-one-good-bet-uniform}
Fix $a\in(0,1)$ and a law $Q$ on $[0,1]$ such that
$I:=I_{\rm bet}(Q,a)>0$.  For every $\eta\in(0,I/2)$ there exist
$p_\eta\in(0,1)$, $\delta_\eta\in(0,\min\{a,1-a\})$, and a bounded
measurable function $Y_\eta:[0,1]\to\mathbb R$ such that, writing
\[
 U_\eta=[a-\delta_\eta,a+\delta_\eta],
 \qquad
 \mu_\eta:=\E_QY_\eta(X),
\]
we have
\begin{equation}
\label{eq:UP-one-good-bet-drift}
 \mu_\eta\ge I-2\eta>0,
\end{equation}
and, for every $n\ge1$, every data sequence $x_1,\ldots,x_n\in[0,1]$,
and every $\theta\in U_\eta$,
\begin{equation}
\label{eq:UP-one-good-bet-pathwise}
 U_{1:n}^{\theta}(x_1,\ldots,x_n)
 \ge \frac{2}{\pi(n+1)}
       \exp\left\{\sum_{i=1}^nY_\eta(x_i)\right\}.
\end{equation}
\end{lemma}

\begin{proof}
By the definition of $I_{\rm bet}(Q,a)$, choose an interior portfolio
$p_\eta\in(0,1)$ such that
\begin{equation}
\label{eq:UP-good-p-choice}
 \E_Q\log f_{a,p_\eta}(X)\ge I-\eta.
\end{equation}
It is enough to optimize over interior portfolios.  Indeed, if an endpoint
$p\in\{0,1\}$ has finite positive expected log wealth, then for
$p_\rho=(1-\rho)p+\rho a\in(0,1)$, weighted AM--GM gives
\[
 f_{a,p_\rho}(x)
 =(1-\rho)f_{a,p}(x)+\rho
 \ge f_{a,p}(x)^{1-\rho},
\]
so $\E_Q\log f_{a,p_\rho}(X)$ approaches the endpoint value from below as
$\rho\downarrow0$.  An endpoint with expected log wealth $-\infty$ cannot be
needed when $I>0$.

Because $p_\eta\in(0,1)$, the function
$(\theta,x)\mapsto\log f_{\theta,p_\eta}(x)$ is continuous and bounded on a
compact neighborhood of $\{a\}\times[0,1]$.  Shrinking
$\delta_\eta>0$ if necessary, we may arrange that
\begin{equation}
\label{eq:UP-uniform-continuity}
 \sup_{\substack{|\theta-a|\le\delta_\eta\\x\in[0,1]}}
 \left|
   \log f_{\theta,p_\eta}(x)-\log f_{a,p_\eta}(x)
 \right|
 \le\eta.
\end{equation}
Define
\[
 Y_\eta(x):=
 \inf_{\theta\in U_\eta}\log f_{\theta,p_\eta}(x).
\]
This function is bounded and measurable.  By
\eqref{eq:UP-good-p-choice}--\eqref{eq:UP-uniform-continuity},
\[
 \mu_\eta
 \ge \E_Q\log f_{a,p_\eta}(X)-\eta
 \ge I-2\eta,
\]
which proves \eqref{eq:UP-one-good-bet-drift}.

For every $\theta\in U_\eta$,
\[
 W_{1:n}^{\theta,p_\eta}
 =\exp\left\{\sum_{i=1}^n
       \log f_{\theta,p_\eta}(x_i)\right\}
 \ge\exp\left\{\sum_{i=1}^nY_\eta(x_i)\right\}.
\]
Applying the pathwise regret bound in
\Cref{lem:UP-pathwise-regret} and retaining the single portfolio
$p_\eta$ proves \eqref{eq:UP-one-good-bet-pathwise}.
\end{proof}

\begin{lemma}[A bounded random walk crossing a logarithmic boundary]
\label{lem:UP-log-boundary-hitting}
Let $Y_1,Y_2,\ldots$ be i.i.d., with
$y_-\le Y_i\le y_+$ almost surely and mean $\mu>0$.  Put
$S_n=\sum_{i=1}^nY_i$.  For $B>0$ and a fixed $c\ge0$, define
\begin{equation}
\label{eq:UP-log-boundary-stopping}
 \sigma_B:=\inf\{n\ge1:S_n\ge B+\log(n+1)+c\}.
\end{equation}
Then $\sigma_B<\infty$ almost surely, $\E\sigma_B<\infty$, and
\begin{equation}
\label{eq:UP-log-boundary-expectation}
 \frac{\E\sigma_B}{B}\longrightarrow\frac1\mu.
\end{equation}
Moreover, for every fixed $\varepsilon>0$, there is $c_\varepsilon>0$ such
that, for all sufficiently large $B$,
\begin{equation}
\label{eq:UP-log-boundary-fixed-tail}
 \bbP\left\{
   \sigma_B>
   \left\lceil(1+\varepsilon)\frac{B}{\mu}\right\rceil
 \right\}
 \le e^{-c_\varepsilon B}.
\end{equation}
If $H_B/B\to\infty$, then there is $c_H>0$ such that, for all sufficiently
large $B$,
\begin{equation}
\label{eq:UP-log-boundary-long-tail}
 \bbP(\sigma_B>H_B)\le e^{-c_HH_B}.
\end{equation}
\end{lemma}

\begin{proof}
The strong law gives $S_n/n\to\mu$ almost surely, while
$\{B+\log(n+1)+c\}/n\to0$.  Hence $\sigma_B<\infty$ almost surely.
To see integrability explicitly, suppose first that $y_+>y_-$ and put
$r_Y=y_+-y_-$.  For every sufficiently large multiple $n$ of $B$,
$B+\log(n+1)+c\le \mu n/2$.  Hence
\[
 \bbP(\sigma_B>n)
 \le \bbP(S_n< B+\log(n+1)+c)
 \le \bbP(S_n-\mu n\le-\mu n/2)
 \le \exp\left\{-\frac{\mu^2n}{2r_Y^2}\right\}.
\]
The tail-sum formula now gives $\E\sigma_B<\infty$.  If $y_+=y_-$,
the random walk is deterministic and all the claims are immediate.

Since the boundary in \eqref{eq:UP-log-boundary-stopping} is at least $B$,
Wald's identity yields
\[
 \mu\E\sigma_B=\E S_{\sigma_B}\ge B,
\]
and therefore
\begin{equation}
\label{eq:UP-hitting-lower}
 \liminf_{B\to\infty}\frac{\E\sigma_B}{B}\ge\frac1\mu.
\end{equation}
For the converse, fix $\xi\in(0,\mu)$ and put
\[
 C_{\xi,c}:=
 \sup_{n\ge1}\{\log(n+1)+c-\xi n\}<\infty.
\]
Let $Z_i=Y_i-\xi$, whose mean is $\mu-\xi>0$, and define
\[
 \rho_B:=\inf\left\{n\ge1:
   \sum_{i=1}^nZ_i\ge B+C_{\xi,c}
 \right\}.
\]
At time $\rho_B$,
\[
 S_{\rho_B}
 =\sum_{i=1}^{\rho_B}Z_i+\xi\rho_B
 \ge B+C_{\xi,c}+\xi\rho_B
 \ge B+\log(\rho_B+1)+c,
\]
so $\sigma_B\le\rho_B$ pathwise.  If
$z_+:=(y_+-\xi)_+$, the overshoot of the $Z$-random walk is at most $z_+$.
Wald's identity therefore gives
\[
 (\mu-\xi)\E\rho_B
 =\E\sum_{i=1}^{\rho_B}Z_i
 \le B+C_{\xi,c}+z_+.
\]
Thus
\[
 \limsup_{B\to\infty}\frac{\E\sigma_B}{B}
 \le\frac1{\mu-\xi}.
\]
Letting $\xi\downarrow0$ and combining with
\eqref{eq:UP-hitting-lower} proves
\eqref{eq:UP-log-boundary-expectation}.

For \eqref{eq:UP-log-boundary-fixed-tail}, put
$n_B=\lceil(1+\varepsilon)B/\mu\rceil$.  Since
$\log(n_B+1)+c=o(B)$, for all sufficiently large $B$,
\[
 \mu n_B-
 \{B+\log(n_B+1)+c\}\ge \frac{\varepsilon B}{2}.
\]
The implication
$\{\sigma_B>n_B\}\subseteq
 \{S_{n_B}<B+\log(n_B+1)+c\}$ and Hoeffding's inequality give
\eqref{eq:UP-log-boundary-fixed-tail}.  If $H_B/B\to\infty$, then
$\{B+\log(H_B+1)+c\}/H_B\to0$.  Hence, eventually, failure to cross by
$H_B$ implies $S_{H_B}-\mu H_B\le-\mu H_B/2$, and another application of
Hoeffding proves \eqref{eq:UP-log-boundary-long-tail}.
\end{proof}

\begin{proof}[Proof of \Cref{thm:bounded-main-delay}]
Write
\[
 I:=I_{\rm bet}(Q,a)>0.
\]
We first prove a generic weighted high-probability bound that contains the
ARL and PFA claims as special cases.  Let $A_\gamma\to\infty$,
$L_\gamma=\log A_\gamma$, let $T_\gamma$ be the changepoint, and let
$w_{s,\gamma}\in(0,1]$ be deterministic start weights.  Define the effective
post-change boundary
\begin{equation}
\label{eq:UP-effective-boundary-detailed}
 B_\gamma
 :=L_\gamma-\log w_{T_\gamma+1,\gamma}.
\end{equation}
Assume
\begin{equation}
\label{eq:UP-generic-late-change-assumptions}
 \frac{T_\gamma}{B_\gamma}\longrightarrow\infty,
 \qquad
 L_\gamma-\log w_{1,\gamma}=o(T_\gamma).
\end{equation}
For the ARL detector, $w_s\equiv1$ and $B_\gamma=L_\gamma$.  For the PFA
detector, $w_s=\pi_s$ and
$B_\gamma=\log(1/\alpha)-\log\pi_{T_\alpha+1}$; because $\pi_1$ is a fixed
positive constant, the second condition in
\eqref{eq:UP-generic-late-change-assumptions} follows from the first.

Fix $\varepsilon>0$.  Choose $\eta\in(0,I/2)$ sufficiently small that
\begin{equation}
\label{eq:UP-eta-choice-detailed}
 (1+\varepsilon)(I-2\eta)>I.
\end{equation}
Apply \Cref{lem:UP-one-good-bet-uniform}, and abbreviate
$p_\eta=p$, $\delta_\eta=\delta$, $U_\eta=U$,
$Y_\eta=Y$, and $\mu_\eta=\mu$.  In particular,
$\mu\ge I-2\eta>0$.  Under the post-change law, define
\[
 S_n:=\sum_{i=1}^nY(X_{T_\gamma+i}),
 \qquad
 c_0:=\log(\pi/2),
\]
and the post-change crossing time
\begin{equation}
\label{eq:UP-post-crossing-time-detailed}
 \sigma_\gamma
 :=\inf\{n\ge1:
   S_n\ge B_\gamma+\log(n+1)+c_0
 \}.
\end{equation}
By \Cref{lem:UP-one-good-bet-uniform}, at time
$T_\gamma+\sigma_\gamma$, for every $\theta\in U$,
\begin{align}
 \log\left\{
   w_{T_\gamma+1,\gamma}
   U_{T_\gamma+1:T_\gamma+\sigma_\gamma}^{\theta}
 \right\}
 &\ge
 \log w_{T_\gamma+1,\gamma}
 -\log(\sigma_\gamma+1)-c_0+S_{\sigma_\gamma}
 \\
 &\ge L_\gamma.
\label{eq:UP-inner-threshold-detailed}
\end{align}
Thus the pure post-change start rules out every candidate mean in $U$ once
$\sigma_\gamma$ occurs.

It remains to anchor the candidates outside $U$.  Put $T=T_\gamma$ and
consider the prefix event
\begin{equation}
\label{eq:UP-prefix-good-event}
 G_T:=\left\{|\bar X_{1:T}-a|\le\delta/4\right\}.
\end{equation}
Hoeffding's inequality gives
\begin{equation}
\label{eq:UP-prefix-good-probability}
 \bbP_{P,Q}^{T}(G_T^c)
 =\bbP_{P}^{\infty}(G_T^c)
 \le2\exp(-\delta^2T/8).
\end{equation}
Let
\begin{equation}
\label{eq:UP-anchor-horizon}
 H_T:=\left\lfloor\frac{\delta T}{4}\right\rfloor.
\end{equation}
On $G_T$, for every $0\le n\le H_T$, boundedness of the observations gives
\begin{align}
 |\bar X_{1:T+n}-a|
 &\le
 \frac{T}{T+n}|\bar X_{1:T}-a|
 +\frac1{T+n}\left|\sum_{i=T+1}^{T+n}(X_i-a)\right|
 \\
 &\le \frac{\delta}{4}+\frac nT
 \le\frac\delta2.
\label{eq:UP-cross-average-deterministic}
\end{align}
Consequently, if $\theta\notin U$, then
$|\theta-\bar X_{1:T+n}|\ge\delta/2$.  Pinsker's inequality for Bernoulli
relative entropy yields
\[
 \kl(\bar X_{1:T+n},\theta)
 \ge2|\bar X_{1:T+n}-\theta|^2
 \ge\frac{\delta^2}{2}.
\]
The direct KL lower bound in \Cref{lem:UP-direct-KL} therefore gives,
uniformly over $\theta\notin U$ and $0\le n\le H_T$,
\begin{equation}
\label{eq:UP-outer-anchor-detailed}
 w_{1,\gamma}U_{1:T+n}^{\theta}
 \ge
 \frac{w_{1,\gamma}c_J}
      {\sqrt{(1+\delta/4)T+1}}
 \exp\left\{\frac{\delta^2T}{2}\right\}.
\end{equation}
Here and in the detector definition, the endpoint values $\theta\in\{0,1\}$
are understood by lower-semicontinuous extension.  On $G_T$ the empirical
mean is bounded away from both endpoints, so
$\kl(\bar X_{1:T+n},\theta)\to\infty$ as $\theta$ approaches an
endpoint; hence the same uniform conclusion includes $\theta=0,1$.

By the second condition in
\eqref{eq:UP-generic-late-change-assumptions}, the logarithm of the
right-hand side of \eqref{eq:UP-outer-anchor-detailed} is at least
$L_\gamma$ for all sufficiently large $\gamma$.  Hence, on $G_T$, every
candidate $\theta\notin U$ is ruled out throughout the entire time interval
$[T,T+H_T]$ by the start-$1$ term.

Combining \eqref{eq:UP-inner-threshold-detailed} and
\eqref{eq:UP-outer-anchor-detailed} gives the following pathwise implication
for all sufficiently large $\gamma$:
\begin{equation}
\label{eq:UP-pathwise-delay-reduction}
 G_T\cap\{\sigma_\gamma\le H_T\}
 \quad\Longrightarrow\quad
 \tau_{A_\gamma}^{\rm UP,w}
 \le T_\gamma+\sigma_\gamma.
\end{equation}
This is the non-partitioned analogue of reducing an SR detector to one good
post-change betting component: the extra prefix term is used only to rule
out candidate means away from $a$.

Set
\[
 d_\gamma
 :=\left\lceil(1+\varepsilon)\frac{B_\gamma}{I}\right\rceil.
\]
The first condition in
\eqref{eq:UP-generic-late-change-assumptions} implies
$d_\gamma\le H_T$ eventually.  By \eqref{eq:UP-eta-choice-detailed},
$\varepsilon_\eta:=(1+\varepsilon)\mu/I-1$ is strictly positive.
Thus, up to the immaterial integer rounding,
$d_\gamma\ge(1+\varepsilon_\eta)B_\gamma/\mu$.
\Cref{lem:UP-log-boundary-hitting} therefore gives
\[
 \bbP_{P,Q}^{T_\gamma}(\sigma_\gamma>d_\gamma)\to0.
\]
Together with \eqref{eq:UP-prefix-good-probability} and
\eqref{eq:UP-pathwise-delay-reduction}, this yields
\begin{equation}
\label{eq:UP-generic-high-probability-delay}
 \bbP_{P,Q}^{T_\gamma}
 \left\{
   \tau_{A_\gamma}^{\rm UP,w}-T_\gamma
   \le
   \left\lceil(1+\varepsilon)\frac{B_\gamma}{I}\right\rceil
 \right\}
 \longrightarrow1.
\end{equation}
Taking $w_s\equiv1$ proves the ARL high-probability claim.  Taking
$A_\gamma=1/\alpha$ and $w_s=\pi_s$ proves the PFA claim.

We now prove the ARL expectation and CADD statements.  Here $w_s\equiv1$,
$L=\log A$, and $B=L$.  The pathwise lower bound in
\Cref{lem:UP-direct-KL} implies, for a numerical constant $c_*>0$,
\[
 D_t^{\rm UP,ARL}
 \ge c_J\sum_{n=1}^t(n+1)^{-1/2}
 \ge c_*\sqrt t.
\]
Thus there is a numerical $C_*<\infty$ such that
\begin{equation}
\label{eq:UP-deterministic-cap-detailed}
 \tau_A^{\rm UP,ARL}\le C_*A^2
 \qquad\text{on every sample path.}
\end{equation}
Let $N_A=C_*A^2$.  From
\eqref{eq:UP-pathwise-delay-reduction} and
\eqref{eq:UP-deterministic-cap-detailed},
\begin{equation}
\label{eq:UP-expectation-decomposition-detailed}
 \E_{P,Q}^{T}\big[(\tau_A^{\rm UP,ARL}-T)^+\big]
 \le
 \E_Q\sigma_L
 +N_A\left\{
   \bbP_{P,Q}^{T}(G_T^c)+\bbP_Q(\sigma_L>H_T)
 \right\}.
\end{equation}
The first error probability is bounded by
$2e^{-\delta^2T/8}$.  Since $H_T\asymp T$ and $T/L\to\infty$,
\Cref{lem:UP-log-boundary-hitting} gives
$\bbP_Q(\sigma_L>H_T)\le e^{-cH_T}$ for some $c>0$ and all sufficiently
large $A$.  Therefore
\[
 N_A\left\{
   \bbP_{P,Q}^{T}(G_T^c)+\bbP_Q(\sigma_L>H_T)
 \right\}
 \longrightarrow0,
\]
because $\log N_A=2L+O(1)=o(T)$.  On the other hand,
\Cref{lem:UP-log-boundary-hitting} gives
\[
 \frac{\E_Q\sigma_L}{L}\longrightarrow\frac1\mu
 \le\frac1{I-2\eta}.
\]
Since $\eta>0$ can be arbitrarily small,
\begin{equation}
\label{eq:UP-unconditional-delay-detailed}
 \limsup_{A\to\infty}
 \frac{
   \E_{P,Q}^{T_A}[(\tau_A^{\rm UP,ARL}-T_A)^+]
 }{\log A}
 \le\frac1I.
\end{equation}
Finally, the finite-horizon false-alarm inequality in
\Cref{thm:arl-control} gives
\[
 \bbP_{P,Q}^{T_A}(\tau_A^{\rm UP,ARL}>T_A)
 \ge1-\frac{T_A}{A}=1-o(1),
\]
where the last step uses $\log T_A=o(\log A)$.  Since
\[
 \CADD_{P,Q}^{T_A}(\tau_A^{\rm UP,ARL})
 =
 \frac{
   \E_{P,Q}^{T_A}[(\tau_A^{\rm UP,ARL}-T_A)^+]
 }{
   \bbP_{P,Q}^{T_A}(\tau_A^{\rm UP,ARL}>T_A)
 },
\]
\eqref{eq:UP-unconditional-delay-detailed} proves the CADD assertion.

By \Cref{lem:UP-ripr}, $I_{\rm bet}(Q,a)$ is the reverse information
projection of $Q$ onto the bounded mean-$a$ class.  Consequently, when the
delay criterion is worst-cased over pre-change laws with mean $a$, the
usual change-of-measure lower bound has the same leading constant.  For a
fixed fully specified pre-change law $P$, the sharper information number
$\KL(Q\|P)$ may be larger.  Finally, by
\Cref{lem:UP-minimax-local}, worst-casing also over post-change laws with
mean $b$ replaces $I_{\rm bet}(Q,a)$ by $\kl(b,a)$, attained by the
Bernoulli pair.
\end{proof}

\subsection{Gaussian unknown-variance mean-change example}
\label{app:proof-gaussian}

\begin{proof}[Proof of \Cref{prop:UI-t-eprocess}]
For a fixed $r>0$, let
\[
  q_{s:i-1}(x)
  :=\phi_{\widetilde\mu_{s:i-1},
           \widetilde\sigma^2_{s:i-1}}(x)
\]
and define
\[
  \Lambda_{s:t}^{\theta,r}
  :=\prod_{i=s}^t
  \frac{q_{s:i-1}(X_i)}{\phi_{\theta,r}(X_i)}.
\]
Under $N(\theta,r)^\infty$, conditional integration of the predictive density
$q_{s:i-1}$ shows that this is a test martingale.  Maximizing the Gaussian
null likelihood over $r>0$ gives
\[
  \sup_{r>0}\prod_{i=s}^t\phi_{\theta,r}(X_i)
  =(2\pi e\widehat r_\theta)^{-n/2},
  \qquad
  \widehat r_\theta
  =\frac1n\sum_{i=s}^t(X_i-\theta)^2.
\]
Therefore
\[
  R_{s:t}^{\theta}
  =\inf_{r>0}\Lambda_{s:t}^{\theta,r}.
\]
For the true null variance $r_0$, this implies
$R_{s:t}^{\theta}\le\Lambda_{s:t}^{\theta,r_0}$ at every time.  Hence
$R^\theta$ is dominated by a test martingale under each member of
$\mathcal N_\theta$, which proves the e-process assertion.

For the growth calculation, write
\begin{align*}
  \frac1n\log R_{1:n}^{\theta}
  =&\ \frac12\log\left\{
       \frac1n\sum_{i=1}^n(X_i-\theta)^2
     \right\}+\frac12
     -\frac1n\sum_{i=1}^n\log\widetilde\sigma_{i-1}\\
   &-\frac1{2n}\sum_{i=1}^n
       \frac{(X_i-\widetilde\mu_{i-1})^2}
            {\widetilde\sigma^2_{i-1}}.
\end{align*}
Under the conditions of \citet[Proposition~3.3]{wang2024anytime}, the last
two averages converge almost surely to
$\frac12\log v$ and $1/2$, respectively.  The first empirical moment
converges to $v+(m-\theta)^2$.  This proves
\eqref{eq:UI-t-growth}.  For the concrete estimators
\eqref{eq:gaussian-regularized-predictors}, the required convergence follows
from the strong law and standard positive and inverse moments of a chi-square
variable; the fixed ridge $v_0$ handles the finitely many small sample sizes.

Finally,
\[
  \KL\{N(m,v)\|N(\theta,r)\}
  =\frac12\left\{
    \log\frac rv+\frac{v+(m-\theta)^2}{r}-1
  \right\}.
\]
The right side is minimized at
$r=v+(m-\theta)^2$, giving
\eqref{eq:t-reverse-I-projection}.
\end{proof}

\begin{proof}[Proof of \Cref{thm:t-delay-unified}]
Write $T=T_\alpha$, $B=B_\alpha^t$, $d=d_\alpha$, and $I=I_t(P,Q)$.
Then $d/T\to0$.  Choose $\eta>0$ so small that
$(1+\varepsilon)(I-2\eta)/I>1$.  By continuity of
$\theta\mapsto J_{b,v}(\theta)$, choose $\delta>0$ such that
\[
  \inf_{|\theta-a|\le\delta}J_{b,v}(\theta)\ge I-\eta.
\]
The compact-uniform version of the growth calculation in
\Cref{prop:UI-t-eprocess}, applied to the pure post-change block, gives
\[
  \inf_{|\theta-a|\le\delta}
  \log R_{T+1:T+d}^{\theta}
  \ge d(I-2\eta)
\]
with probability tending to one.  Hence, uniformly on this inner set,
\[
  \log\{w_{T+1,\alpha}R_{T+1:T+d}^{\theta}\}
  \ge L_\alpha
\]
eventually with probability tending to one, because
$d(I-2\eta)>B+o(B)$.

For the exterior set, \Cref{lem:UI-cross-stability} gives, with
$c_\delta=\frac12\log(1+\delta^2/u)>0$,
\[
  \inf_{|\theta-a|\ge\delta}
  \log R_{1:T+d}^{\theta}
  \ge (T+d)c_\delta/2
\]
with probability tending to one.  The final condition in
\eqref{eq:t-late-change-condition} then implies
\[
  \inf_{|\theta-a|\ge\delta}
  \log\{w_{1,\alpha}R_{1:T+d}^{\theta}\}
  \ge L_\alpha
\]
eventually.  Thus every $\theta$ has one detector summand at least
$A_\alpha$, proving \eqref{eq:t-delay-hp}.  The ARL and PFA displays follow
by substituting $w_s\equiv1$ and \eqref{eq:pi-weights}.  Conditional versions
follow by dividing by the probability of no pre-change alarm, which tends to
one by \Cref{thm:arl-control,thm:pfa-control} under the stated regimes.
\end{proof}

\begin{proof}[Proof of \Cref{thm:full-Gaussian-delay}]
The martingale assertion for \eqref{eq:full-Gaussian-martingale} follows from
\[
  \E_{N(\theta,r)}
  \left[
    \frac{q_{s:i-1}(X_i)}{\phi_{\theta,r}(X_i)}
    \biggm|\mathcal F_{i-1}
  \right]
  =\int q_{s:i-1}(x)\,dx=1.
\]
Under $N(m,v)$, consistency of the predictive density and the strong law
give
\[
  \frac1n\sum_{i=1}^n\log q_{i-1}(X_i)
  \to \E_{N(m,v)}\log\phi_{m,v}(X),
\]
whereas
$n^{-1}\sum_i\log\phi_{\theta,r}(X_i)
\to\E_{N(m,v)}\log\phi_{\theta,r}(X)$.  Their difference is
\eqref{eq:full-Gaussian-growth}.

We now prove the delay bound.  Let $I=\KL(Q\|P)$ and fix $\eta>0$ small
enough that $(1+\varepsilon)(I-2\eta)/I>1$.  By continuity of Gaussian KL,
there is a compact parameter neighborhood $U$ of $P=(a,u)$ such that
\begin{equation}
\label{eq:full-G-proof-inside}
  \inf_{R\in U}\KL(Q\|R)\ge I-\eta.
\end{equation}
Moreover,
\begin{equation}
\label{eq:full-G-proof-outside}
  c_U:=\inf_{R\notin U}\KL(P\|R)>0.
\end{equation}
To see this, use the explicit Gaussian KL formula: it is continuous, vanishes
only at $R=P$, and tends to infinity as the candidate mean diverges or the
candidate variance tends to $0$ or $\infty$.

On a pure i.i.d. Gaussian block, the explicit representation
\[
  \frac1n\log L_{1:n}^{\theta,r}
  =C_n+\frac12\log(2\pi r)
    +\frac{\widehat v_n+(\bar X_n-\theta)^2}{2r},
\]
where $C_n$ is independent of $(\theta,r)$, shows that the convergence in
\eqref{eq:full-Gaussian-growth} is uniform on compact parameter sets.  The
same display and coercivity show uniform convergence of the infimum over the
closed exterior $U^c$.  It remains valid for the $T+d$ cross block when
$d=o(T)$, because its empirical mean, variance, and predictive score converge
to their $P$ limits exactly as in \Cref{lem:UI-cross-stability}.  Therefore,
with probability tending to one,
\begin{align*}
  \inf_{R\in U}
  \log L_{T+1:T+d}^{R}&\ge d(I-2\eta),\\
  \inf_{R\notin U}
  \log L_{1:T+d}^{R}&\ge (T+d)c_U/2.
\end{align*}
The first bound, the definition
$d=\lceil(1+\varepsilon)B_\alpha^{\rm G}/I\rceil$, and
$B_\alpha^{\rm G}=L_\alpha-\log w_{T+1,\alpha}$ imply that the weighted
post-change summand exceeds $A_\alpha$ uniformly on $U$.  The second bound
and $L_\alpha-\log w_{1,\alpha}=o(T_\alpha)$ imply that the weighted cross
summand exceeds $A_\alpha$ uniformly on $U^c$.  Thus the infimum in
\eqref{eq:full-Gaussian-detector} is at least $A_\alpha$ at time $T+d$ with
probability tending to one, proving \eqref{eq:full-Gaussian-delay-hp}.
\end{proof}

\section{Adjusters and local REGROW witnesses}\label{app:proof-local-witnesses}

\begin{proof}[Proof of \Cref{lem:adjusted-max}]
Let $\tau$ be any almost surely finite stopping time.  By Ville's inequality applied to the stopped process,
\[
  \mathbb P(E_\tau^*\ge x)\le \frac1x,
  \qquad x\ge1.
\]
Thus $p_\tau=1/E_\tau^*$ is a p-value in the sense that
\[
  \mathbb P(p_\tau\le u)\le u,
  \qquad 0\le u\le1.
\]
Define $b(u)=a(1/u)$ for $u\in(0,1]$.  Since $a$ is increasing, $b$ is decreasing, and the change of variables $x=1/u$ gives
\[
  \int_0^1 b(u)\,du
  =\int_0^1 a(1/u)\,du
  =\int_1^\infty \frac{a(x)}{x^2}\,dx\le1.
\]
For completeness, we recall the calibration step.  If $p$ is superuniform and $b$ is decreasing, nonnegative, and integrable, then
\[
  \mathbb E b(p)
  =\int_0^\infty \mathbb P\{b(p)>y\}\,dy
  \le \int_0^\infty \lambda\{u\in[0,1]: b(u)>y\}\,dy
  = \int_0^1 b(u)\,du,
\]
where $\lambda$ denotes Lebesgue measure.  Applying this to $p_\tau$ yields
\[
  \mathbb E a(E_\tau^*)=\mathbb E b(p_\tau)\le1.
\]
Since this holds at every almost surely finite stopping time $\tau$, $(a(E_t^*))$ is an e-process.
\end{proof}

\begin{proof}[Proof of \Cref{cor:valid-retention}]
Validity follows by applying \cref{lem:adjusted-max} conditionally after the restart time $s$.  The inequalities follow from $M_{s:t}^{R,*}\ge M_{s:u}^R$ and monotonicity of $a$.
\end{proof}

\begin{proof}[Proof of \Cref{lem:neighborhoods}]
The map $(S,R)\mapsto\KL(S\|R)$ is lower semicontinuous for weak convergence on Polish spaces by the Donsker--Varadhan variational formula.  Hence $R\mapsto\KL(Q\|R)$ is weakly lower semicontinuous, and the strict superlevel set $G_\eta$ is weakly open.  Since $\KL(Q\|P)=I^*>I^*-\eta$, $P\in G_\eta$.

It remains to prove the positive separation outside an arbitrary weak neighborhood $U$ of $P$.  If $c_U=0$, there exist $R_n\in\calP\setminus U$ such that $\KL(P\|R_n)\to0$.  Pinsker's inequality gives $\|P-R_n\|_{\mathrm{TV}}\to0$, hence $R_n\Rightarrow P$ weakly.  Since $U$ is a weak neighborhood of $P$, eventually $R_n\in U$, a contradiction.  Thus $c_U>0$.  The displayed inequalities follow immediately from the definitions.
\end{proof}

\begin{proof}[Proof of \Cref{lem:weakly-compact-regrow}]
By the REGROW result of \citet[Theorem 2]{ram2026power}, there
exists a single e-process
\[
    E^K=\{E_n^K\}_{n\ge0},
    \qquad E_0^K=1,
\]
which is valid under every $R\in K$ and satisfies, for every
$Q\notin K$,
\begin{equation}
\label{eq:ram-ramdas-regrow}
    \liminf_{n\to\infty}
    \frac1n\log E_n^K
    \ge
    \Phi_K(Q),
    \qquad
    Q^\infty\text{-a.s.}
\end{equation}
Importantly, $E^K$ depends only on $K$, not on $Q$.

Let $a$ be a fixed growth-preserving e-process adjuster, and
define
\[
    (E^K)_n^*
    :=
    \max_{0\le t\le n}E_t^K,
    \qquad
    \widetilde E_n^K
    :=
    a\bigl((E^K)_n^*\bigr).
\]
By \Cref{lem:adjusted-max}, $\widetilde E^K$ is a nondecreasing
e-process valid under every $R\in K$.

Moreover, \[ \Phi_K(Q)>0. \] Indeed, if \(\Phi_K(Q)=0\), there would exist \(R_m\in K\) such that \(\KL(Q\|R_m)\to0\). Pinsker's inequality would then imply \(R_m\to Q\) in total variation, and hence weakly. By weak compactness of \(K\), some subsequence converges weakly to an element of \(K\); uniqueness of the weak limit would then imply \(Q\in K\), a contradiction.
Since $a$ is growth preserving,
\Cref{lem:growth-preserved} gives
\[
    \liminf_{n\to\infty}
    \frac1n\log \widetilde E_n^K
    \ge
    \Phi_K(Q),
    \qquad
    Q^\infty\text{-a.s.}
\]

Neither $E^K$ nor the adjuster $a$ depends on $Q$. Therefore,
the same nondecreasing e-process $\widetilde E^K$ has the stated
REGROW guarantee for every $Q\notin K$.
\end{proof}

\begin{proof}[Proof of \Cref{lem:compact-witness-basis}]
Let $\Delta(\mathcal{X})$ denote the space of all probability
measures on $\mathcal{X}$ equipped with the topology of weak
convergence. Because $\mathcal{X}$ is a Polish space,
$\Delta(\mathcal{X})$ is metrizable, for example by the
L\'evy--Prokhorov metric. Consequently, $\mathcal{P}$, being a
weakly compact subset of a metrizable space, is a compact metric
space under the relative weak topology.

Because compact metric spaces are second countable, the relative
weak topology on $\mathcal{P}$ admits a countable basis, which we
denote by $\mathcal{B}$. Furthermore, compact metric spaces are
regular. Therefore, for any $P\in\mathcal{P}$ and any relatively
weakly open neighborhood $G$ of $P$, there exists a relatively
weakly open set $V$ such that
\[
    P\in V\subseteq\overline V\subseteq G,
\]
where $\overline V$ denotes the closure in the relative weak
topology of $\mathcal{P}$. Because $\mathcal{B}$ is a
topological basis, there exists a basis element
$B\in\mathcal{B}$ such that
\[
    P\in B\subseteq V.
\]
It follows that
\[
    P\in B\subseteq\overline B
    \subseteq\overline V\subseteq G.
\]
Thus, $\mathcal{B}$ satisfies Condition~\textup{(i)} of
Definition~\ref{def:witness-basis}.

To verify Condition~\textup{(ii)}, fix $B\in\mathcal{B}$. Its
relative closure $\overline B$ is a relatively closed subset of
the weakly compact space $\mathcal{P}$. Hence $\overline B$ is
weakly compact as a subset of $\Delta(\mathcal X)$. By
Lemma~\ref{lem:weakly-compact-regrow}, there exists a single
nondecreasing e-process, valid under every
$R\in\overline B$, whose REGROW guarantee holds simultaneously
for every $S\notin\overline B$. Therefore, $\overline B$ is
simultaneously REGROW-regular over
$\mathcal{P}\setminus\overline B$, and
Condition~\textup{(ii)} holds.

Finally, we verify Condition~\textup{(iii)}. For any
$B\in\mathcal{B}$, define the exterior class
\[
    B^c:=\mathcal{P}\setminus B.
\]
Because $B$ is relatively weakly open in $\mathcal{P}$, its
complement $B^c$ is relatively weakly closed in $\mathcal{P}$.
Since $\mathcal{P}$ is weakly compact, $B^c$ is also weakly
compact. Applying Lemma~\ref{lem:weakly-compact-regrow} to
$B^c$, there exists a single nondecreasing e-process, valid
under every $R\in B^c$, whose REGROW guarantee holds
simultaneously for every $S\notin B^c$. Since
\[
    \mathcal{P}\setminus B^c=B,
\]
this guarantee holds simultaneously for every $S\in B$.
Therefore, $B^c$ is simultaneously REGROW-regular over $B$, and
Condition~\textup{(iii)} holds.

Having constructed a countable family $\mathcal{B}$ satisfying
all three required conditions, we conclude that $\mathcal{P}$
admits a countable local REGROW witness basis.
\end{proof}

\begin{proof}[Proof of \Cref{prop:gaussian-local-witness}]
We divide the proof into two parts.
\paragraph{Part 1: $\mathcal P$ is not weakly compact.}

By Prokhorov's theorem, a family of probability measures on a
Polish space is relatively weakly compact if and only if it is
uniformly tight.

Let $K\subseteq\mathbb R$ be compact. Then there exists $M>0$
such that
\[
    K\subseteq[-M,M].
\]
Consider the sequence
\[
    P_n=N(n,1)\in\mathcal P,
    \qquad n\ge1.
\]
Writing $\Phi$ for the standard Gaussian distribution function,
we have
\[
\begin{aligned}
    P_n(K)
    &\le
    P_n([-M,M])\\
    &=
    \Phi(M-n)-\Phi(-M-n)
    \longrightarrow0
\end{aligned}
\]
as $n\to\infty$. Consequently, for every compact
$K\subseteq\mathbb R$,
\[
    \inf_{P\in\mathcal P}P(K)=0.
\]
In particular, there is no compact $K$ such that
\[
    \inf_{P\in\mathcal P}P(K)\ge\frac12.
\]
Thus $\mathcal P$ is not uniformly tight. By Prokhorov's
theorem, $\mathcal P$ is not relatively weakly compact and hence
is not weakly compact.

\paragraph{Part 2: Construction of a countable local
REGROW witness basis.}

Consider the parameterization
\[
    f:\mathbb R\longrightarrow\mathcal P,
    \qquad
    f(\theta)=P_\theta.
\]
We first verify that $f$ is a homeomorphism from $\mathbb R$,
equipped with its usual topology, onto $\mathcal P$, equipped
with the relative weak topology.

If $\theta_n\to\theta$, then the characteristic functions satisfy
\[
    \exp\left\{
        i u\theta_n-\frac{u^2}{2}
    \right\}
    \longrightarrow
    \exp\left\{
        i u\theta-\frac{u^2}{2}
    \right\}
\]
for every $u\in\mathbb R$. Hence, by L\'evy's continuity theorem,
\[
    P_{\theta_n}\Rightarrow P_\theta.
\]
Thus $f$ is continuous.

Conversely, suppose
\[
    P_{\theta_n}\Rightarrow P_\theta.
\]
Because the distribution function of $P_\theta$ is continuous
at $\theta$, weak convergence implies
\[
\begin{aligned}
    \Phi(\theta-\theta_n)
    &=
    P_{\theta_n}((-\infty,\theta])\\
    &\longrightarrow
    P_\theta((-\infty,\theta])
    =
    \frac12.
\end{aligned}
\]
Since $\Phi$ is continuous and strictly increasing,
\[
    \theta-\theta_n\longrightarrow0,
\]
and therefore $\theta_n\to\theta$. Since the weak topology on
$\mathcal P$ is metrizable, this proves continuity of $f^{-1}$.
Hence $f$ is a homeomorphism.

Define
\[
    \mathscr B
    :=
    \left\{
        B_{q,r}:
        q,r\in\mathbb Q,\ q<r
    \right\},
\]
where
\[
    B_{q,r}
    :=
    \{P_\theta\in\mathcal P:q<\theta<r\}.
\]
The collection $\mathscr B$ is countable. We verify the three
conditions in \Cref{def:witness-basis}.

\medskip

\noindent
\emph{Condition \textup{(i)}: local refinement.}

Fix $P_{\theta_0}\in\mathcal P$, and let $G$ be a relatively
weakly open neighborhood of $P_{\theta_0}$. Since $f$ is a
homeomorphism,
\[
    f^{-1}(G)
\]
is an open neighborhood of $\theta_0$ in $\mathbb R$. Hence
there exists $\epsilon>0$ such that
\[
    [\theta_0-\epsilon,\theta_0+\epsilon]
    \subseteq
    f^{-1}(G).
\]
By density of the rational numbers, we can choose
$q,r\in\mathbb Q$ such that
\[
    \theta_0-\epsilon<q<\theta_0<r<\theta_0+\epsilon.
\]
Then
\[
    P_{\theta_0}\in B_{q,r}.
\]
Moreover, because $f$ is a homeomorphism,
\[
    \overline{B_{q,r}}^{\,\mathcal P}
    =
    \{P_\theta:q\le\theta\le r\}
    =
    f([q,r]),
\]
where the closure is taken relative to $\mathcal P$. Therefore,
\[
    P_{\theta_0}
    \in
    B_{q,r}
    \subseteq
    \overline{B_{q,r}}^{\,\mathcal P}
    \subseteq
    G.
\]
Thus Condition~\textup{(i)} holds.

\medskip

\noindent
\emph{Condition \textup{(ii)}: simultaneous inner witnesses.}

Fix $q,r\in\mathbb Q$ with $q<r$. The relative weak closure of
$B_{q,r}$ is
\[
    K_{q,r}^{\mathrm{in}}
    :=
    \overline{B_{q,r}}^{\,\mathcal P}
    =
    \{P_\theta:q\le\theta\le r\}
    =
    f([q,r]).
\]
Since $[q,r]$ is compact and $f$ is continuous,
$K_{q,r}^{\mathrm{in}}$ is weakly compact as a subset of
$\mathcal M_1(\mathbb R)$.

By \Cref{lem:weakly-compact-regrow}, there exists a single
e-process
\[
    V^{q,r,\mathrm{in}}
    =
    \{V_t^{q,r,\mathrm{in}}\}_{t\ge0}
\]
which is valid under every $R\in K_{q,r}^{\mathrm{in}}$ and
satisfies, for every
$Q\notin K_{q,r}^{\mathrm{in}}$,
\[
    \liminf_{t\to\infty}
    \frac1t
    \log V_t^{q,r,\mathrm{in}}
    \ge
    \Phi_{K_{q,r}^{\mathrm{in}}}(Q),
    \qquad
    Q^\infty\text{-a.s.}
\]
Importantly, the same process $V^{q,r,\mathrm{in}}$ works for
every $Q\notin K_{q,r}^{\mathrm{in}}$.

If $V^{q,r,\mathrm{in}}$ is not already nondecreasing, let $a$
be a fixed growth-preserving e-process adjuster and define
\[
    \bigl(V^{q,r,\mathrm{in}}\bigr)_t^*
    :=
    \max_{0\le u\le t}
    V_u^{q,r,\mathrm{in}},
\]
and
\[
    E_t^{q,r,\mathrm{in}}
    :=
    a\left(
        \bigl(V^{q,r,\mathrm{in}}\bigr)_t^*
    \right).
\]
By \Cref{lem:adjusted-max},
$E^{q,r,\mathrm{in}}$ is a nondecreasing e-process valid under
every $R\in K_{q,r}^{\mathrm{in}}$.

We briefly verify preservation of the simultaneous growth rate.
Fix
\[
    Q\notin K_{q,r}^{\mathrm{in}}.
\]
Since $K_{q,r}^{\mathrm{in}}$ is weakly compact and
$Q\notin K_{q,r}^{\mathrm{in}}$,
\[
    \Phi_{K_{q,r}^{\mathrm{in}}}(Q)>0.
\]
Furthermore,
\[
    \bigl(V^{q,r,\mathrm{in}}\bigr)_t^*
    \ge
    V_t^{q,r,\mathrm{in}},
\]
so
\[
    \liminf_{t\to\infty}
    \frac1t
    \log
    \bigl(V^{q,r,\mathrm{in}}\bigr)_t^*
    \ge
    \Phi_{K_{q,r}^{\mathrm{in}}}(Q)
    >0,
    \qquad
    Q^\infty\text{-a.s.}
\]
In particular,
\[
    \log
    \bigl(V^{q,r,\mathrm{in}}\bigr)_t^*
    \longrightarrow\infty
\]
almost surely under $Q^\infty$. Since $a$ is growth preserving,
\[
    \frac{\log a(e^y)}{y}\longrightarrow1
    \qquad\text{as }y\to\infty.
\]
It follows that
\[
    \liminf_{t\to\infty}
    \frac1t
    \log E_t^{q,r,\mathrm{in}}
    \ge
    \Phi_{K_{q,r}^{\mathrm{in}}}(Q),
    \qquad
    Q^\infty\text{-a.s.}
\]
The process $E^{q,r,\mathrm{in}}$ is independent of $Q$.
Therefore $K_{q,r}^{\mathrm{in}}$ is simultaneously
REGROW-regular over
\[
    \mathcal P\setminus K_{q,r}^{\mathrm{in}}.
\]
This proves Condition~\textup{(ii)}.

\medskip

\noindent
\emph{Condition \textup{(iii)}: simultaneous exterior
witnesses.}

Fix $q,r\in\mathbb Q$ with $q<r$, and write
\[
    K_{q,r}^{\mathrm{out}}
    :=
    B_{q,r}^c
    =
    \{P_\theta:\theta\le q\}
    \cup
    \{P_\theta:\theta\ge r\},
\]
where the complement is taken relative to $\mathcal P$.

We construct a single e-process, depending only on $(q,r)$,
which is valid under every distribution in
$K_{q,r}^{\mathrm{out}}$ and achieves the REGROW rate under
every $P_{\theta_0}\in B_{q,r}$.

Enumerate the countable dense subset
\[
    \mathbb Q\cap(q,r)
    =
    \{h_m:m\ge1\},
\]
and set
\[
    \rho_m:=2^{-m},
    \qquad m\ge1.
\]
Then
\[
    \rho_m>0,
    \qquad
    \sum_{m=1}^\infty\rho_m=1.
\]

For each $h\in(q,r)$, define
\begin{align}
    L_t^{(q,h)}
    &:=
    \prod_{i=1}^t
    \frac{p_h(X_i)}{p_q(X_i)}
    =
    \exp\left\{
        (h-q)\sum_{i=1}^tX_i
        -
        \frac t2(h^2-q^2)
    \right\},
    \label{eq:gaussian-left-lr}\\
    L_t^{(r,h)}
    &:=
    \prod_{i=1}^t
    \frac{p_h(X_i)}{p_r(X_i)}
    =
    \exp\left\{
        (h-r)\sum_{i=1}^tX_i
        -
        \frac t2(h^2-r^2)
    \right\},
    \label{eq:gaussian-right-lr}
\end{align}
with
\[
    L_0^{(q,h)}=L_0^{(r,h)}=1.
\]
Define
\[
    E_t^{(h)}
    :=
    \min\left\{
        L_t^{(q,h)},
        L_t^{(r,h)}
    \right\},
    \qquad
    E_0^{(h)}=1.
\]

We first verify that $E^{(h)}$ is an e-process for the
composite null $K_{q,r}^{\mathrm{out}}$.

If $X\sim P_\theta$, then
\begin{align*}
    \mathbb E_\theta
    \left[
        \frac{p_h(X)}{p_q(X)}
    \right]
    &=
    \exp\left\{
        (h-q)\theta
        +
        \frac12(h-q)^2
        -
        \frac12(h^2-q^2)
    \right\}\\
    &=
    \exp\{(h-q)(\theta-q)\}.
\end{align*}
Because $h>q$, this quantity is at most one whenever
$\theta\le q$. Hence
\[
    \mathbb E_\theta
    \left[
        L_t^{(q,h)}
        \mid
        \mathcal F_{t-1}
    \right]
    \le
    L_{t-1}^{(q,h)}
\]
for every $\theta\le q$. Thus $L^{(q,h)}$ is a nonnegative test
supermartingale under every $P_\theta$ with $\theta\le q$.

Similarly,
\[
    \mathbb E_\theta
    \left[
        \frac{p_h(X)}{p_r(X)}
    \right]
    =
    \exp\{(h-r)(\theta-r)\}.
\]
Because $h-r<0$, this quantity is at most one whenever
$\theta\ge r$. Therefore $L^{(r,h)}$ is a nonnegative test
supermartingale under every $P_\theta$ with $\theta\ge r$.

Let $\tau$ be an arbitrary stopping time. By optional stopping
for nonnegative supermartingales, with the usual Fatou argument
when $\tau$ is unbounded, if $\theta\le q$, then
\[
\begin{aligned}
    \mathbb E_\theta
    \left[
        E_\tau^{(h)}
    \right]
    &\le
    \mathbb E_\theta
    \left[
        L_\tau^{(q,h)}
    \right]\\
    &\le1.
\end{aligned}
\]
If $\theta\ge r$, then
\[
\begin{aligned}
    \mathbb E_\theta
    \left[
        E_\tau^{(h)}
    \right]
    &\le
    \mathbb E_\theta
    \left[
        L_\tau^{(r,h)}
    \right]\\
    &\le1.
\end{aligned}
\]
Thus, for every $h\in(q,r)$,
\[
    \sup_{P_\theta\in K_{q,r}^{\mathrm{out}}}
    \mathbb E_\theta
    \left[
        E_\tau^{(h)}
    \right]
    \le1
\]
for every stopping time $\tau$. Hence $E^{(h)}$ is an e-process
for the composite null $K_{q,r}^{\mathrm{out}}$.

Now define the fixed mixture
\begin{equation}
\label{eq:gaussian-universal-mixture}
    V_t^{q,r,\mathrm{out}}
    :=
    \sum_{m=1}^\infty
    \rho_m E_t^{(h_m)}.
\end{equation}
This process depends only on the interval $(q,r)$ and not on the
true interior parameter.

For every fixed $t$ and every realized
$(X_1,\ldots,X_t)$, the functions
\[
    h\longmapsto L_t^{(q,h)}
    \quad\text{and}\quad
    h\longmapsto L_t^{(r,h)}
\]
are continuous and bounded on the compact interval $[q,r]$.
Consequently, the series in
\eqref{eq:gaussian-universal-mixture} is finite.

Moreover, for every stopping time $\tau$ and every
$P_\theta\in K_{q,r}^{\mathrm{out}}$, Tonelli's theorem gives
\begin{align*}
    \mathbb E_\theta
    \left[
        V_\tau^{q,r,\mathrm{out}}
    \right]
    &=
    \sum_{m=1}^\infty
    \rho_m
    \mathbb E_\theta
    \left[
        E_\tau^{(h_m)}
    \right]\\
    &\le
    \sum_{m=1}^\infty\rho_m\\
    &=1.
\end{align*}
Thus $V^{q,r,\mathrm{out}}$ is an e-process under every
distribution in $K_{q,r}^{\mathrm{out}}$.

We next prove that this single mixture achieves the desired
growth rate under every $P_{\theta_0}\in B_{q,r}$. Fix an
arbitrary
\[
    \theta_0\in(q,r).
\]
For every fixed $h\in(q,r)$, the strong law of large numbers
gives
\[
    \frac1t\sum_{i=1}^tX_i
    \longrightarrow
    \theta_0,
    \qquad
    P_{\theta_0}^\infty\text{-a.s.}
\]
It follows from \eqref{eq:gaussian-left-lr} that
\begin{align*}
    \frac1t\log L_t^{(q,h)}
    &\longrightarrow
    (h-q)\theta_0-\frac12(h^2-q^2)\\
    &=
    \frac12(\theta_0-q)^2
    -
    \frac12(\theta_0-h)^2,
\end{align*}
and from \eqref{eq:gaussian-right-lr} that
\begin{align*}
    \frac1t\log L_t^{(r,h)}
    &\longrightarrow
    (h-r)\theta_0-\frac12(h^2-r^2)\\
    &=
    \frac12(\theta_0-r)^2
    -
    \frac12(\theta_0-h)^2,
\end{align*}
almost surely under $P_{\theta_0}^\infty$. Therefore,
\begin{equation}
\label{eq:gaussian-component-rate}
    \frac1t\log E_t^{(h)}
    \longrightarrow
    I_{q,r}(\theta_0)
    -
    \frac12(\theta_0-h)^2,
    \qquad
    P_{\theta_0}^\infty\text{-a.s.},
\end{equation}
where
\begin{align}
    I_{q,r}(\theta_0)
    &:=
    \min\left\{
        \frac12(\theta_0-q)^2,
        \frac12(\theta_0-r)^2
    \right\}
    \nonumber\\
    &=
    \inf_{P_\theta\in K_{q,r}^{\mathrm{out}}}
    D_{\mathrm{KL}}(P_{\theta_0}\|P_\theta)
    \nonumber\\
    &=
    \Phi_{K_{q,r}^{\mathrm{out}}}(P_{\theta_0}).
    \label{eq:gaussian-exterior-kl}
\end{align}
Since the set $\{h_m:m\ge1\}$ is countable,
\eqref{eq:gaussian-component-rate} holds simultaneously for all
$m$ on an event of $P_{\theta_0}^\infty$-probability one.

For each fixed $m$,
\[
    V_t^{q,r,\mathrm{out}}
    \ge
    \rho_m E_t^{(h_m)}.
\]
Hence
\begin{align*}
    \liminf_{t\to\infty}
    \frac1t
    \log V_t^{q,r,\mathrm{out}}
    &\ge
    \lim_{t\to\infty}
    \left\{
        \frac{\log\rho_m}{t}
        +
        \frac1t\log E_t^{(h_m)}
    \right\}\\
    &=
    I_{q,r}(\theta_0)
    -
    \frac12(\theta_0-h_m)^2
\end{align*}
almost surely under $P_{\theta_0}^\infty$.

Because this inequality holds for every $m$,
\[
\begin{aligned}
    \liminf_{t\to\infty}
    \frac1t
    \log V_t^{q,r,\mathrm{out}}
    &\ge
    \sup_{m\ge1}
    \left\{
        I_{q,r}(\theta_0)
        -
        \frac12(\theta_0-h_m)^2
    \right\}\\
    &=
    I_{q,r}(\theta_0),
\end{aligned}
\]
where the last equality follows from density of
$\{h_m:m\ge1\}$ in $(q,r)$. Together with
\eqref{eq:gaussian-exterior-kl}, this gives
\[
    \liminf_{t\to\infty}
    \frac1t
    \log V_t^{q,r,\mathrm{out}}
    \ge
    \Phi_{K_{q,r}^{\mathrm{out}}}(P_{\theta_0}),
    \qquad
    P_{\theta_0}^\infty\text{-a.s.}
\]

The process $V^{q,r,\mathrm{out}}$ is fixed independently of
$\theta_0$, and the preceding conclusion holds for every
$\theta_0\in(q,r)$. Thus it is a simultaneous exterior witness,
except that it need not be nondecreasing.

To obtain a nondecreasing process, define
\[
    \bigl(V^{q,r,\mathrm{out}}\bigr)_t^*
    :=
    \max_{0\le u\le t}
    V_u^{q,r,\mathrm{out}},
\]
and
\[
    E_t^{q,r,\mathrm{out}}
    :=
    a\left(
        \bigl(V^{q,r,\mathrm{out}}\bigr)_t^*
    \right),
\]
where $a$ is the same fixed growth-preserving adjuster as above.
By \Cref{lem:adjusted-max},
$E^{q,r,\mathrm{out}}$ is a nondecreasing e-process under every
$P_\theta\in K_{q,r}^{\mathrm{out}}$.

Finally, fix $\theta_0\in(q,r)$. Since
\[
    \Phi_{K_{q,r}^{\mathrm{out}}}(P_{\theta_0})
    =
    I_{q,r}(\theta_0)
    >0,
\]
we have
\[
    \liminf_{t\to\infty}
    \frac1t
    \log
    \bigl(V^{q,r,\mathrm{out}}\bigr)_t^*
    \ge
    I_{q,r}(\theta_0)>0.
\]
Therefore,
\[
    \log
    \bigl(V^{q,r,\mathrm{out}}\bigr)_t^*
    \longrightarrow\infty
\]
almost surely under $P_{\theta_0}^\infty$. The
growth-preserving property of $a$ then yields
\[
    \liminf_{t\to\infty}
    \frac1t
    \log E_t^{q,r,\mathrm{out}}
    \ge
    I_{q,r}(\theta_0)
    =
    \Phi_{K_{q,r}^{\mathrm{out}}}(P_{\theta_0}),
    \qquad
    P_{\theta_0}^\infty\text{-a.s.}
\]
Since the same process $E^{q,r,\mathrm{out}}$ works for every
$\theta_0\in(q,r)$,
$K_{q,r}^{\mathrm{out}}=B_{q,r}^c$ is simultaneously
REGROW-regular over $B_{q,r}$. This proves
Condition~\textup{(iii)}.

We have verified all three conditions of
\Cref{def:witness-basis}. Therefore,
\[
    \mathscr B
    =
    \left\{
        B_{q,r}:q,r\in\mathbb Q,\ q<r
    \right\}
\]
is a countable local REGROW witness basis for
$\mathcal P$.
\end{proof}

\begin{proof}[Proof of \Cref{thm:witness-implies-growth}]
We first construct the entire family of point-null e-processes
without fixing $P$, $Q$, $\eta_{\mathrm{in}}$, or
$\eta_{\mathrm{out}}$.

For every $j\ge1$, let
\[
    E^{j,\mathrm{in}}
    =
    \{E_n^{j,\mathrm{in}}\}_{n\ge0}
\] 
be a
nondecreasing e-process
which is valid under every $R\in\overline B_j$ and satisfies,
for every $S\notin\overline B_j$,
\[
    \liminf_{n\to\infty}
    \frac1n\log E_n^{j,\mathrm{in}}
    \ge
    \Phi_{\overline B_j}(S),
    \qquad
    S^\infty\text{-a.s.}
\]
Existence of such an e-process follows directly from
Condition~\textup{(ii)} of
\Cref{def:witness-basis} (because weak compactness of $B_j$ implies that it is simultaneously REGROW-regular). Similarly, it follows from
Condition~\textup{(iii)} of
\Cref{def:witness-basis} that we have
a nondecreasing e-process
\[
    E^{j,\mathrm{out}}
    =
    \{E_n^{j,\mathrm{out}}\}_{n\ge0}
\]
which is valid under every $R\in B_j^c$ and satisfies, for every
$S\in B_j$,
\[
    \liminf_{n\to\infty}
    \frac1n\log E_n^{j,\mathrm{out}}
    \ge
    \Phi_{B_j^c}(S),
    \qquad
    S^\infty\text{-a.s.}
\]

For every starting time $s$, let
\[
    E_{s:t}^{j,\mathrm{in}}
    \quad\text{and}\quad
    E_{s:t}^{j,\mathrm{out}},
    \qquad t\in\N,
\]
with $E_{s:t}^{j,\mathrm{in}}
   =
    E_{s:t}^{j,\mathrm{out}}=1, $ for $t< s$
denote the corresponding $s$ delay processes obtained by
applying the original processes to
$X_s,\ldots,X_t$ for $t\geq s$.
Under any i.i.d.\ no-change law, these are $s$-delay
e-processes for their respective null classes. We use the
inactive convention that every delayed process equals one for
$t<s$. Monotonicity is asserted only after activation, that is,
for $t\ge s$.

For each $j\ge1$, $R\in\calP$, and $t\ge s-1$, define
\begin{equation}
\label{eq:component-point-null}
    \overline M_{s:t}^{R,j}
    :=
    \begin{cases}
    1,
        & t<s,\\[0.5em]
    \displaystyle
    \frac{1+E_{s:t}^{j,\mathrm{in}}}{2},
        & t\ge s,\ R\in B_j,\\[1.1em]
    \displaystyle
    \frac{1+E_{s:t}^{j,\mathrm{out}}}{2},
        & t\ge s,\ R\notin B_j.
    \end{cases}
\end{equation}
Fix a sequence $\{\gamma_j\}_{j\in\N}$ such that $\gamma_j> 0$ and $\sum_j\gamma_j=1$.
Now define the universal point-null process
\begin{equation}
\label{eq:universal-point-null-mixture}
    \overline M_{s:t}^{R}
    :=
    \sum_{j=1}^\infty
    \gamma_j\overline M_{s:t}^{R,j}.
\end{equation}
This construction depends only on the fixed basis, its fixed
simultaneous witnesses, and the fixed weights
$\{\gamma_j\}_{j\ge1}$.

We first verify validity. Fix $R\in\calP$, $s\ge1$, and
$j\ge1$. If $R\in B_j$, then
\[
    R\in B_j\subseteq\overline B_j,
\]
and hence $E_{s:t}^{j,\mathrm{in}}$ is an $s$-delay e-process
under $R^\infty$. Therefore,
\[
    \overline M_{s:t}^{R,j}
    =
    \frac{1+E_{s:t}^{j,\mathrm{in}}}{2},
    \qquad t\ge s,
\]
together with the inactive value one before time $s$, is an
$s$-delay e-process under $R^\infty$.

If $R\notin B_j$, then $R\in B_j^c$, and the same argument,
using $E^{j,\mathrm{out}}$, shows that
$\overline M^{R,j}_{s:t}$ is an $s$-delay e-process under
$R^\infty$.

Since all summands in \eqref{eq:universal-point-null-mixture}
are nonnegative, conditional Tonelli's theorem gives, for every
stopping time $\tau$,
\[
\begin{aligned}
    \mathbb E_R^\infty
    \left[
        \overline M_{s:\tau}^{R}
        \,\middle|\,
        \mathcal F_{s-1}
    \right]
    &=
    \sum_{j=1}^\infty
    \gamma_j
    \mathbb E_R^\infty
    \left[
        \overline M_{s:\tau}^{R,j}
        \,\middle|\,
        \mathcal F_{s-1}
    \right]  \\
    &\le
    \sum_{j=1}^\infty\gamma_j
    =1.
\end{aligned}
\]
Thus
$\{\overline M_{s:t}^{R}\}_{t\ge s-1}$ is an $s$-delay
point-null e-process under $R^\infty$. Since all active component
processes are nondecreasing, so is
$t\mapsto\overline M_{s:t}^{R}$ for $t\ge s$.

We now fix $P,Q\in\calP$ with
\[
    I^*=\KL(Q\|P)\in(0,\infty)
\]
and fix $\eta_{\mathrm{in}}\in(0,I^*)$. Define
\[
    G_{\eta_{\mathrm{in}}}
    :=
    \left\{
        R\in\calP:
        \KL(Q\|R)>
        I^*-\frac{\eta_{\mathrm{in}}}{2}
    \right\}.
\]
By \Cref{lem:neighborhoods},
$G_{\eta_{\mathrm{in}}}$ is a relatively weakly open
neighborhood of $P$. By Condition~\textup{(i)} of
\Cref{def:witness-basis}, there exists an index $j_*$ such that
\[
    P\in B_{j_*}
    \subseteq
    \overline B_{j_*}
    \subseteq
    G_{\eta_{\mathrm{in}}}.
\]
Set
\[
    U_{\eta_{\mathrm{in}}}:=B_{j_*},
    \qquad
    K_{\eta_{\mathrm{in}},\mathrm{in}}
    :=
    \overline B_{j_*},
    \qquad
    K_{\eta_{\mathrm{in}},\mathrm{out}}
    :=
    B_{j_*}^c.
\]

Since
$\overline B_{j_*}\subseteq G_{\eta_{\mathrm{in}}}$,
\[
\begin{aligned}
    \Phi_{K_{\eta_{\mathrm{in}},\mathrm{in}}}(Q)
    &=
    \inf_{R\in\overline B_{j_*}}
    \KL(Q\|R)\\
    &\ge
    I^*-\frac{\eta_{\mathrm{in}}}{2}.
\end{aligned}
\]
Also, because $U_{\eta_{\mathrm{in}}}=B_{j_*}$ is a relatively
weakly open neighborhood of $P$,
\Cref{lem:neighborhoods} gives
\[
    c_{\eta_{\mathrm{in}}}
    :=
    \inf_{R\notin U_{\eta_{\mathrm{in}}}}
    \KL(P\|R)
    >0.
\]
Since
$K_{\eta_{\mathrm{in}},\mathrm{out}}
=B_{j_*}^c$, this is exactly
\[
    c_{\eta_{\mathrm{in}}}
    =
    \Phi_{K_{\eta_{\mathrm{in}},\mathrm{out}}}(P).
\]

Notice also that
\[
    Q\notin\overline B_{j_*}.
\]
Indeed,
\[
    \KL(Q\|Q)=0
    <
    I^*-\frac{\eta_{\mathrm{in}}}{2},
\]
so $Q\notin G_{\eta_{\mathrm{in}}}$, whereas
$\overline B_{j_*}\subseteq G_{\eta_{\mathrm{in}}}$.
Consequently, the simultaneous inner-witness guarantee applies
to $Q$, while $P\in B_{j_*}$ ensures that the simultaneous
outer-witness guarantee applies to $P$.

For brevity, define the fixed mixture penalty
\[
    C_{j_*}
    :=
    \log\frac{2}{\gamma_{j_*}}
    <\infty.
\]

We first prove the exterior growth statement. For every
$R\notin U_{\eta_{\mathrm{in}}}=B_{j_*}$,
\eqref{eq:universal-point-null-mixture} and
\eqref{eq:component-point-null} imply
\[
\begin{aligned}
    \overline M_{1:n}^{R}
    &\ge
    \gamma_{j_*}
    \overline M_{1:n}^{R,j_*}\\
    &=
    \gamma_{j_*}
    \frac{1+E_n^{j_*,\mathrm{out}}}{2}\\
    &\ge
    \frac{\gamma_{j_*}}{2}
    E_n^{j_*,\mathrm{out}}.
\end{aligned}
\]
The right-hand side is independent of $R$, and hence
\begin{equation}
\label{eq:universal-outer-lower}
    \inf_{R\notin U_{\eta_{\mathrm{in}}}}
    \log\overline M_{1:n}^{R}
    \ge
    \log E_n^{j_*,\mathrm{out}}
    -
    C_{j_*}.
\end{equation}

By simultaneous exterior REGROW regularity and
$P\in B_{j_*}$,
\[
    \liminf_{n\to\infty}
    \frac1n\log E_n^{j_*,\mathrm{out}}
    \ge
    \Phi_{B_{j_*}^c}(P)
    =
    c_{\eta_{\mathrm{in}}},
    \qquad
    P^\infty\text{-a.s.}
\]
Therefore, for every
$\eta_{\mathrm{out}}\in(0,c_{\eta_{\mathrm{in}}})$,
\[
    \mathbb P_P^\infty
    \left(
        \frac1n
        \log E_n^{j_*,\mathrm{out}}
        \ge
        c_{\eta_{\mathrm{in}}}
        -
        \frac{\eta_{\mathrm{out}}}{2}
    \right)
    \longrightarrow1.
\]
Since $C_{j_*}$ is fixed,
\[
    \frac{C_{j_*}}{n}
    \le
    \frac{\eta_{\mathrm{out}}}{2}
\]
for all sufficiently large $n$. Combining this with
\eqref{eq:universal-outer-lower} gives
\[
    \mathbb P_P^\infty
    \left(
        \inf_{R\notin U_{\eta_{\mathrm{in}}}}
        \log\overline M_{1:n}^{R}
        \ge
        n\bigl(
            c_{\eta_{\mathrm{in}}}
            -
            \eta_{\mathrm{out}}
        \bigr)
    \right)
    \longrightarrow1.
\]
This proves \eqref{eq:pre-growth}.

We next prove the inner growth statement. For every
$R\in U_{\eta_{\mathrm{in}}}=B_{j_*}$,
\[
\begin{aligned}
    \overline M_{1:n}^{R}
    &\ge
    \gamma_{j_*}
    \overline M_{1:n}^{R,j_*}\\
    &=
    \gamma_{j_*}
    \frac{1+E_n^{j_*,\mathrm{in}}}{2}\\
    &\ge
    \frac{\gamma_{j_*}}{2}
    E_n^{j_*,\mathrm{in}}.
\end{aligned}
\]
Thus
\begin{equation}
\label{eq:universal-inner-lower}
    \inf_{R\in U_{\eta_{\mathrm{in}}}}
    \log\overline M_{1:n}^{R}
    \ge
    \log E_n^{j_*,\mathrm{in}}
    -
    C_{j_*}.
\end{equation}

By simultaneous inner REGROW regularity and
$Q\notin\overline B_{j_*}$,
\[
\begin{aligned}
    \liminf_{n\to\infty}
    \frac1n\log E_n^{j_*,\mathrm{in}}
    &\ge
    \Phi_{\overline B_{j_*}}(Q)\\
    &\ge
    I^*-\frac{\eta_{\mathrm{in}}}{2},
    \qquad
    Q^\infty\text{-a.s.}
\end{aligned}
\]
It follows that
\[
    \mathbb P_Q^\infty
    \left(
        \frac1n\log E_n^{j_*,\mathrm{in}}
        \ge
        I^*-\frac{3\eta_{\mathrm{in}}}{4}
    \right)
    \longrightarrow1.
\]
Since $C_{j_*}$ is fixed,
\[
    \frac{C_{j_*}}{n}
    \le
    \frac{\eta_{\mathrm{in}}}{4}
\]
for all sufficiently large $n$. Combining this with
\eqref{eq:universal-inner-lower} gives
\[
    \mathbb P_Q^\infty
    \left(
        \inf_{R\in U_{\eta_{\mathrm{in}}}}
        \log\overline M_{1:n}^{R}
        \ge
        n\bigl(I^*-\eta_{\mathrm{in}}\bigr)
    \right)
    \longrightarrow1.
\]
This proves \eqref{eq:post-growth}.

The family
$\{\overline M_{s:t}^{R}\}$ was constructed before fixing
$(P,Q)$ or either slack parameter. Therefore, the same detector
works simultaneously for every pair $(P,Q)$ and for every fixed
choice of $\eta_{\mathrm{in}}$ and
$\eta_{\mathrm{out}}$; only the basis element used in the
analysis changes.
\end{proof}

\section{Auxiliary Lemmas}

\begin{lemma}[Balancing two quadratic terms]
\label{lem:two-quadratics}
For $a,b\in\R$ and $K_0,K_1>0$,
\[
  \inf_{\theta\in\R}
  \max\{K_0(\theta-a)^2,K_1(\theta-b)^2\}
  =
  \frac{(a-b)^2}{(K_0^{-1/2}+K_1^{-1/2})^2}.
\]
\end{lemma}

\begin{proof}
For $u>0$, the inequalities
\[
  K_0(\theta-a)^2\le u,
  \qquad
  K_1(\theta-b)^2\le u
\]
hold simultaneously if and only if the intervals
\[
  \left[a-\sqrt{u/K_0},a+\sqrt{u/K_0}\right]
  \quad\text{and}\quad
  \left[b-\sqrt{u/K_1},b+\sqrt{u/K_1}\right]
\]
intersect.  They first intersect when
\[
  |a-b|=\sqrt{u/K_0}+\sqrt{u/K_1}.
\]
Solving for $u$ gives the claim.
\end{proof}

Fix a changepoint $T$ and a delay $d\ge 1$.  At time $t=T+d$, define the post-change interval
\[
  J_0=(T+1):(T+d),
  \qquad |J_0|=d,
\]
and, for a pre-change block length $m\le T$, define the crossing interval
\[
  J_1=(T-m+1):(T+d),
  \qquad |J_1|=m+d.
\]
Let
\[
  \bar X_0=\bar X_{J_0},
  \qquad
  \bar X_1=\bar X_{J_1}.
\]
Define
\begin{equation}
\label{eq:Bmd-def}
  B_{m,d}
  =
  \frac{(\bar X_0-\bar X_1)^2}
       {(\kappa_d^{-1/2}+\kappa_{m+d}^{-1/2})^2}.
\end{equation}
and
\begin{equation}
\label{eq:kappa-c-def}
  c_n=(1+\rho^2\sigma^2 n)^{-1/2},
  \qquad
  \kappa_n=\frac{\rho^2 n^2}{2(1+\rho^2\sigma^2 n)}.
\end{equation}
Then, we have
\[
  M_{s:t}^{\theta}=c_{n_{s:t}}\exp\{\kappa_{n_{s:t}}(\bar X_{s:t}-\theta)^2\}.
\]
As $n\to\infty$,
\begin{equation}
\label{eq:kappa-asymp}
  \kappa_n=\frac{n}{2\sigma^2}\{1+o(1)\},
  \qquad
  \log c_n=-\frac12\log(1+\rho^2\sigma^2 n).
\end{equation}
\begin{lemma}[A sufficient condition for stopping]
\label{lem:sufficient-stopping}
At time $T+d$,
\begin{equation}
\label{eq:logD-lower-B}
  \log D_{T+d}
  \ge
  B_{m,d}
  -\frac12\log(1+\rho^2\sigma^2(m+d)).
\end{equation}
Consequently, if
\begin{equation}
\label{eq:finite-sample-detection-condition}
  B_{m,d}
  \ge
  L+\frac12\log(1+\rho^2\sigma^2(m+d)),
\end{equation}
then $\tau_A\le T+d$.
\end{lemma}

\begin{proof}
Since $D_{T+d}$ sums over all intervals ending at $T+d$, it is at least the sum of the two terms indexed by $J_0$ and $J_1$:
\[
  D_{T+d}
  \ge
  \inf_\theta\{M_{J_0}^\theta+M_{J_1}^\theta\}
  \ge
  \inf_\theta\max\{M_{J_0}^\theta,M_{J_1}^\theta\}.
\]
Because $c_d\ge c_{m+d}$,
\[
  M_{J_0}^\theta
  \ge
  c_{m+d}\exp\{\kappa_d(\theta-\bar X_0)^2\},
\]
and
\[
  M_{J_1}^\theta
  =
  c_{m+d}\exp\{\kappa_{m+d}(\theta-\bar X_1)^2\}.
\]
Therefore
\[
  \log D_{T+d}
  \ge
  \log c_{m+d}
  +
  \inf_\theta\max\{\kappa_d(\theta-\bar X_0)^2,
                    \kappa_{m+d}(\theta-\bar X_1)^2\}.
\]
Lemma~\ref{lem:two-quadratics} gives the second term as $B_{m,d}$, while
\[
  \log c_{m+d}
  =-\frac{1}{2}\log(1+\rho^2\sigma^2(m+d)).
\]
This proves \eqref{eq:logD-lower-B}.  If the right side is at least $L=\log A$, then $D_{T+d}\ge A$ and hence $\tau_A\le T+d$.
\end{proof}
 Define the two-start spending penalty
\begin{equation}
\label{eq:ell-T-m}
  \ell_\pi(T,m)
  :=
  \max\{-\log\pi_{T+1},-\log\pi_{T-m+1}\}.
\end{equation}
\begin{lemma}[Weighted sufficient condition for stopping]
\label{lem:weighted-sufficient-stopping}
At time $T+d$,
\begin{equation}
\label{eq:weighted-logD-lower}
  \log D_{T+d}^{\mathrm{PFA}}
  \ge
  B_{m,d}
  -\ell_\pi(T,m)
  -\frac12\log(1+\rho^2\sigma^2(m+d)).
\end{equation}
Consequently, if the right side is at least $L$, then $\tau_\alpha^{\mathrm{PFA}}\le T+d$.
\end{lemma}

\begin{proof}
The weighted detector contains the two terms corresponding to $J_0$ and $J_1$, hence
\[
  D_{T+d}^{\mathrm{PFA}}
  \ge
  \inf_\theta\{\pi_{T+1}M_{J_0}^\theta+
                 \pi_{T-m+1}M_{J_1}^\theta\}
  \ge
  \inf_\theta\max\{\pi_{T+1}M_{J_0}^\theta,
                   \pi_{T-m+1}M_{J_1}^\theta\}.
\]
By definition of $\ell_\pi(T,m)$ and because $c_d\ge c_{m+d}$,
\[
  \pi_{T+1}M_{J_0}^\theta
  \ge
  \exp\{-\ell_\pi(T,m)\}
  c_{m+d}\exp\{\kappa_d(\theta-\bar X_0)^2\},
\]
and
\[
  \pi_{T-m+1}M_{J_1}^\theta
  \ge
  \exp\{-\ell_\pi(T,m)\}
  c_{m+d}\exp\{\kappa_{m+d}(\theta-\bar X_1)^2\}.
\]
Taking logarithms and applying Lemma~\ref{lem:two-quadratics} yields
\[
  \log D_{T+d}^{\mathrm{PFA}}
  \ge
  -\ell_\pi(T,m)+\log c_{m+d}+B_{m,d}.
\]
Since $\log c_{m+d}=-\frac12\log(1+\rho^2\sigma^2(m+d))$, the result follows.
\end{proof}

\begin{lemma}[Uniform Robbins--Stirling lower bound]
\label{lem:uniform-robbins-stirling}
There exists a universal numerical constant \(c_J>0\) such that, for every
integer \(n\ge 1\) and every \(z\in[0,n]\),
\begin{equation}
\label{eq:uniform-robbins-stirling}
\frac{
    \Gamma(z+1/2)\Gamma(n-z+1/2)
}{
    \pi\Gamma(n+1)
}
\ge
\frac{c_J}{\sqrt{n+1}}
\left(\frac zn\right)^z
\left(1-\frac zn\right)^{n-z},
\end{equation}
where we use the convention \(0^0=1\).
\end{lemma}

\begin{proof}
Write
\[
    a:=z,
    \qquad
    b:=n-z,
\]
so that \(a,b\ge 0\) and \(a+b=n\).

We first record uniform Stirling bounds that will be used below. The standard
Stirling formula with remainder gives, for every \(y>0\),
\begin{equation}
\label{eq:stirling-gamma-bounds}
\sqrt{2\pi}\,
y^{y+1/2}e^{-y}
\le
\Gamma(y+1)
\le
\sqrt{2\pi}\,
y^{y+1/2}e^{-y+1/(12y)}.
\end{equation}
In particular, for every \(x\ge1\), applying the lower bound in
\eqref{eq:stirling-gamma-bounds} with \(y=x-\frac12\) gives
\begin{align}
\Gamma\left(x+\frac12\right)
&\ge
\sqrt{2\pi}
\left(x-\frac12\right)^x
e^{-(x-1/2)}
\nonumber\\
&=
\sqrt{2\pi}\,e^{1/2}
x^xe^{-x}
\left(1-\frac{1}{2x}\right)^x.
\label{eq:half-gamma-lower-pre}
\end{align}
By Bernoulli's inequality, since \(x\ge1\),
\[
\left(1-\frac{1}{2x}\right)^x
\ge
1-\frac{x}{2x}
=
\frac12.
\]
Consequently,
\begin{equation}
\label{eq:half-gamma-lower}
\Gamma\left(x+\frac12\right)
\ge
c_-x^xe^{-x},
\qquad x\ge1,
\end{equation}
where
\[
    c_-:=\frac{\sqrt{2\pi}\,e^{1/2}}{2}>0.
\]

Similarly, applying the upper bound in
\eqref{eq:stirling-gamma-bounds} with \(y=n\ge1\) yields
\begin{equation}
\label{eq:factorial-upper}
\Gamma(n+1)
\le
c_+\sqrt n\,n^ne^{-n},
\qquad
c_+:=\sqrt{2\pi}\,e^{1/12}.
\end{equation}

We now divide the proof into cases.

\medskip
\noindent
\textbf{Case 1: \(a\ge1\) and \(b\ge1\).}

Using \eqref{eq:half-gamma-lower} twice and
\eqref{eq:factorial-upper}, we obtain
\begin{align*}
\frac{
    \Gamma(a+1/2)\Gamma(b+1/2)
}{
    \pi\Gamma(n+1)
}
&\ge
\frac{
    c_-^2a^ab^be^{-(a+b)}
}{
    \pi c_+\sqrt n\,n^ne^{-n}
}.
\end{align*}
Since \(a+b=n\), the exponential factors cancel, and hence
\begin{align}
\frac{
    \Gamma(a+1/2)\Gamma(b+1/2)
}{
    \pi\Gamma(n+1)
}
&\ge
\frac{c_-^2}{\pi c_+\sqrt n}
\frac{a^ab^b}{n^n}
\nonumber\\
&=
\frac{c_-^2}{\pi c_+\sqrt n}
\left(\frac an\right)^a
\left(\frac bn\right)^b
\nonumber\\
&\ge
\frac{c_-^2}{\pi c_+\sqrt{n+1}}
\left(\frac an\right)^a
\left(\frac bn\right)^b.
\label{eq:interior-gamma-bound}
\end{align}

\medskip
\noindent
\textbf{Case 2: \(0\le a\le1\) and \(b\ge1\).}

Define
\[
    m_0
    :=
    \min_{0\le x\le1}
    \Gamma\left(x+\frac12\right).
\]
Since the Gamma function is continuous and strictly positive on the compact
interval \([1/2,3/2]\), we have \(m_0>0\). Therefore,
\[
    \Gamma\left(a+\frac12\right)\ge m_0.
\]
Using this inequality, \eqref{eq:half-gamma-lower} for \(b\), and
\eqref{eq:factorial-upper}, we obtain
\begin{align*}
\frac{
    \Gamma(a+1/2)\Gamma(b+1/2)
}{
    \pi\Gamma(n+1)
}
&\ge
\frac{
    m_0c_-b^be^{-b}
}{
    \pi c_+\sqrt n\,n^ne^{-n}
}\\
&=
\frac{m_0c_-}{\pi c_+\sqrt n}
\frac{e^ab^b}{n^n}
\\
&=
\frac{m_0c_-}{\pi c_+\sqrt n}
\frac{e^a}{n^a}
\left(\frac bn\right)^b.
\end{align*}
For \(a\in[0,1]\), with the convention \(0^0=1\), we have
\[
    e^a\ge a^a.
\]
Indeed, for \(a\in(0,1]\),
\[
    \log\left(\frac{e^a}{a^a}\right)
    =
    a-a\log a
    \ge0,
\]
and the assertion also holds at \(a=0\) by continuity and the convention
\(0^0=1\). It follows that
\begin{align}
\frac{
    \Gamma(a+1/2)\Gamma(b+1/2)
}{
    \pi\Gamma(n+1)
}
&\ge
\frac{m_0c_-}{\pi c_+\sqrt n}
\left(\frac an\right)^a
\left(\frac bn\right)^b
\nonumber\\
&\ge
\frac{m_0c_-}{\pi c_+\sqrt{n+1}}
\left(\frac an\right)^a
\left(\frac bn\right)^b.
\label{eq:left-boundary-gamma-bound}
\end{align}

The case \(a\ge1\) and \(0\le b\le1\) follows symmetrically.

\medskip
\noindent
\textbf{Case 3: \(n=1\).}

In this case \(a,b\in[0,1]\). Hence
\[
    \Gamma\left(a+\frac12\right)\ge m_0,
    \qquad
    \Gamma\left(b+\frac12\right)\ge m_0.
\]
Since \(\Gamma(2)=1\),
\[
\frac{
    \Gamma(a+1/2)\Gamma(b+1/2)
}{
    \pi\Gamma(2)
}
\ge
\frac{m_0^2}{\pi}.
\]
Moreover,
\[
    a^ab^b\le1,
\]
because \(a,b\in[0,1]\). Therefore,
\begin{equation}
\label{eq:n-one-gamma-bound}
\frac{
    \Gamma(a+1/2)\Gamma(b+1/2)
}{
    \pi\Gamma(2)
}
\ge
\frac{\sqrt2\,m_0^2}{\pi\sqrt{2}}
a^ab^b.
\end{equation}

Finally, define
\[
c_J
:=
\min\left\{
    \frac{c_-^2}{\pi c_+},
    \frac{m_0c_-}{\pi c_+},
    \frac{\sqrt2\,m_0^2}{\pi}
\right\}.
\]
This is a strictly positive universal constant. Combining
\eqref{eq:interior-gamma-bound},
\eqref{eq:left-boundary-gamma-bound}, its symmetric counterpart, and
\eqref{eq:n-one-gamma-bound}, we conclude that
\[
\frac{
    \Gamma(a+1/2)\Gamma(b+1/2)
}{
    \pi\Gamma(n+1)
}
\ge
\frac{c_J}{\sqrt{n+1}}
\left(\frac an\right)^a
\left(\frac bn\right)^b.
\]
Substituting \(a=z\) and \(b=n-z\) proves
\eqref{eq:uniform-robbins-stirling}.
\end{proof}

\begin{lemma}[Positive lower growth is preserved]
\label{lem:growth-preserved}
Let \(a\) be a growth-preserving adjuster, and let \(Y_n\ge1\) be random
variables. If, for some \(r>0\),
\[
    \liminf_{n\to\infty}\frac{1}{n}\log Y_n
    \ge r
    \qquad\text{almost surely},
\]
then 
\[
    \liminf_{n\to\infty}\frac{1}{n}\log a(Y_n)
    \ge r
    \qquad\text{almost surely}.
\]
\end{lemma}
\begin{proof}
Because \(a\) is growth preserving,
\[
    \frac{\log a(e^y)}{y}\longrightarrow1
    \qquad (y\to\infty).
\]
Fix \(\varepsilon\in(0,1)\). Then there exists \(y_\varepsilon<\infty\)
such that
\[
    \log a(e^y)\ge(1-\varepsilon)y,
    \qquad y\ge y_\varepsilon.
\]

Work on the probability-one event on which
\[
    \liminf_{n\to\infty}\frac{1}{n}\log Y_n\ge r.
\]
For all sufficiently large \(n\),
\[
    \frac{1}{n}\log Y_n\ge r-\varepsilon>0,
\]
where \(\varepsilon<r\) may be assumed. Consequently,
\[
    \log Y_n\ge n(r-\varepsilon)\longrightarrow\infty,
\]
so \(Y_n\ge1\) and \(\log Y_n\ge y_\varepsilon\) eventually. Therefore,
for all sufficiently large \(n\),
\[
    \frac{1}{n}\log a(Y_n)
    \ge
    (1-\varepsilon)\frac{1}{n}\log Y_n.
\]
Taking lower limits gives
\[
    \liminf_{n\to\infty}\frac{1}{n}\log a(Y_n)
    \ge
    (1-\varepsilon)
    \liminf_{n\to\infty}\frac{1}{n}\log Y_n
    \ge
    (1-\varepsilon)r.
\]
Letting \(\varepsilon\downarrow0\) proves the result.
\end{proof}

\begin{lemma}[Stability under a vanishing contaminated suffix]
\label{lem:UI-cross-stability}
Let $T_n\to\infty$, $d_n\to\infty$, and $d_n/T_n\to0$.  Suppose the first
$T_n$ observations are i.i.d. $N(a,u)$ and the remaining $d_n$ observations
are i.i.d. from a fixed Gaussian law, independently.  Put $n=T_n+d_n$ and
use \eqref{eq:gaussian-regularized-predictors}.  Then, for every $\delta>0$,
\begin{equation}
\label{eq:UI-cross-outside}
  \inf_{|\theta-a|\ge\delta}
  \frac1n\log R_{1:n}^{\theta}
  \longrightarrow
  \frac12\log\left(1+\frac{\delta^2}{u}\right)
\end{equation}
in probability.  Moreover, for every compact $K\subset\mathbb R$,
\begin{equation}
\label{eq:UI-cross-compact}
  \sup_{\theta\in K}
  \left|
    \frac1n\log R_{1:n}^{\theta}-J_{a,u}(\theta)
  \right|
  \longrightarrow0
\end{equation}
in probability.
\end{lemma}
\begin{proof}[Proof of \Cref{lem:UI-cross-stability}]
Write
\[
N_n:=T_n+d_n,
\]
so that \(N_n\) is denoted by \(n\) in the statement.  Let the
contaminating Gaussian law be \(N(b,v)\), where \(v>0\).

Since only convergence in probability is claimed, we may realize all
rows of the triangular array on a common probability space as follows.
Let
\[
Z_1,Z_2,\ldots \stackrel{\mathrm{iid}}{\sim} N(a,u),
\qquad
Y_1,Y_2,\ldots \stackrel{\mathrm{iid}}{\sim} N(b,v),
\]
with the two sequences independent, and, in the \(n\)-th row, set
\[
X_i=
\begin{cases}
Z_i, & 1\le i\le T_n,\\
Y_{i-T_n}, & T_n<i\le N_n.
\end{cases}
\]
Each row then has exactly the distribution specified in the lemma.
We prove the desired convergence almost surely under this coupling,
which implies the stated convergence in probability.

For \(k\ge1\), let
\[
\bar X_k:=\frac1k\sum_{i=1}^k X_i,
\qquad
\widehat v_k
:=
\frac1k\sum_{i=1}^k(X_i-\bar X_k)^2.
\]

\medskip
\noindent
\textit{Step 1: Uniform stability of the empirical moments and
predictors over the suffix.}

We first show that
\begin{equation}
\label{eq:uniform-cross-moments}
\sup_{0\le j\le d_n}
\left|\bar X_{T_n+j}-a\right|
\longrightarrow0,
\qquad
\sup_{0\le j\le d_n}
\left|\widehat v_{T_n+j}-u\right|
\longrightarrow0
\end{equation}
almost surely.

Indeed, for \(0\le j\le d_n\),
\[
\bar X_{T_n+j}-a
=
\frac{
\sum_{i=1}^{T_n}(Z_i-a)
+
\sum_{\ell=1}^{j}(Y_\ell-a)
}{
T_n+j
}.
\]
Consequently,
\begin{align*}
\sup_{0\le j\le d_n}
|\bar X_{T_n+j}-a|
&\le
\left|
\frac1{T_n}\sum_{i=1}^{T_n}(Z_i-a)
\right|
+
\frac1{T_n}
\max_{0\le j\le d_n}
\left|
\sum_{\ell=1}^{j}(Y_\ell-a)
\right|.
\end{align*}
The first term converges to zero almost surely by the strong law.  For
the second term, write
\[
Y_\ell-a=(Y_\ell-b)+(b-a).
\]
Thus
\begin{align*}
\frac1{T_n}
\max_{0\le j\le d_n}
\left|
\sum_{\ell=1}^{j}(Y_\ell-a)
\right|
&\le
\frac{d_n}{T_n}|b-a|
+
\frac{d_n}{T_n}
\frac1{d_n}
\max_{0\le j\le d_n}
\left|
\sum_{\ell=1}^{j}(Y_\ell-b)
\right|.
\end{align*}
Since
\[
\frac1m
\max_{0\le j\le m}
\left|
\sum_{\ell=1}^{j}(Y_\ell-b)
\right|
\longrightarrow0
\qquad\text{a.s.},
\]
and \(d_n/T_n\to0\), the second term also converges to zero almost
surely.  This proves the first assertion in
\eqref{eq:uniform-cross-moments}.

The same argument applied to the second moments gives
\[
\sup_{0\le j\le d_n}
\left|
\frac1{T_n+j}
\sum_{i=1}^{T_n+j}X_i^2-(u+a^2)
\right|
\longrightarrow0
\qquad\text{a.s.}
\]
Indeed, decompose
\[
Y_\ell^2-(u+a^2)
=
\bigl\{Y_\ell^2-\E[Y_\ell^2]\bigr\}
+
\bigl\{\E[Y_\ell^2]-(u+a^2)\bigr\},
\]
and use the same maximal strong-law argument, noting that
\(\E[Y_1^2]<\infty\).  Since
\[
\widehat v_k
=
\frac1k\sum_{i=1}^kX_i^2-\bar X_k^2,
\]
the second assertion in \eqref{eq:uniform-cross-moments} follows.

For \(k\ge1\), abbreviate the regularized predictors after observing
\(X_1,\ldots,X_k\) by
\[
\widetilde\mu_k:=\bar X_k,
\qquad
\widetilde\sigma_k^2
:=
\frac{
v_0+\sum_{i=1}^k(X_i-\bar X_k)^2
}{
k+\nu_0
}
=
\frac{v_0+k\widehat v_k}{k+\nu_0}.
\]
It follows from \eqref{eq:uniform-cross-moments} that
\begin{equation}
\label{eq:uniform-predictor-stability}
\sup_{T_n\le k\le N_n}
|\widetilde\mu_k-a|
\longrightarrow0,
\qquad
\sup_{T_n\le k\le N_n}
|\widetilde\sigma_k^2-u|
\longrightarrow0
\end{equation}
almost surely.  For the variance predictor, this follows from
\[
\widetilde\sigma_k^2-u
=
\frac{k}{k+\nu_0}(\widehat v_k-u)
+
\frac{v_0-u\nu_0}{k+\nu_0}.
\]

\medskip
\noindent
\textit{Step 2: The suffix contributes \(o(N_n)\) to the predictive
score.}

Define
\[
\ell_i
:=
\log\widetilde\sigma_{i-1}
+
\frac{(X_i-\widetilde\mu_{i-1})^2}
     {2\widetilde\sigma_{i-1}^2},
\]
with the prescribed initial predictors used when \(i=1\).  By
\eqref{eq:uniform-predictor-stability}, almost surely, for all
sufficiently large \(n\), there are deterministic constants
\(0<c<C<\infty\) such that
\[
|\widetilde\mu_{T_n+j-1}|\le C,
\qquad
c\le\widetilde\sigma_{T_n+j-1}^2\le C,
\qquad
1\le j\le d_n.
\]
Hence, for another finite constant \(C_1\),
\[
|\ell_{T_n+j}|
\le
C_1(1+Y_j^2),
\qquad 1\le j\le d_n.
\]
Therefore,
\begin{align}
\frac1{N_n}
\left|
\sum_{j=1}^{d_n}\ell_{T_n+j}
\right|
&\le
C_1\frac{d_n}{N_n}
\left(
1+\frac1{d_n}\sum_{j=1}^{d_n}Y_j^2
\right)
\longrightarrow0
\label{eq:suffix-score-negligible}
\end{align}
almost surely, because \(d_n/N_n\to0\) and
\[
\frac1{d_n}\sum_{j=1}^{d_n}Y_j^2
\longrightarrow
\E[Y_1^2]
\qquad\text{a.s.}
\]

\medskip
\noindent
\textit{Step 3: Identification of the \(\theta\)-independent term.}

For any sample size \(m\), write
\[
C_m
:=
\frac12-\frac1m\sum_{i=1}^m\ell_i.
\]
Then the definition of \(R_{1:m}^{\theta}\) gives the exact
decomposition
\begin{equation}
\label{eq:UI-cross-exact-decomp}
\frac1m\log R_{1:m}^{\theta}
=
C_m
+
\frac12
\log\left\{
\widehat v_m+(\bar X_m-\theta)^2
\right\}.
\end{equation}

Let \(C_{T_n}^{P}\) denote the corresponding quantity computed from
the uncontaminated observations \(Z_1,\ldots,Z_{T_n}\).  Proposition
\ref{prop:UI-t-eprocess}, applied under \(N(a,u)\) with \(\theta=a\),
gives
\[
\frac1{T_n}\log R_{1:T_n}^{a}
\longrightarrow
J_{a,u}(a)=0
\qquad\text{a.s.}
\]
On the other hand,
\[
\frac1{T_n}\log R_{1:T_n}^{a}
=
C_{T_n}^{P}
+
\frac12
\log\left\{
\widehat v_{T_n}
+
(\bar Z_{T_n}-a)^2
\right\}.
\]
Since
\[
\widehat v_{T_n}\longrightarrow u,
\qquad
\bar Z_{T_n}\longrightarrow a
\qquad\text{a.s.},
\]
we obtain
\begin{equation}
\label{eq:pure-prefix-C-limit}
C_{T_n}^{P}
\longrightarrow
-\frac12\log u
\qquad\text{a.s.}
\end{equation}

The prefix predictors in the contaminated row coincide exactly with
those computed from \(Z_1,\ldots,Z_{T_n}\).  Hence
\[
C_{N_n}
=
\frac{T_n}{N_n}C_{T_n}^{P}
+
\frac{d_n}{2N_n}
-
\frac1{N_n}
\sum_{j=1}^{d_n}\ell_{T_n+j}.
\]
Using \(T_n/N_n\to1\), \(d_n/N_n\to0\),
\eqref{eq:suffix-score-negligible}, and
\eqref{eq:pure-prefix-C-limit}, we conclude that
\begin{equation}
\label{eq:Cn-cross-limit}
C_{N_n}
\longrightarrow
-\frac12\log u
\qquad\text{a.s.}
\end{equation}

Also, by \eqref{eq:uniform-cross-moments} with \(j=d_n\),
\begin{equation}
\label{eq:full-cross-moment-limit}
\bar X_{N_n}\longrightarrow a,
\qquad
\widehat v_{N_n}\longrightarrow u
\qquad\text{a.s.}
\end{equation}

\medskip
\noindent
\textit{Step 4: Uniform convergence on compact sets.}

Let \(K\subset\mathbb R\) be compact.  From
\eqref{eq:full-cross-moment-limit},
\begin{align*}
&\sup_{\theta\in K}
\left|
\widehat v_{N_n}
+
(\bar X_{N_n}-\theta)^2
-
\left\{
u+(a-\theta)^2
\right\}
\right|\\
&\qquad\le
|\widehat v_{N_n}-u|
+
|\bar X_{N_n}-a|
\sup_{\theta\in K}
|\bar X_{N_n}+a-2\theta|
\longrightarrow0
\end{align*}
almost surely.  Moreover, since \(u>0\),
\[
\inf_{\theta\in K}
\left\{
\widehat v_{N_n}
+
(\bar X_{N_n}-\theta)^2
\right\}
\ge
\widehat v_{N_n}
\ge
\frac u2
\]
for all sufficiently large \(n\), almost surely.  The logarithm is
therefore uniformly Lipschitz on the relevant range.  Combining this
with \eqref{eq:UI-cross-exact-decomp} and
\eqref{eq:Cn-cross-limit} yields
\begin{align*}
\sup_{\theta\in K}
\left|
\frac1{N_n}\log R_{1:N_n}^{\theta}
-
\left[
-\frac12\log u
+
\frac12\log\{u+(a-\theta)^2\}
\right]
\right|
\longrightarrow0
\end{align*}
almost surely.  Since
\[
-\frac12\log u
+
\frac12\log\{u+(a-\theta)^2\}
=
\frac12\log\left(
1+\frac{(a-\theta)^2}{u}
\right)
=
J_{a,u}(\theta),
\]
this proves \eqref{eq:UI-cross-compact}.

\medskip
\noindent
\textit{Step 5: Uniform convergence over the exterior set.}

Because the logarithm is increasing and \(C_{N_n}\) does not depend
on \(\theta\),
\begin{align*}
\inf_{|\theta-a|\ge\delta}
\frac1{N_n}\log R_{1:N_n}^{\theta}
&=
C_{N_n}
+
\frac12\log\left\{
\widehat v_{N_n}
+
\inf_{|\theta-a|\ge\delta}
(\bar X_{N_n}-\theta)^2
\right\}.
\end{align*}
The distance from \(\bar X_{N_n}\) to the closed set
\(\{\theta:|\theta-a|\ge\delta\}\) is
\[
\left\{
\delta-|\bar X_{N_n}-a|
\right\}_+,
\]
so
\[
\inf_{|\theta-a|\ge\delta}
(\bar X_{N_n}-\theta)^2
=
\left\{
\delta-|\bar X_{N_n}-a|
\right\}_+^2
\longrightarrow\delta^2
\qquad\text{a.s.}
\]
Together with \eqref{eq:Cn-cross-limit} and
\eqref{eq:full-cross-moment-limit}, this gives
\[
\inf_{|\theta-a|\ge\delta}
\frac1{N_n}\log R_{1:N_n}^{\theta}
\longrightarrow
-\frac12\log u
+
\frac12\log(u+\delta^2)
=
\frac12\log\left(1+\frac{\delta^2}{u}\right)
\]
almost surely under the coupling, and hence in probability under the
original triangular-array laws.  This proves
\eqref{eq:UI-cross-outside}.
\end{proof}

\begin{lemma}[Concentration of empirical transition frequencies]
\label{lem:markov-transition-concentration}
For \(r=(p,q)\in\Theta_\kappa\), let \(K_r\) be its transition
matrix, let \(\varpi_r\) be its stationary distribution, and define
the stationary transition frequencies
\[
    \Gamma_r(x,y)
    :=
    \varpi_r(x)K_r(x,y),
    \qquad x,y\in\{0,1\}.
\]
For a block of \(n\) transitions, write
\[
    \widehat\Gamma_n(x,y)
    :=
    \frac{N_{xy}^{1:n}}{n}.
\]
There exist constants \(C_\kappa,c_\kappa>0\), depending only on
\(\kappa\), such that, uniformly over \(r\in\Theta_\kappa\), the
initial state, \(n\ge1\), and \(z\in(0,1)\),
\[
    \mathbb P_r
    \left(
        \|\widehat\Gamma_n-\Gamma_r\|_\infty>z
    \right)
    \le
    C_\kappa e^{-c_\kappa nz^2}.
\]
\end{lemma}

\begin{proof}
Fix \(x,y\in\{0,1\}\), and put
\[
    f(u,v):=\mathbbm 1\{u=x,v=y\},
    \qquad
    h(u):=\mathbb E_r[f(X_{i-1},X_i)\mid X_{i-1}=u].
\]
Then
\[
    \varpi_r h=\Gamma_r(x,y).
\]
The Dobrushin contraction coefficient of \(K_r\) is
\[
    \delta(K_r)=|p-q|\le1-2\kappa.
\]
Consequently, the Poisson series
\[
    g
    :=
    \sum_{m=0}^\infty
    \left(
        K_r^m h-\varpi_r h
    \right)
\]
converges and satisfies
\[
    \|g\|_\infty\le\frac{1}{2\kappa}.
\]
Moreover,
\[
    g-K_rg=h-\varpi_rh.
\]

Writing
\[
    \xi_i
    :=
    f(X_{i-1},X_i)-h(X_{i-1})
    +g(X_i)-K_rg(X_{i-1}),
\]
we have
\[
    \mathbb E_r[\xi_i\mid\mathcal F_{i-1}]=0
\]
and
\[
    |\xi_i|
    \le
    1+2\|g\|_\infty
    \le
    1+\frac1\kappa.
\]
The Poisson decomposition gives
\[
    \sum_{i=1}^n
    \{f(X_{i-1},X_i)-\Gamma_r(x,y)\}
    =
    \sum_{i=1}^n\xi_i+g(X_0)-g(X_n).
\]
Azuma--Hoeffding therefore yields constants \(C_\kappa,c_\kappa>0\)
such that
\[
    \mathbb P_r
    \left(
        \left|
            \widehat\Gamma_n(x,y)-\Gamma_r(x,y)
        \right|>z
    \right)
    \le
    C_\kappa e^{-c_\kappa nz^2}.
\]
A union bound over the four pairs \((x,y)\) proves the result.
\end{proof}

\end{document}